\documentclass[11pt,letterpaper,english,letter, 11pt]{article}
\usepackage[T1]{fontenc}
\usepackage[latin9]{inputenc}
\usepackage{color}
\usepackage{float}
\usepackage{calc}
\usepackage{amsmath}
\usepackage{amssymb}
\usepackage{graphicx}
\usepackage{esint}

\makeatletter

\newcommand{\noun}[1]{\textsc{#1}}

\usepackage{graphics}\usepackage{amsthm}\usepackage{amsopn}\usepackage{amstext}\usepackage{amsfonts}\@ifundefined{definecolor}
 {\usepackage{color}}{}
\usepackage{fullpage}\usepackage{hyperref}

\newcounter{algo}\newenvironment{algo}[2]{\refstepcounter{algo}\label{#2}   \begin{center}
\begin{minipage}{0.9\textwidth}   \hrule\smallskip
\textbf{Algorithm \thealgo: #1}
\par\smallskip\hrule\smallskip\ignorespaces}{\par\smallskip\hrule
\end{minipage}
\end{center}}

\newcommand{\beq}{\begin{equation}}
\newcommand{\eeq}{\end{equation}}

\newcommand{\epc}{\hspace{1pc}}

\newcommand{\wh}{\widehat}

\newtheorem{theorem}{Theorem}
\newtheorem{proposition}[theorem]{Proposition}
\newtheorem{lemma}[theorem]{Lemma}
\newtheorem{corollary}[theorem]{Corollary}

\makeatother

\usepackage{babel}
\begin{document}

\title{{\Huge {\vspace{-2cm}}}\textbf{\Huge Joint Sensing and Power Allocation
in Nonconvex Cognitive Radio Games: Nash Equilibria and Distributed
Algorithms}}

\author{Gesualdo Scutari$^{1}$ and Jong-Shi Pang$^{2}$ \\
 {\small $^{1}$Department of Electrical Engineering, State University
of New York at Buffalo, Buffalo, NY 14051, U.S.A.}\\
 $^{2}${\small Department of Industrial and Enterprise Systems Engineering,
University of Illinois, Urbana, IL 61801, U.S.A. }\\
 Emails:\texttt{ gesualdo@buffalo.edu, jspang@illinois.edu}.%
\thanks{The work of Pang is based on research supported by the U.S.A. National
Science Foundation grant CMMI 0969600 and by the Air Force Office
of Sponsored Research award No. FA9550-09-10329. The work of Scutari
was supported by U.S.A. National Science Foundation grant CSM 1218717.{\small }\protect \\
Copyright (c) 2012 IEEE. Personal use of this material is permitted.
However, permission to use this material for any other purposes must
be obtained from the IEEE by sending a request to pubs-permissions@ieee.org.%
}}

\date{{\normalsize Submitted to }\emph{\normalsize IEEE Transactions on}{\normalsize{}
}\emph{\normalsize Information Theory}{\normalsize , March 27, 2011.
Revised August 8, 2012. \vspace{-0.5cm}}}
\maketitle
\begin{abstract}
\noindent In this paper, we propose a novel class of Nash problems
for Cognitive Radio (CR) networks, modeled as Gaussian frequency-selective
interference channels, wherein each secondary user (SU) competes against
the others to maximize his own opportunistic throughput by choosing
\emph{jointly} the sensing duration, the detection thresholds, \emph{and}
the vector power allocation. The proposed general formulation allows
to accommodate several (transmit) power and (deterministic/probabilistic)
interference constraints, such as constraints on the maximum individual
and/or aggregate (probabilistic) interference tolerable at the primary
receivers. To keep the optimization as decentralized as possible,
global (coupling) interference constraints are imposed by penalizing
each SU with a set of time-varying prices based upon his contribution
to the total interference; the prices are thus additional variable
to optimize. The resulting players' optimization problems are \emph{nonconvex};
moreover, there are possibly price clearing conditions associated
with the global constraints to be satisfied by the solution. All this
makes the analysis of the proposed games a challenging task; none
of classical results in the game theory literature can be successfully
applied. 

\noindent The main contribution of this paper is to develop a novel
optimization-based theory for studying the proposed nonconvex games;
we provide a comprehensive analysis of the existence and uniqueness
of a standard Nash equilibrium, devise alternative best-response based
algorithms, and establish their convergence. Some of the proposed
algorithms are totally distributed and asynchronous, whereas some
others require limited signaling among the SUs (in the form of consensus
algorithms) in favor of better performance; overall, they are thus
applicable to a variety of CR scenarios, either cooperative or noncooperative,
which allows the SUs to explore the existing trade-off between signaling
and performance. 
\end{abstract}

\section{Introduction {\normalsize \vspace{-0.2cm}}}

Over the past decade, there has been a growing interest in Cognitive
Radio (CR) as an emerging paradigm to address the \emph{de jure} shortage
of allocated spectrum that contrasts with the \emph{de facto} abundance
of unused spectrum in virtually any spatial location at almost any
given time. The paradigm posits that so-called cognitive radios {[}also
termed as secondary users (SUs){]} would use licensed spectrum in
an ad-hoc fashion in such a way as to cause no harmful interference
to the primary spectrum license holders {[}also termed as primary
users (PUs){]}. Evidently, such an opportunistic spectrum access is
intertwined with the design of multiple secondary system components,
such as (but not limited to) spectrum sensing and transmission parameters
adaptation. Indeed, the choice of the sensing parameters (e.g., the
detection thresholds and the sensing duration) as well as the consequent
design of the physical layer transmission strategies (e.g., the transmission
rate, the power allocation) have both a direct impact on the performance
of primary and secondary systems. The interplay between these two
interacting components calls for a\emph{ joint optimization of the
sensing and transmission parameters} of the SUs, which is the main
focus of this paper.

\subsection{Motivation and related work}

The joint optimization of the sensing and transmission strategies
has been only partially addressed in the literature, even for simple
CR scenarios composed of one PU and one SU. For example, in \cite{Quan-Cui-Poor-Sayed,Quan-Cui-Poor-Sayed_TSP09},
the authors proposed alternative centralized schemes that optimize
the detection thresholds for a bank of energy detectors, in order
to maximize the\textit{\emph{ opportunistic throughput}}\textit{ }\textit{\emph{of
a SU}}\textit{, }\textit{\emph{for a given sensing time and constant-rate/power
transmissions}}. The optimization of the sensing time and the sensing
time/detection thresholds for a given missed detection probability
and constant rate of one SU was addressed in \cite{Liang-Zeng-Peh-Hoang_TWC08,Paysarvi-Beaulieu_IEEETSP11}
and \cite{Rongfei-Hai_TWC10}, respectively. A throughput-sensing
trade-off for a fixed transmission rate was studied in \cite{Kim-Giannakis08}.
In \cite{Barbarossa-Sardellitti-Scutari_CAMSAP09} (or \cite{Pei-Liang-Teh-Li_TWC09})
the authors focused on the joint optimization of the power allocation
and the equi-false alarm rate (or the sensing time) of a SU over multi-channel
links, for a fixed sensing time (or detection probability). All the
aforementioned schemes however are not applicable to scenarios composed
of multiple SUs (and PUs). The case of multiple SUs and one PU was
considered in \cite{Barbarossa-Sardellitti-Scutari_Cogis09} (and
more recently in \cite{Huang-Lozano_ICASSP11}), under the same assumptions
of \cite{Barbarossa-Sardellitti-Scutari_CAMSAP09}; however no formal
analysis of the proposed formulation was provided.

The transceiver design of OFDM-based CR systems composed of \emph{multiple}
primary and secondary users have been largely studied in the literature
of power control problems over the interference channel, and have
been traditionally approached from two very different perspectives:
a holistic design of the system and an individual selfish design of
each of the users. The former is also referred to as Network Utility
Maximization (NUM) (other approaches within this perspective are based
on Nash bargaining formulations) and has the potential of obtaining
the best of the network at the expense of a centralized computation
or heavy signaling/cooperation among the users; examples are \cite{Xing-Chetan-Mathur-Haleem-Chandramouli-Subbalakshmi_TranMoBComp07_IntTempLim,Lu-WWang-TWang-Peng_GLOB08,WangPengWang_WCNC07,SchmidtShiBerryHonigUtschick-SPMag,Luo-Zhang,YuLui_TCOM06centr,Rossi-Tulino-Simeone-Haimovich_ICASSP11}.
The latter fits perfectly within the mathematical framework of Game
Theory and usually leads to distributed algorithms at the expense
of a loss of global performance; related papers are \cite{Luo-Pang_IWFA-Eurasip,Scutari-Palomar-Barbarossa_SP08_PI,Scutari-Palomar-Barbarossa_AIWFA_IT08,Pang-Scutari-Palomar-Facchinei_SP_10,ScutariPalomarFacchineiPang-Monotone_bookCh,Mohsenian-Rad-Huang-Chiang-Wong_TW09},
and two recent overviews are \cite{Leshem-Zehavip_SPMag_sub09,Larsson-Jorswieck-Lindblom-Mochaourab_SPMag_sub09}.
In both the aforementioned approaches and classes of papers the \emph{sensing
process is not considered as part of the optimization}; in fact the
SUs do not perform any sensing but they are allowed to transmit over
the licensed spectrum provided that they satisfy interference constraints
imposed by the PUs, no matter if the PUs are active of not.

When the sensing comes explicitly into the system design, the application
of the holistic approach mentioned above leads to nonconvex NP hard
optimization problems. These cases cannot be globally solved by efficient
algorithms in polynomial time; one typically can design (centralized)
sub-optimal algorithms that converge just to a stationary solution.
Their implementation however would require heavy signaling among the
users (or the presence of a centralized network controller having
the knowledge of all the system parameters); which strongly limits
the range of applicability of such formulations to practical CR networks.
For these reasons, in this paper, we attack the multi-agent decision
making problem from a different perspective; we concentrate on optimization
strategies where the SUs are able to self-enforce the negotiated agreements
on the usage of the licensed spectrum either in a totally decentralized
way or by requiring limited and local signaling among the SUs (in
the form of consensus algorithms). 
Aiming at exploring the trade-off between signaling and performance,
the proposed approach is then expected to be more flexible than classical
optimization techniques and applicable to a wider range of CR scenarios.

\subsection{Main contributions }

This paper along with our companion work \cite{Pang-Scutari-NNConvex_PI}
advances the current approaches (based on the optimization of specific
components of a CR system in isolation), in the direction of a\emph{
joint }and distributed design of sensing and transmission parameters
of a CR network, composed of\emph{ multiple} PUs and SUs. 

We study a novel class of Nash equilibrium problems as proposed in
\cite{Pang-Scutari-NNConvex_PI}, wherein each SU aims at maximizing
his own opportunistic throughput by \emph{jointly} optimizing the
sensing parameters$-$the sensing time and the false alarm rate (and
thus the decision thresholds) of a bank of energy detectors$-$and
the power allocation over the multi-channel links. Because of sensing
errors, the SUs might access the licensed spectrum when it is still
occupied by active PUs, thus causing harmful interference. This motivates
the introduction of \emph{probabilistic} interference constraints
that are imposed to control the power radiated over the licensed spectrum
\emph{whenever a missed detection event occurs} (in a probabilistic
sense). The proposed formulation accommodates alternative combinations
of power/interference constraints. For instance, on top of classical
(deterministic) transmit power (and possibly spectral masks) constraints,
we envisage the use of average \emph{individual} (i.e., on each SU)
and/or\emph{ global} (i.e., over all the SUs) interference tolerable
at the primary receivers. The former class of constraints is more
suitable for scenarios where the SUs are not willing to cooperate;
whereas the latter constraints, which are less conservative, seem
more realistic in settings where SUs may want to trade some limited
signaling for better performance. By imposing a coupling among the
transmit and sensing strategies of the SUs, global interference constraints
introduce a new challenge in the system design: how to enforce global
interference constraints without requiring a centralized optimization
but possibly only limited signaling among the SUs? We address this
issue by introducing a pricing mechanism in the game, through a penalization
in the players\textquoteright{} objective functions. The prices need
to be chosen so that the interference constraints are satisfied at
any solution of the game and a clearing condition holds; they are
thus additional variables to be determined.

The resulting class of games is nonconvex (because of the nonconvexity
of the players' payoff functions and constraints), lacks boundedness
in the price variables, and there are side constraints with associated
price equilibration that are required to be satisfied by the equilibrium;
all these features make the analysis a challenging task. The convexity
of the players' individual optimization problems is, in fact, one
indispensable assumption under which noncooperative games have traditionally
been studied and analyzed. The classical case where a NE exists is
indeed when the players' objective functions are (quasi-)convex in
their own variables with the other players' strategies fixed, and
the players' constraint sets are compact and convex and independent
of their rivals' strategies (see, e.g., \cite{Rosen_econ65,Osborne-Rubinstein_GT_book}).
Without such convexity, a NE may not exist (as in the well-known case
of a matrix game with pure strategies); analytically, abstract mathematical
theories granting its existence, like those in \cite{Baye-Tian-Zhou_RES93,CornetaCzarnecki01},
are difficult to be applied to games arising from realistic applications
such as those occurred in the present paper.

The main contribution of this work is to develop a novel optimization-based
theory for the solution analysis of the proposed class of nonconvex
games (possibly) with side constraints and price clearing conditions,
and to design distributed best-response based algorithms for computing
the Nash equilibria, along with their convergence properties. Building
on \cite{PScutari10}, the solution analysis is addressed by introducing
a ``best-response'' map (including price variables) defined on a
proper \emph{convex} and \emph{compact} set, whose fixed-points, if
they exist, are Nash equilibria of the original nonconvex games; the
obtained conditions are in fact sufficient for such a map to be a
\emph{single-valued continuous} map; this enables the application
of the Brouwer fixed-point theorem to deduce the existence of a fixed-point
of the best-response map, thus of a NE of the whole class of proposed
games. While seemingly very simple, the technical details lie in deriving
(reasonable) conditions for which the best-response map is single-valued
and for the boundedness of the prices in order for the existence of
a compact set on which the Brouwer result can be based. Interestingly,
the obtained conditions have the same physical interpretation of those
obtained for the convergence of the renowned iterative waterfilling
algorithm solving the power control game over interference channels
\cite{Luo-Pang_IWFA-Eurasip,Scutari-Palomar-Barbarossa_SP08_PI,Scutari-Palomar-Barbarossa_AIWFA_IT08,Pang-Scutari-Palomar-Facchinei_SP_10,ScutariPalomarFacchineiPang-Monotone_bookCh}.
We then focus on solutions schemes for the proposed class of games;
we design  alternative distributed (possibly) asynchronous best-response
based algorithms that differ in performance, level of protection of
the PUs, computational effort and degree of cooperation/signaling
among the SUs, and convergence speed; which makes them applicable
to a variety of CR scenarios (either cooperative or noncooperative).
For each algorithm, we establish its convergence and also quantify
the time and communication costs for its implementation. Our numerical
results show that: i) the proposed joint sensing/transmission optimization
outperforms current \emph{centralized and decentralized} state-of-the-art
results based on separated optimization of the sensing and the transmission
parts; ii) our algorithms exhibit a fast convergence behavior; and
iii) as expected, some (limited) cooperation among the SUs (in the
form of consensus algorithms) yields a significant improvement in
the system performance. The proposed solution schemes can also be
used to compute the so-called Quasi-NE of the associated games, a
relaxed equilibrium concept introduced and studied in our companion
paper \cite{Pang-Scutari-NNConvex_PI}.

The paper is organized as follows. Sec. \ref{sec:System-Model} briefly
introduces the system model, as proposed in \cite{Pang-Scutari-NNConvex_PI};
Sec. \ref{sec:Game-Theoretical-Formulation} focuses on the system
design and formulates the joint optimization of the sensing parameters
and the power allocation of the SUs within the framework of game theory;
several games are introduced. The solution analysis of the proposed
games is addressed in Sec. \ref{sec:Solution-Analysis:-Nash}, where
sufficient conditions for the existence and uniqueness of a standard
NE along with their interpretation are derived. Distributed algorithms
solving the proposed games along with their convergence properties
and computational/communication complexity are studied in Sec. \ref{sub:Game_with_fixed_prices}.
Numerical experiments are reported in Sec. \ref{sec:Numerical-Results},
whereas Sec. \ref{sec:Conclusions} draws the conclusions. Proofs
of our results are given in Appendix \ref{app:Proof-of-Proposition_uniqueness_NE}-\ref{app:Convergence-of-Best-Response}.
The paper requires a background on Variational Inequalities (VIs);
we refer to \cite{Scutari-Palomar-Facchinei-Pang_SPMag10,ScutariFacchineiPangPalomar-MonotoneIT}
for an introductory overview of the subject and its application to
equilibrium problems in signal processing and communications. A comprehensive
treatment of VIs can be found in the two monographs \cite{Cottle-Pang-Stone_bookLCP92,Facchinei-Pang_FVI03};
a detailed study of convex games based on the VI and complementarity
approach is addressed in \cite{Facchinei_Pang_VI-NE_bookCh_09,ScutariPalomarFacchineiPang-Monotone_bookCh}.
The main properties of Z and P matrices, which are widely used in
the paper, can be found in \cite{Cottle-Pang-Stone_bookLCP92,Berman-Plemmons_bookNonNegMat87}.\vspace{-0.2cm}

\section{System Model \label{sec:System-Model}{\normalsize \vspace{-0.3cm}}}

We consider a scenario composed of $Q$ active SUs, each consisting
of a transmitter-receiver pair, coexisting in the same area and sharing
the same band with PUs. The network of the SUs is modeled as an $N$-frequency-selective
SISO Interference Channel (IC), where $N$ is the number of subcarriers
available to the cognitive users. We focus on multicarrier block-transmissions
without loss of generality. In order not to interfere with on-going
PU transmissions, before transmitting, the SUs sense periodically
the licensed spectrum looking for the subcarriers that are temporarily
not occupied by the PUs. A brief description of the sensing mechanism
and transmission phase performed by the SUs as proposed in the companion
paper \cite{Pang-Scutari-NNConvex_PI} is given in the following,
where we introduce the basic definitions and notation used throughout
the paper; we refer the reader to \cite{Pang-Scutari-NNConvex_PI}
for details and the assumptions underlying the proposed model.

\subsection{The spectrum sensing phase \label{sub:The-spectrum-sensing}}

In  \cite{Pang-Scutari-NNConvex_PI}, we formulated the sensing problem
as a binary hypothesis testing; the decision rule of SU $q$ over
carrier $k=1,\ldots,N$ based on the energy detector is 
\begin{equation}
D_{q,k}\triangleq\dfrac{{1}}{K_{q}}\sum_{n=1}^{K_{q}}\left|y_{q,k}[n]\right|^{2}\begin{array}{c}
\overset{\mathcal{H}_{1,k}}{>}\vspace{-0.4cm}\\
\underset{\mathcal{H}_{0,k}}{<}
\end{array}\gamma_{q,k}\label{eq:energy_detector_test}
\end{equation}
where $y_{q,k}[n]$ is the received baseband complex signal over carrier
$k$; $K_{q}=\left\lfloor \tau_{q}\, f_{q}\right\rfloor \backsimeq\tau_{q}\, f_{q}$
is the number of samples, with $\tau_{q}$ and $f_{q}$ denoting the
sensing time and the sampling frequency, respectively; $\gamma_{q,k}$
is the decision threshold for the carrier $k$; $\mathcal{H}_{0,k}$
represents the absence of any primary signal over the subcarrier $k$,
whereas $\mathcal{H}_{1,k}$ represents the presence of the primary
signaling. 

The performance of the energy detection performed by SU $q$ over
carrier $k$ is measured in terms of the detection probability $P_{q,k}^{\text{\,{d}}}(\gamma_{q,k},\tau_{q})\triangleq\text{{Prob}}\left\{ D_{q,k}>\gamma_{q,k}\,|\,\mathcal{H}_{1,k}\right\} $
and false alarm probability $P_{q,k}^{\,\text{{fa}}}(\gamma_{q,k},\tau_{q})\triangleq\text{{Prob}}\{D_{q,k}>\gamma_{q,k}\,|\,\mathcal{H}_{0,k}\}$.
Under standard assumptions in decision theory, these probabilities
are given by \cite{Pang-Scutari-NNConvex_PI}\vspace{-0.4cm}

\textcolor{black}{
\begin{equation}
P_{q,k}^{\,\text{{fa}}}\left(\gamma_{q,k},\,\tau_{q}\right)=\mathcal{Q}\left(\sqrt{\tau_{q}\, f_{q}}\,\dfrac{{\gamma_{q,k}\,-\mu_{q,k|0}}}{{\sigma_{q,k|0}}}\right)\quad\mbox{and}\quad P_{q,k}^{\text{\,{d}}}\left(\gamma_{q,k},\,\tau_{q}\right)=\mathcal{Q}\left(\sqrt{\tau_{q}\, f_{q}}\,\dfrac{{\gamma_{q,k}\,-\mu_{q,k|1}}}{{\sigma_{q,k|1}}}\right),\label{eq:Pfa-worst_case}
\end{equation}
}where $\mathcal{\mathcal{Q}}(x)\triangleq(1/\sqrt{{2\pi}})\int_{x}^{\infty}e^{-t^{2}/2}dt$
is the Q-function, and $\mu_{q,k|0}$, $\mu_{q,k|1}$, ${\sigma_{q,k|0}}$,
and ${\sigma_{q,k|1}}$ are constant parameters, whose explicit expressions
are given in \cite{Pang-Scutari-NNConvex_PI}. The detection probability
$P_{q,k}^{\text{\,{d}}}$ can also be rewritten as a function of the
false alarm rate $P_{q,k}^{\,\text{{fa}}}$ as: 
\begin{equation}
P_{q,k}^{\text{\,{d}}}\left(P_{q,k}^{\,\text{{fa}}},\,\tau_{q}\right)=\mathcal{Q}\left(\dfrac{{\sigma_{q,k|0}}\,}{{\sigma_{q,k|1}}}\mathcal{Q}^{-1}\left(P_{q,k}^{\,\text{{fa}}}\right)-\sqrt{\tau_{q}\, f_{q}}\,\dfrac{{\mu_{q,k|1}-\mu_{q,k|0}}}{{\sigma_{q,k|1}}}\right)\triangleq1-P_{q,k}^{\text{{miss}}}\left({\tau}_{q},\, P_{q,k}^{\,\text{{fa}}}\right),\label{eq:Pd-worst_case_2}
\end{equation}
where we also introduced the definition of the missed detection probability
$P_{q,k}^{\text{{miss}}}({\tau}_{q},P_{q,k}^{\,\text{{fa}}})\triangleq1-P_{q,k}^{\text{{d}}}({\tau}_{q},P_{q,k}^{\,\text{{fa}}})$.

The interpretation of $P_{q,k}^{\,\text{{fa}}}\left(\gamma_{q,k},\,\tau_{q}\right)$
and $P_{q,k}^{\text{\,{d}}}\left(\gamma_{q,k},\,\tau_{q}\right)$
within the CR scenario is the following: $1-P_{q,k}^{\,\text{{fa}}}$
signifies the probability of successfully identifying from the SU
$q$ a spectral hole over carrier $k$, whereas the missed detection
probability $P_{q,k}^{\text{\,{miss}}}$ represents the probability
of SU $q$ failing to detect the presence of the PUs on the subchannel
$k$ and thus generating interference against the PUs. The free variables
to optimize are the detection thresholds $\gamma_{q,k}$'s and the
sensing times $\tau_{q}$'s; ideally, we would like to choose $\gamma_{q,k}$'s
and $\tau_{q}$'s in order to minimize both $P_{q,k}^{\,\text{{fa}}}$
and $P_{\text{{miss}}}^{(q,k)}$, but (\ref{eq:Pd-worst_case_2})
shows that there exists a trade-off between these two quantities that
will affect both primary and secondary performance. It turns out that,
$\gamma_{q,k}$'s and $\tau_{q}$'s can not be chosen by focusing
only on the detection problem (as in classical decision theory), but
the optimal choice of $\gamma_{q,k}$ and $\tau_{q}$ must be the
result of a \emph{joint} optimization of the sensing and transmission
strategies over the two phases; such an optimization is introduced
in Sec. \ref{sec:Game-Theoretical-Formulation}. \medskip{}

\noindent\textbf{Robust sensing model.} The proposed sensing model
can be generalized in several directions; see \cite{Scutari-Pang_DSP11,Pang-Scutari-NNConvex_PI}.
For instance, one can explicitly take into account device-level uncertainties
(e.g., uncertainty in the power spectral density of the PUs' signals
and thermal noise) as well as system level uncertainties (e.g., the
current number of active PUs) by modeling the detection process of
the primary signals as a \emph{composite} hypothesis testing. This
leads to a \emph{uniformly most-powerful} detector scheme that is
robust against device-level and system-level uncertainties; detailed
can be found in \cite{Scutari-Pang_DSP11,Pang-Scutari-NNConvex_PI}
and are omitted here. It is important however to remark that the resulting
detection probability and false alarm rate of the aforementioned robust
scheme are still given by (\ref{eq:Pfa-worst_case}) and (\ref{eq:Pd-worst_case_2}),
but with a different expression for $\mu_{q,k|i}$'s and $\sigma_{q,k|i}^{2}$'s
\cite{Scutari-Pang_DSP11}. This means that analysis and results developed
in the next sections are valid also for this more general model.

\subsection{The transmission phase }

The transmission strategy of each SU $q$ is the power allocation
vector $\mathbf{p}_{q}=\{p_{q,k}\}_{k=1}^{N}$ over the $N$ subcarriers,
subject to the following (local) transmit power constraints
\begin{equation}
{\mathcal{P}}_{q}\triangleq\left\{ \mathbf{p}_{q}\triangleq(p_{q,k})_{k=1}^{N}\in\mathbb{R}^{N}\,:\,\sum_{k=1}^{N}p_{q,k}\leq P_{q},\quad\mathbf{0}\leq\mathbf{p}_{q}\leq\mathbf{p}_{q}^{\max}\right\} ,\label{set_P_q}
\end{equation}
where $\mathbf{p}_{q}^{\max}=(p_{q,k}^{\max})_{k=1}^{N}$ denotes
possibly spectral mask {[}the vector inequality in (\ref{set_P_q})
is component-wise{]}.

According to the opportunistic transmission paradigm, each subcarrier
$k$ is available for the transmission of SU $q$ if no primary signal
is detected over that frequency band, which happens with probability
$1-P_{q,k}^{\,\text{{fa}}}$. This motivates the use of the \emph{aggregate
opportunistic throughput }as a measure of the spectrum efficiency
of each SU $q$. Given the power allocation profile $\mathbf{p}=(\mathbf{p}_{q})_{q=1}^{Q}$
of the SUs, the target false alarm rate $P_{q}^{\,\text{{fa}}}$ (assumed
to be equal over the whole licensed spectrum), the sensing time ${\tau}_{q}$,
and taking the log of the opportunistic throughput, the payoff function
of each SU $q$ is then (see \cite{Pang-Scutari-NNConvex_PI} for
more details) 
\begin{equation}
R_{q}\left(\tau_{q},\,\mathbf{p},\, P_{q}^{\,\text{{fa}}}\right)=\log\left(\left(1-\dfrac{\tau_{q}}{T_{q}}\right)\,\left(1-P_{q}^{\,\text{{fa}}}\right)\,\sum_{k=1}^{N}\, r_{q,k}\left(\mathbf{p}\right)\right)\label{eq:Opportunistic-throughtput}
\end{equation}
where $1-{\tau}_{q}/T_{q}$, with ${\tau}_{q}\leq T_{q}$, is the
portion of the frame duration $T_{q}$ available for opportunistic
transmissions and $r_{q,k}(\mathbf{p})$ is the maximum information
rate achievable on link $q$ over carrier $k$ \emph{when no primary
signal is detected} and the power allocation profile of the SUs is
$p_{1,k},\ldots,p_{Q,k}$: 
\begin{equation}
r_{q,k}(\mathbf{p})=\log\left(1+\dfrac{p_{q,k}}{\hat{{\sigma}}_{q,k}^{2}+\sum_{r\neq q}|\hat{{H}}{}_{qr}(k)|^{2}p_{r,k}}\right),\label{eq:rate_FSIC}
\end{equation}
with $\hat{{H}}{}_{qr}(k)\triangleq H{}_{qr}(k)/H{}_{qq}(k)$ and
$\hat{{\sigma}}_{q,k}^{2}\triangleq\sigma_{q,k}^{2}/|H{}_{qq}(k)|^{2}$,
where $\{H_{qq}(k)\}_{k=1}^{N}$ is the channel transfer function
of the direct link $q$ and $\{H_{qr}(k)\}_{k=1}^{N}$ is the cross-channel
transfer function between the secondary transmitter $r$ and the secondary
receiver $q$; and $\sigma_{q,k}^{2}$ is the power spectral density
(PSD) of the background noise over carrier $k$ at the receiver $q$
(assumed to be Gaussian zero-mean distributed). 

As a final remark note that the throughput defined in (\ref{eq:Opportunistic-throughtput})
is not the average throughput experienced by the SUs, which instead
would include an additional rate contribution resulting from the erroneous
decision of the SUs to transmit over the licensed spectrum still occupied
by the PUs. We have not included this contribution in the objective
functions of the SUs because in maximizing the function we do not
want to ``incentivize'' the undue usage of the licensed spectrum.
Moreover, differently from the opportunistic throughput in (\ref{eq:Opportunistic-throughtput}),
the maximization of the average throughput would require the knowledge
from the SUs of the a-priori probabilities of the PUs' spectrum occupancy,
which is in general not available.

\subsection{Probabilistic interference constraints }

Due to the inherent trade-off between $P_{q}^{\,\text{{fa}}}$ and
$P_{q,k}^{\text{{miss}}}(P_{fa}^{(q)})$ {[}see (\ref{eq:Pfa-worst_case})
and (\ref{eq:Pd-worst_case_2}){]}, maximizing the aggregate opportunistic
throughput (\ref{eq:Opportunistic-throughtput}) of SUs will result
in low $P_{q}^{\,\text{{fa}}}$ and thus large $P_{q,k}^{\text{{miss}}}$,
hence causing harmful interference to PUs. To allow the SUs' transmissions
while preserving the QoS of the PUs, we envisage the use of probabilistic
interference constraints that limit the interference generated by
the SUs whenever they misdetect the presence of a PU. Examples of
these constraints are the following: 
\begin{description}
\item [{-}] \emph{Individual overall bandwidth interference constraint}:
for each SU $q,$ 
\begin{equation}
\sum_{k=1}^{N}P_{q,k}^{\text{{miss}}}\left({\tau}_{q},\, P_{q}^{\,\text{{fa}}}\right)\cdot w_{q,k}\cdot p_{q,k}\leq I_{q}^{\text{{max}}},\vspace{-0.2cm}\label{eq:individual_overal_interference_constraint}
\end{equation}

\item [{\emph{-}}] \emph{Global overall bandwidth interference constraints}:
\begin{equation}
\sum_{q=1}^{Q}\sum_{k\in\mathcal{K}_{p}}P_{q,k}^{\text{{miss}}}\left({\tau}_{q},\, P_{q}^{\,\text{{fa}}}\right)\cdot w_{q,k}\cdot p_{q,k}\leq I^{\text{{max}}},\vspace{-0.2cm}\label{eq:global_interference_constraints_2}
\end{equation}

\end{description}
where $I_{q}^{\text{{max}}}$ {[}or $I^{\text{\ensuremath{\max}}}${]}
are the maximum average interference allowed to be generated by the
SU $q$ {[}or all the SU's{]} that is tolerable at the primary receiver;
and $w_{q,k}$'s are a given set of positive weights. If an estimate
of the cross-channel transfer functions $\{G_{P,q}(k)\}_{k=1}^{N}$
between the secondary transmitters and the primary receiver is available,
then the natural choice for $w_{q,k}$ is $w_{q,k}=|G_{P,q}(k)|^{2}$,
so that (\ref{eq:individual_overal_interference_constraint}) and
(\ref{eq:global_interference_constraints_2}) become the average interference
experienced at the primary receiver. Methods to obtain the interference
limits along with some implementation aspects related to this issue
and alternative interference constraints are discussed in Sec. \ref{Rmk_implementation_issues}. 

We wish to point out that other interference constraints, like per-carrier
interference constraints, as well as multiple PUs can be readily accommodated,
without affecting the analysis and results that will be presented
in the forthcoming sections. For notational simplicity, we stay within
the above setting.

\section{System Design based on Game Theory \label{sec:Game-Theoretical-Formulation}}

We focus now on the system design and formulate the joint optimization
of the sensing parameters and the power allocation of the SUs within
the framework of game theory. We consider next two classes of equilibrium
problems: i) games with \emph{ individual }constraints only (Sec.
\ref{sub:GT_local} below); and ii) games with \emph{individual and
global} constraints (Sec. \ref{sub:GT_local} and Sec. \ref{sub:GT_global_equi_sensing}
below). The former formulation is suitable for modeling scenarios
where the SUs are selfish users who are not willing to cooperate,
whereas the latter class of games is applicable to the design of systems
where the SUs can exchange limited signaling in favor of better performance.
Indeed, being less conservative than individual interference constraints,
global interference constraints are expected to yield better performance
of the SUs at the cost of more signaling. The aforementioned formulations
are thus applicable to complementary CR scenarios.

\subsection{Game with local interference constraints\label{sub:GT_local}}

In the proposed game, each SU is modeled as a player who aims to maximize
his own opportunistic throughput $R_{q}\left(\tau_{q},\,\mathbf{p},\, P_{q}^{\,\text{{fa}}}\right)$
by choosing \emph{jointly} a proper power allocation strategy $\mathbf{p}_{q}=(p_{q,k})_{k=1}^{N}$,
sensing time $\tau_{q}$, and false alarm rate $P_{q}^{\,\text{{fa}}}$,
subject to power and individual probabilistic interference constraints.
Stated in mathematical terms we have the following formulation. \emph{\vspace{.3cm}}

\framebox{\begin{minipage}[t]{0.93\columnwidth}%
\textbf{Player $q$'s optimization problem} is to determine, for given
$\mathbf{p}_{-q}\triangleq((p_{r}(k)_{k=1}^{N})_{q\neq r=1}^{Q}\geq\mathbf{0}$,
a tuple $\left({\tau}_{q},\,\mathbf{p}_{q},\, P_{q}^{\,\text{{fa}}}\right)$
in order to
\begin{equation}
\begin{array}{ll}
{\displaystyle {\operatornamewithlimits{\mbox{maximize}}_{\tau_{q},\mathbf{p}_{q},P_{q}^{\,\text{{fa}}}}}} & R_{q}\left(\tau_{q},\,\mathbf{p},\, P_{q}^{\,\text{{fa}}}\right)\vspace{-0.2cm}\\[0.25in]
\mbox{subject to} & \vspace{-.1cm}\\[5pt]
\mbox{{\bf (a)}} & \begin{array}{l}
{\displaystyle {\sum_{k=1}^{N}}\, P_{q,k}^{\text{{miss}}}(P_{q}^{\,\text{{fa}}},\tau_{q})\cdot w_{q,k}\cdot p_{q,k}\,\leq\, I_{q}^{\text{{max}}}},\\[0.3in]\end{array}\vspace{-1cm}\\[0.5in]
\mbox{{\bf (b)}} & \begin{array}{l}
P_{q}^{\,\text{{fa}}}\,\leq\,\beta_{q},\quad\mbox{and}\quad P_{q,k}^{\text{{miss}}}(P_{q}^{\,\text{{fa}}},\tau_{q})\,\leq\,\alpha_{q,k},\quad\forall k\,=\,1,\cdots,N,\\[0.15in]\end{array}\vspace{-1.2cm}\\[0.5in]
\mbox{{\bf (c)}} & \,\,{\displaystyle \mathbf{p}_{q}\in\mathcal{P}_{q}\epc\mbox{and}\epc\tau_{q}^{\min}\,\leq\,\tau_{q}\,\leq\,\tau_{q}^{\max}.}\vspace{.1cm}
\end{array}\label{eq:player q_individual_interference_constraints}
\end{equation}
\end{minipage}}\emph{\vspace{.3cm}}

In (\ref{eq:player q_individual_interference_constraints}) we also
included additional lower and upper bounds of $\tau_{q}$ satisfying
$0<\tau_{q}^{\min}<\tau_{q}^{\max}<T_{q}$ and upper bounds on detection
and missed detection probabilities $0<\alpha_{q,k}\leq1/2$ and $0<\beta_{q}\leq1/2$,
respectively. These bounds provide additional degrees of freedom to
limit the probability of interference to the PUs as well as to maintain
a certain level of opportunistic spectrum utilization from the SUs
{[}$1-P_{q}^{\,\text{{fa}}}\geq1-\beta_{q}${]}. Note that the constraints
$\alpha_{q,k}\leq1/2$ and $\beta_{q}\leq1/2$ do not represent a
real loss of generality, because practical CR systems are required
to satisfy even stronger constraints on false alarm and detection
probabilities; for instance, in the WRAN standard, $\alpha_{q,k}=\beta_{q,k}=0.1$.

\subsection{Game with global interference constraints\label{sub:GT_global}}

We add now global interference constraints to the game theoretical
formulation in (\ref{eq:player q_individual_interference_constraints}).
This introduces a new challenge: how to enforce global interference
constraints in a distributed way? By imposing a coupling among the
transmissions and the sensing strategies of all the SUs, global interference
constraints in principle would call for a centralized optimization.
To overcome this issue, we introduce a pricing mechanism in the game,
based on the relaxation of the coupling interference constraints as
penalty term in the SUs' objective functions, so that the interference
generated by all the SUs will depend on these prices. Prices are thus
addition variables to be optimized (there is one common price associated
with any of the global interference constraints); they must be chosen
so that any solution of the game will satisfy the global interference
constraints, which requires the introduction of additional constraints
on the prices, in the form of price clearance conditions. Denoting
by $\pi$ the price variable associated with the global interference
constraint (\ref{eq:global_interference_constraints_2}), we have
the following formulation.\vspace{-0.6cm}

\begin{center}
\framebox{\begin{minipage}[t]{1\columnwidth}%
\textbf{Player $q$'s optimization problem} is to determine, for given
$\mathbf{p}_{-q}\geq\mathbf{0}$ and $\pi$, a tuple $\left({\tau}_{q},\,\mathbf{p}_{q},\, P_{q}^{\,\text{{fa}}}\right)$
such that\vspace{-0.2cm}
\begin{equation}
\begin{array}{ll}
{\displaystyle {\operatornamewithlimits{\mbox{maximize}}_{\tau_{q},\mathbf{p}_{q},P_{q}^{\,\text{{fa}}}}}} & R_{q}\left(\tau_{q},\,\mathbf{p},\, P_{q}^{\,\text{{fa}}}\right)-\pi\cdot{\displaystyle {\displaystyle {\sum_{k=1}^{N}}\, P_{q,k}^{\text{{miss}}}(P_{q}^{\,\text{{fa}}},\tau_{q})\cdot w_{q,k}\cdot p_{q,k}\vspace{-0.8cm}}}\\[0.25in]
\mbox{subject to} & \mbox{\mbox{constraints (a), (b), (c) as in} (\ref{eq:player q_individual_interference_constraints}).}\\[5pt]
\end{array}\label{eq:paler_q_global}
\end{equation}
\textbf{Price equilibrium}: The price $\pi$ obeys the following complementarity
condition:\vspace{-0.2cm} 
\begin{equation}
0\,\leq\,\pi\,\perp\, I^{\text{{max}}}-{\sum_{k=1}^{N}}\,{\displaystyle {\sum_{q=1}^{Q}}{\displaystyle \, P_{q,k}^{\text{{miss}}}(P_{q}^{\,\text{{fa}}},\tau_{q})\cdot w_{q,k}\cdot p_{q,k}}\,\geq\,0}.\label{eq:price equilibrium_pricing_global}
\end{equation}
\end{minipage}}
\par\end{center}

In (\ref{eq:price equilibrium_pricing_global}), the compact notation
$0\leq a\perp b\geq0$ means $a\geq0$, $b\geq0$, and $a\,\cdot\, b=0$.
The price clearance conditions (\ref{eq:price equilibrium_pricing_global})
state that global interference constraints (\ref{eq:global_interference_constraints_2})
must be satisfied together with nonnegative price; in addition, they
imply that if the global interference constraint holds with strict
inequality then the price should be zero (no penalty is needed). Thus,
at any solution of the game, the optimal price is such that the global
interference constraint is satisfied.

\subsection{The equi-sensing case\label{sub:GT_global_equi_sensing}}

The decision model proposed in Sec. \ref{sub:The-spectrum-sensing}
is based on the assumption that the SUs are somehow able to distinguish
between primary and secondary signaling. This can be naturally accomplished
if there is a \emph{common} sensing time (still to optimize) during
which \emph{all} the SUs stay silent while sensing the spectrum. However,
the formulation (\ref{eq:paler_q_global}), in general, leads to different
optimal sensing times of the SUs, implying that some SU may start
transmitting while some others are still in the sensing phase. To
overcome this issue, several directions have been explored in the
companion paper \cite{Pang-Scutari-NNConvex_PI}, under the model
(\ref{eq:paler_q_global})-(\ref{eq:price equilibrium_pricing_global}).
Here we follow the approach of modifying the formulation in (\ref{eq:paler_q_global})
in order to ``force'' in a \emph{distributed way} the same \emph{optimal}
sensing time for all the SUs. Roughly speaking, the idea is to perturb
the payoff functions of the players by a penalty term that discourages
the players to deviate from equi-sensing strategies. Stated in mathematical
terms, we have the following formulation. \vspace{-0.6cm}

\begin{center}
\framebox{\begin{minipage}[t]{1\textwidth}%
\begin{flushleft}
\textbf{Player $q$'s optimization problem} is to determine, for given
$c\geq0,$ $\mathbf{p}_{-q}\geq\mathbf{0}$, $(\tau_{r})_{q\neq r=1}^{Q}\geq\mathbf{0}$
and $\pi\geq0$, a tuple $\left({\tau}_{q},\,\mathbf{p}_{q},\, P_{q}^{\,\text{{fa}}}\right)$
in order to 
\begin{equation}
\begin{array}{ll}
{\displaystyle {\operatornamewithlimits{\mbox{maximize}}_{\tau_{q},\mathbf{p}_{q},P_{fa}^{\, q}}}} & R_{q}\left(\tau_{q},\,\mathbf{p},\, P_{q}^{\,\text{{fa}}}\right)-\pi\cdot{\displaystyle {\displaystyle {\sum_{k=1}^{N}}\, P_{q,k}^{\text{{miss}}}(P_{q}^{\,\text{{fa}}},\tau_{q})\cdot w_{q,k}\cdot p_{q,k}\,-\,\dfrac{{c}}{2}\cdot\left({\tau_{q}}-\dfrac{{1}}{Q}\,{\displaystyle {\sum_{r=1}^{Q}}}\,{\tau_{r}}\right)^{2}\vspace{-0.8cm}}}\\[0.25in]
\mbox{subject to} & \mbox{\mbox{constraints (a), (b), (c) as in} (\ref{eq:player q_individual_interference_constraints})}.\\[5pt]
\end{array}\label{eq:player q_equi-sensing}
\end{equation}
\textbf{Price equilibrium}: The price $\pi$ obeys the complementarity
condition (\ref{eq:price equilibrium_pricing_global}).
\par\end{flushleft}%
\end{minipage}}
\par\end{center}

The third term in the objective function of each SU in (\ref{eq:player q_equi-sensing})
helps to induce the same optimal sensing time for all the SUs. Roughly
speaking, one expects that for sufficiently large $c$, the aforementioned
term will become the dominant term in the objective functions of the
SUs, leading thus to solutions of the game having sensing times that
differ from their average by any prescribed accuracy. This intuition
has been made formal in our companion paper \cite{Pang-Scutari-NNConvex_PI}
for stationary solutions of the game (\ref{eq:player q_equi-sensing}),
and it can be similarly extended to the Nash equilibria; we omit the
details because of space limitation. \vspace{-0.2cm}

\subsection{Unified formulation and summary of notation}

In this section, we introduce a compact and unified formulation of
the proposed games that simplifies their analysis. Let us start by
separating the convex constraints in the feasible set of the players
from the nonconvex ones. The interference constraints (a) in (\ref{eq:player q_individual_interference_constraints})
are bi-convex and thus not convex, whereas constraints (b) are convex
in $P_{q}^{\,\text{{fa}}}$ and $\sqrt{{\tau_{q}}}$. This motivates
the following change of variables:
\begin{equation}
\tau_{q}\mapsto\wh{\tau}_{q}\triangleq\sqrt{\tau_{q}\, f_{q}}\,\quad q=1,\ldots,Q,\label{eq:bijection}
\end{equation}
so that the constraints on $P_{q,k}^{\text{{miss}}}(P_{q}^{\,\text{{fa}}},\tau_{q})$
in each player's feasible set become convex in the tuple $(P_{q}^{\,\text{{fa}}},\wh{\tau}_{q})$
{[}with $P_{q}^{\,\text{{fa}}}\leq\beta_{q}$ {]}. Indeed, for each
$k=1,\ldots,N$, we have
\begin{equation}
\begin{array}{ll}
P_{q,k}^{\text{{miss}}}(P_{q}^{\,\text{{fa}}},\tau_{q})\,\leq\,\alpha_{q,k} & \Leftrightarrow\,\,\,\dfrac{{\sigma_{q,k|0}}\,}{{\sigma_{q,k|1}}}\,\mathcal{Q}^{-1}\left(P_{q}^{\,\text{{fa}}}\right)-\wh{\tau}_{q}\,\dfrac{{\mu_{q,k|1}-\mu_{q,k|0}}}{{\sigma_{q,k|1}}}\,\leq\,\mathcal{Q}^{-1}\left(1-\alpha_{q,k}\right)\end{array},\label{eq:ineq_new_variables}
\end{equation}
where $\mathcal{Q}^{-1}\left(\cdot\right)$ denotes the inverse of
the Q-function {[}$\mathcal{Q}(x)$ is a strictly decreasing function
on $\mathbb{R}${]}, which are convex constraints in $(P_{q}^{\,\text{{fa}}},\wh{\tau}_{q})$
{[}provided that $P_{q}^{\,\text{{fa}}}\leq\beta_{q}$ {]}. Using
the above transformation, we can equivalently rewrite the missed detection
probability $P_{q,k}^{\text{{miss}}}(P_{q}^{\,\text{{fa}}},\tau_{q})$
and the throughput $R_{q}(\tau_{q},\,\mathbf{p},\, P_{q}^{\,\text{{fa}}})$
of each player $q$ in terms of the tuples $\left(\wh{\tau}_{q},\,\mathbf{p}_{q},\, P_{q}^{\,\text{{fa}}}\right)$'s,
denoted by $\wh{P}_{q,k}^{\text{{miss}}}(P_{q}^{\,\text{{fa}}},\wh{\tau}_{q})$
and $\wh{R}_{q}(\wh{\tau}_{q},\,\mathbf{p},\, P_{q}^{\,\text{{fa}}})$,
respectively; the explicit expression of these quantities is:
\begin{equation}
P_{q,k}^{\text{{miss}}}(P_{q}^{\,\text{{fa}}},\tau_{q})=\wh{P}_{q,k}^{\text{{miss}}}(P_{q}^{\,\text{{fa}}},\wh{\tau}_{q})\triangleq\mathcal{Q}\left({\displaystyle {\frac{\sigma_{{q,k}|0}\,\mathcal{Q}^{-1}\left(P_{q}^{\,\text{{fa}}}\right)-(\,\mu_{{q,k}|1}-\mu_{{q,k}|0}\,)\,\wh{\tau}_{q}}{\sigma_{{q,k}|1}}}}\right)\label{eq:P_miss_new}
\end{equation}
\begin{equation}
R_{q}\left({\tau}_{q},\,\mathbf{p},\, P_{q}^{\,\text{{fa}}}\right)=\wh{R}_{q}\left(\wh{\tau}_{q},\,\mathbf{p},\, P_{q}^{\,\text{{fa}}}\right)\triangleq\log\left(\left(\,1-{\displaystyle {\frac{\wh{\tau}_{q}^{2}}{f_{q}\, T_{q}}}\,}\right)\,{\displaystyle {\sum_{k=1}^{N}}\,\left(\,1-{P}_{q,k}^{\text{{fa}}}\,\right)\, r_{q,k}\left(\mathbf{p}\right)}\right).\label{eq:R_q_new}
\end{equation}

To incorporate the equi-sensing case in our unified formulation, we
introduce the functions $\theta_{q}(\mathbf{x}_{q},\,\mathbf{x}_{-q})$,
which represent the objective functions of the users including the
equi-sensing term, with $(\wh{\boldsymbol{{\tau}}},\,\mathbf{p},\,\mathbf{P}^{\,\text{{fa}}})\triangleq\left((\wh{{\tau}_{q}},\,\mathbf{p}_{q},\, P_{q}^{\,\text{{fa}}})\right)_{q=1}^{Q}$
denoting the strategy profile of all the players: 
\begin{equation}
\theta_{q}(\wh{\boldsymbol{{\tau}}},\,\mathbf{p},\,\mathbf{P}^{\,\text{{fa}}})\triangleq\hat{R}_{q}(\wh{\tau}_{q},\,\mathbf{p},\, P_{q}^{\,\text{{fa}}})-\,\dfrac{{c}}{2}\,\left(\dfrac{{\wh{\tau}_{q}}}{\sqrt{{f_{q}}}}-\dfrac{{1}}{Q}\,{\displaystyle {\sum_{r=1}^{Q}}}\,\dfrac{{\wh{\tau}_{r}}}{\sqrt{{f_{r}}}}\right)^{2}.\label{eq:theta_payoff}
\end{equation}
 We can now rewrite the feasible set of each player's optimization
problem in terms of the new variables $\left(\wh{\tau}_{q},\,\mathbf{p}_{q},\, P_{q}^{\,\text{{fa}}}\right)$,
denoted by $\mathcal{X}_{q}$: for each $q=1,\ldots,Q,$ let 
\begin{equation}
\mathcal{X}_{q}\triangleq\left\{ \left(\wh{\tau}_{q},\,\mathbf{p}_{q},\, P_{q}^{\,\text{{fa}}}\right)\in\mathcal{Y}_{q}\,\,|\,\, I_{q}\left(\wh{\tau}_{q},\,\mathbf{p}_{q},\, P_{q}^{\,\text{{fa}}}\right)\leq0\right\} \label{eq:set_Xq}
\end{equation}
 where we have separated the convex part and the nonconvex part; the
convex part is given by the polyhedron $\mathcal{Y}_{q}$ corresponding
to the constraints (b) and (c) in (\ref{eq:player q_individual_interference_constraints})
under the transformation (\ref{eq:bijection}) {[}cf. (\ref{eq:ineq_new_variables}){]}:
\begin{equation}
\hspace{-1em}\mathcal{Y}_{q}\triangleq\left\{ \begin{array}{ll}
\left(\wh{\tau}_{q},\,\mathbf{p}_{q},\, P_{q}^{\,\text{{fa}}}\right)\,|\, & P_{q}^{\,\text{{fa}}}\,\leq\,\beta_{q},\quad{\displaystyle \dfrac{{\sigma_{q,k|0}}\,}{{\sigma_{q,k|1}}}\mathcal{Q}^{-1}\left(P_{q}^{\,\text{{fa}}}\right)-\wh{\tau}_{q}\,\dfrac{{\mu_{q,k|1}-\mu_{q,k|0}}}{{\sigma_{q,k|1}}}\,\leq\,\wh{\alpha}_{q,k}},\quad\forall k=1,\ldots,N\\
 & \mathbf{p}_{q}\in\mathcal{P}_{q},\qquad\epc\epc\epc\wh{\tau}_{q}^{\,\min}\,\leq\,\wh{\tau}_{q}\,\leq\,\wh{\tau}_{q}^{\,\max}
\end{array}\right\} ,\label{eq:def_Y_q}
\end{equation}
with 
\begin{equation}
\wh{\alpha}_{q,k}\triangleq\mathcal{Q}^{-1}\left(1-\alpha_{q,k}\right),\quad\wh{\tau}_{q}^{\,\max}\triangleq\sqrt{{\tau}_{q}^{\,\max}\, f_{q}},\quad\mbox{and}\quad\wh{\tau}_{q}^{\,\min}\triangleq\sqrt{\tau_{q}^{\,\min}\, f_{q}},\label{eq:def_tau_hat}
\end{equation}
whereas the nonconvex part in (\ref{eq:set_Xq}) is given by the constraint
(a) that we have rewritten as $I_{q}(\wh{\tau}_{q},\,\mathbf{p}_{q},\, P_{q}^{\,\text{{fa}}})\leq0$
by introducing the local interference violation function 
\begin{equation}
I_{q}\left(\wh{\tau}_{q},\,\mathbf{p}_{q},\, P_{q}^{\,\text{{fa}}}\right)\triangleq{\displaystyle {\sum_{k=1}^{N}}\,\wh{P}_{q,k}^{\text{{miss}}}\left(P_{q}^{\,\text{{fa}}},\wh{\tau}_{q}\right)\cdot w_{q,k}\cdot p_{q,k}\,-\,{I}_{q}^{\text{\ensuremath{\max}}}}.\label{eq:def_h_function}
\end{equation}
This measures the violation of the \emph{local} interference constraint
(a) at $(\wh{\tau}_{q},\,\mathbf{p}_{q},\, P_{q}^{\,\text{{fa}}})$.
Similarly, it is convenient to introduce also the \emph{global} interference
violation function $I(\wh{\boldsymbol{{\tau}}},\,\mathbf{p},\,\mathbf{P}^{\,\text{{fa}}})$,
which depends on the strategy profile $(\wh{\boldsymbol{{\tau}}},\,\mathbf{p},\,\mathbf{P}^{\,\text{{fa}}})$
of all the players:
\begin{equation}
I(\wh{\boldsymbol{{\tau}}},\,\mathbf{p},\,\mathbf{P}^{\,\text{{fa}}})\triangleq{\sum_{k=1}^{N}}\,{\displaystyle {\sum_{q=1}^{Q}}{\displaystyle \,\wh{P}_{q,k}^{\text{{miss}}}\left(P_{q}^{\,\text{{fa}}},\wh{\tau}_{q}\right)\cdot w_{q,k}\cdot p_{q,k}}\,-I^{\text{{max}}}};\label{eq:map_interference}
\end{equation}
$I(\wh{\boldsymbol{{\tau}}},\,\mathbf{p},\,\mathbf{P}^{\,\text{{fa}}})$
measures the violation of the \emph{global} interference constraint
(\ref{eq:global_interference_constraints_2}) at $(\wh{\boldsymbol{{\tau}}},\,\mathbf{p},\,\mathbf{P}^{\,\text{{fa}}})$;
global interference constraints (\ref{eq:global_interference_constraints_2})
can be then rewritten in terms of $I(\wh{\boldsymbol{{\tau}}},\,\mathbf{p},\,\mathbf{P}^{\,\text{{fa}}})$
as $I(\wh{\boldsymbol{{\tau}}},\,\mathbf{p},\,\mathbf{P}^{\,\text{{fa}}})\leq0$.

Based on the above definitions, throughout the paper, we will use
the following notation. The convex part of the \emph{joint} strategy
set is denoted by $\mathcal{Y}\triangleq\prod_{q=1}^{Q}\mathcal{Y}_{q}$,
whereas the set containing all the (convex part of) players' strategy
sets except the $q$-th one is denoted by $\mathcal{Y}_{-q}\triangleq\prod_{r\neq q}\mathcal{Y}_{r}$;
similarly, we define $\mathcal{X}\triangleq\prod_{q=1}^{Q}\mathcal{X}_{q}$
and $\mathcal{X}_{-q}\triangleq\prod_{r\neq q}\mathcal{X}_{r}$. For
notational simplicity, when it is needed, we will use interchangeably
either $(\wh{\tau}_{q},\,\mathbf{p}_{q},\, P_{q}^{\,\text{{fa}}})$
or $\mathbf{x}_{q}\triangleq(\wh{\tau}_{q},\,\mathbf{p}_{q},\, P_{q}^{\,\text{{fa}}})$
to denote the strategy tuple of player $q$; similarly, the strategy
profile of all the players will be denoted either by $\mathbf{x}\triangleq(\mathbf{x}_{q})_{q=1}^{Q}$
or $(\wh{\boldsymbol{{\tau}}},\,\mathbf{p},\,\mathbf{P}^{\,\text{{fa}}})$,
with $\wh{\boldsymbol{{\tau}}}\triangleq(\wh{\tau}{}_{q})_{q=1}^{Q}$,
$\mathbf{p}\triangleq(\mathbf{p}_{q})_{q=1}^{Q},$ and $\mathbf{P}^{\text{{fa}}}\triangleq(P_{q}^{\text{{fa}}})_{q=1}^{Q}$,
whereas $\mathbf{x}_{-q}\triangleq(\mathbf{x}_{r})_{q\neq r=1}^{Q}$
is the strategy profile of all the players except the $q$-th one.
All the tuples above are intended to be column vectors; for instance,
$(\wh{\boldsymbol{{\tau}}},\,\mathbf{p},\,\mathbf{P}^{\,\text{{fa}}})$
signifies $(\wh{\boldsymbol{{\tau}}},\,\mathbf{p},\,\mathbf{P}^{\,\text{{fa}}})=[\wh{\boldsymbol{{\tau}}}^{T},\,\mathbf{p}^{T},\,\mathbf{P}^{\,\text{{fa}}^{T}}]^{T}$,
with $\wh{\boldsymbol{{\tau}}}\triangleq(\wh{\tau}{}_{q})_{q=1}^{Q}=[\wh{\tau}{}_{1},\ldots,\wh{\tau}{}_{Q}]^{T}$,
$ $ $\mathbf{p}\triangleq(\mathbf{p}_{q})_{q=1}^{Q}=[\mathbf{p}_{1}^{T},\ldots,\mathbf{p}_{Q}^{T}]^{T},$
where each $\mathbf{p}_{q}=(p_{q,k})_{k=1}^{N}=[p_{q,1},\ldots,p_{q,N}]^{T},$
and $\mathbf{P}^{\,\text{{fa}}}=(P_{q}^{\,\text{{fa}}})_{q=1}^{Q}=[P_{1}^{\,\text{{fa}}},\ldots,P_{Q}^{\,\text{{fa}}}]^{T}.$
For future convenience, Table \ref{table_notation} collects the above
definitions and symbols. Using the above notation, the games introduced
in the previous sections can be unified under the following reformulation.

\begin{center}
\vspace{-0.7cm}%
\framebox{\begin{minipage}[t]{1\columnwidth}%
\textbf{Players' optimization}. The optimization problem of player
$q$ is: 
\begin{equation}
\begin{array}{lll}
\underset{\mathbf{x}_{q}}{\mbox{maximize}} &  & \theta_{q}(\mathbf{x}_{q},\,\mathbf{x}_{-q})-{\pi}\cdot I(\mathbf{x})\\
\mbox{subject to} &  & \mathbf{x}_{q}\triangleq\left(\wh{\tau}_{q},\,\mathbf{p}_{q},\, P_{q}^{\,\text{{fa}}}\right)\in\mathcal{X}_{q}.
\end{array}\vspace{-0.1cm}\label{eq:player q transformed 1}
\end{equation}
\textbf{Price equilibrium}. The price obeys the following complementarity
condition:\vspace{-0.3cm}

\begin{equation}
0\,\leq\,{\pi}\,\perp\,-I(\mathbf{x})\geq0.\label{eq:price equilibrium_vector_form}
\end{equation}
\end{minipage}}
\par\end{center}

Throughout the paper, we will refer to the game (\ref{eq:player q transformed 1})
along with the side constraint (\ref{eq:price equilibrium_vector_form})
as game $\mathcal{G}(\mathcal{X},\,{\boldsymbol{{\theta}}})$, where
${\boldsymbol{{\theta}}}\triangleq(\theta_{q}(\mathbf{x}_{q},\mathbf{x}_{-q},\pi))_{q=1}^{Q}$.

\begin{table}[ht]\vspace{-0.1cm} 
\caption{Glossary of notation of game $\mathcal{G}(\mathcal{X},\,{\boldsymbol{{\theta}}})$ [cf. (\ref{eq:player q transformed 1})-(\ref{eq:price equilibrium_vector_form})]} \vspace{0.3cm} 
\centering 
\vline\vline
\begin{tabular}{l l } 
\hline\hline                  
\hspace{1.5cm}Symbol & \hspace{3.2cm}Meaning \\
\hline\vspace{-0.3cm}\\
$\tau_q$ & sensing time of SU $q$   \\ 
$\mathbf{p}_q\triangleq (p_{q,k})_{k=1}^N$ & power allocation vector of SU $q$ \\ 
$\pi$ & scalar price variable \\$P_q^{\text{fa}}$ & false alarm probability of SU $q$  \\
$P_{q,k}^{\text{miss}}$ & missed detection probability of SU $q$ on carrier $k$ [cf. (\ref{eq:Pd-worst_case_2})]\\
 $\wh{\tau}_{q}\triangleq \sqrt{\tau_q f_q}$ & normalized sensing time of SU $q$ [cf. (\ref{eq:bijection})] \\
$\mathbf{x}_{q}\triangleq(\wh{\tau}_{q},\,\mathbf{p}_{q},\, P_{q}^{\,\text{{fa}}})$ & strategy tuple of SU $q$ \\
$\mathbf{x}_{-q}\triangleq (\wh{\tau}_{r},\,\mathbf{p}_{r},\, P_{r}^{\,\text{{fa}}})_{r\neq q}$ & strategy profile of all the SUs except the $q$-th one\\
$\mathbf{x}\triangleq (\mathbf{x}_q)_{q=1}^{Q}=(\wh{\boldsymbol{{\tau}}},\,\mathbf{p},\,\mathbf{P}^{\,\text{{fa}}})$ & strategy profile of all the SUs  \\
$\theta_{q}(\mathbf{x}_{q},\,\mathbf{x}_{-q})$ & payoff function of SU $q$ including the equisensing penalization [cf. (\ref{eq:theta_payoff})]\\
$I_q(\mathbf{x}_q)$ & local interference constraint violation of SU $q$ [cf. (\ref{eq:def_h_function})]\\
$I(\mathbf{x})$ & global interference constraint violation of SU $q$ [cf. (\ref{eq:map_interference})]\\
${\cal{X}}_q$,  $\mathcal{X}\triangleq\prod_{q=1}^{Q}\mathcal{X}_{q}$ & feasible set of SU $q$ [cf. (\ref{eq:set_Xq})],  joint feasible strategy set of $\mathcal{G}(\mathcal{X},\,{\boldsymbol{{\theta}}})$\\
$\mathcal{X}_{-q}\triangleq\prod_{r\neq q}\mathcal{X}_{r}$ & joint strategy set of the SUs except the $q$-th one\\
${\cal{Y}}_q$, $\mathcal{Y}\triangleq\prod_{q=1}^{Q}\mathcal{Y}_{q}$ & convex part of ${\cal{X}}_q$ [cf. (\ref{eq:def_Y_q})], Cartesian product of all ${\cal{Y}}_q$'s \vspace{0.2cm}\\
\hline\hline  
\end{tabular}\vline\vline
\label{table_notation} 
\end{table}

Needless to say, when $\pi=0$ and $c=0$, $\mathcal{G}(\mathcal{X},\,{\boldsymbol{{\theta}}})$
reduces to the game in (\ref{eq:player q_individual_interference_constraints})
where there are only individual interference constraints (\ref{eq:individual_overal_interference_constraint}),
whereas when $c=0$, $\mathcal{G}(\mathcal{X},\,{\boldsymbol{{\theta}}})$
coincides with the game in (\ref{eq:paler_q_global})-(\ref{eq:price equilibrium_pricing_global})
with local and global interference constraints. 

As a final remark, we observe that the proposed formulations may be
extended to cover more general settings, without affecting the validity
of the results we are going to present. For instance, the case of
multiple active PUs and additional local/global interference constraints
(such as per-carrier constraints) can be readily accommodated: Instead
of having a single price variable, we associate a different price
to each global interference constraint and proceed similarly as in
(\ref{eq:player q transformed 1})-(\ref{eq:price equilibrium_vector_form}).
Also, the sensing model introduced in Sec. \ref{sub:The-spectrum-sensing}
can be generalized to the case of multiple active PUs, and the presence
of device-level uncertainties (e.g., uncertainty in the power spectral
density of the PUs' signals and thermal noise) as well as system level
uncertainties (e.g., lack of knowledge of the number of active PUs).
The mathematical details of these more general formulations can be
found in our companion paper \cite{Pang-Scutari-NNConvex_PI}; for
notational simplicity, here we will stay within the formulation (\ref{eq:player q transformed 1})-(\ref{eq:price equilibrium_vector_form}),
without loss of generality.

\section{Solution Analysis: Nash Equilibria\label{sec:Solution-Analysis:-Nash}}

This section is devoted to the solution analysis of the games introduced
in the previous section. In order to provide a unified analysis, we
focus on the general game $\mathcal{G}(\mathcal{X},\,{\boldsymbol{{\theta}}})$
with side constraints; results for the other proposed formulations
are obtained as special cases. We start our analysis by studying the
feasibility of each optimization problem in (\ref{eq:player q transformed 1})
(cf. Sec. \ref{sub:Feasibility-conditions}); we then extend the definitions
of NE\emph{ }to a game with side constraints and establish its main
properties (cf. Sec. \ref{sub:Existence-and-uniqueness}).

\subsection{Feasibility conditions\label{sub:Feasibility-conditions}}

Introducing the SNR detection $\texttt{{snr}}_{q,k}^{\text{{d}}}\triangleq{\sigma_{I_{q,k}}^{2}}/\sigma_{q,k}^{2}$
experimented by SU $q$ over carrier $k$ and using the definitions
given in Sec. \ref{sub:The-spectrum-sensing}, sufficient conditions
guaranteeing the existence of an optimal solution for each player's
optimization problem (\ref{eq:player q transformed 1}) are the following:
For all $q=1,\ldots,Q$ and $k=1,\ldots,N$, there must exist a common
sensing time $\tau$ (corresponding to normalized sensing times $\wh{\tau}_{q}=\sqrt{{\tau\, f_{q}}}$)
such that

\begin{equation}
\dfrac{\wh{\tau}_{q}^{\min}}{\sqrt{{f_{q}}}}\leq\sqrt{\tau}\leq\dfrac{\wh{\tau}_{q}^{\max}}{\sqrt{f_{q}}},\quad\mbox{and}\quad\sqrt{f_{q}{\tau}}\geq\,{\displaystyle {\frac{\mathcal{Q}^{-1}({\beta}_{q,k})+|\mathcal{Q}^{-1}({\alpha}_{q,k})|\,\left({\sigma}_{{q,k}|1}/{\sigma}_{{q,k}|0}\right)}{\texttt{{snr}}_{q,k}^{\mbox{d}}}}}.\label{eq:nec_suff_feasibility_cond}
\end{equation}
The first set of conditions in (\ref{eq:nec_suff_feasibility_cond})
simply postulates the existence of an overlap among the (normalized)
sensing time intervals $[\wh{\tau}_{q}^{\min}/\sqrt{{f_{q}}},\,\wh{\tau}_{q}^{\max}/\sqrt{{f_{q}}}]$
in (\ref{eq:player q transformed 1}), which is necessary to guarantee
the existence of a common value for the sensing times in the original
variables $\tau_{q}$'s. The second set of conditions guarantees that
the strategy sets $\mathcal{Y}_{q}$'s (and thus $\mathcal{X}_{q}$'s)
are not empty. Interestingly, they quantify the existing trade-off
between the sensing time (the product ``time-bandwidth'' $f_{q}{\tau}$
of the system) and detection accuracy: the smaller both false alarm
and missed detection probability values, the larger the sensing time
(the decision process must be more accurate). 

When the sensing times are not forced to be the same, as in the formulations
(\ref{eq:player q_individual_interference_constraints}) and (\ref{eq:paler_q_global})-(\ref{eq:price equilibrium_pricing_global}),
the feasibility conditions (\ref{eq:nec_suff_feasibility_cond}) can
be weakened by the following: For all $q=1,\ldots,Q$ and $k=1,\ldots,N$,

\begin{equation}
\sqrt{f_{q}{\tau}_{q}^{\text{{max}}}}\geq\,{\displaystyle {\frac{\mathcal{Q}^{-1}({\beta}_{q,k})+|\mathcal{Q}^{-1}({\alpha}_{q,k})|\,\left({\sigma}_{{q,k}|1}/{\sigma}_{{q,k}|0}\right)}{\texttt{{snr}}_{q,k}^{\mbox{d}}}}}.\label{eq:feasibility_no_equisensing}
\end{equation}
Throughout the paper, we tacitly assume that each user's optimization
problem under consideration has a nonempty strategy set (the associated
feasibility conditions above are satisfied).

\subsection{Existence and uniqueness of the NE\label{sub:Existence-and-uniqueness}}

We focus in this section on the NE of $\mathcal{G}(\mathcal{X},\,{\boldsymbol{{\theta}}})$.
The definition of NE for  a game with price equilibrium conditions
such as $\mathcal{G}(\mathcal{X},\,{\boldsymbol{{\theta}}})$ is the
natural generalization of the same concept introduced for classical
noncooperative games having no side constraints (see, e.g., \cite{Rosen_econ65})
and is given next.\vspace{-0.7cm}

\begin{center}
\textbf{}%
\framebox{\begin{minipage}[t]{1\columnwidth}%
\begin{flushleft}
\textbf{Definition}. A \textbf{Nash equilibrium} of the game $\mathcal{G}(\mathcal{X},\,{\boldsymbol{{\theta}}})$
is a strategy-price tuple $\left(\mathbf{x}^{\star},\,{\pi}\right)$,
such that 
\begin{equation}
\mathbf{x}_{q}^{\star}\,\in\,{\displaystyle {\operatornamewithlimits{\mbox{argmax}}_{\mathbf{x}_{q}\,\in\,\mathcal{X}_{q}}}\,\left\{ \theta_{q}(\mathbf{x}_{q},\mathbf{x}_{-q}^{\star})-{\pi}^{\star}\cdot I(\mathbf{x}_{q},\mathbf{x}_{-q}^{\star})\right\} },\quad\forall q=1,\ldots,Q,\label{eq:player q opt}
\end{equation}
 and 
\begin{equation}
0\,\leq\,{\pi}^{\star}\,\perp\,-\, I(\mathbf{x}^{\star})\geq0.\label{eq:side constraint-1}
\end{equation}
A NE is said to be \emph{trivial }if the power-component $\mathbf{p}_{q}^{\star}=\mathbf{0}$
for all $q=1,\ldots,Q$. \hfill{}$\Box$ 
\par\end{flushleft}%
\end{minipage}}
\par\end{center}

In words, the proposed notion of equilibrium is a stable state of
the network consisting of an equilibrium power/sensing profile $\mathbf{x}^{\star}$
and price $\pi^{\star}$: at $(\mathbf{x}^{\star},\pi^{\star})$,
the SUs have no incentive to change their power/sensing profiles $\mathbf{x}^{\star}$
based on the current state of the network {[}represented by (\ref{eq:player q opt}){]},
while the optimal value $\pi^{\star}$ of the price is such that all
global interference constraints are met {[}a situation represented
by (\ref{eq:side constraint-1}){]}. Note that, for a set of fixed
price $\pi^{\star}$, the equilibrium power/sensing profile $\mathbf{x}^{\star}$
can be interpreted as the NE of a classical noncooperative game (having
thus only local constraints), wherein the payoff function of each
player $q$ is $\theta_{q}(\bullet,\mathbf{x}_{-q},{\pi}^{\star})$
and the strategy set is $\mathcal{X}_{q}$. The proposed equilibrium
concept is thus a NE of the aforementioned game with an appropriately
selected price. 

The game $\mathcal{G}(\mathcal{X},\,{\boldsymbol{{\theta}}})$ is
nonconvex with the nonconvexity occurring in the players' objective
functions and the local/global interference constraints; moreover,
the feasible price {[}satisfying (\ref{eq:side constraint-1}){]}
is not explicitly bounded {[}note that this price cannot be normalized
due to the lack of homogeneity in the players' optimization problem
(\ref{eq:player q transformed 1}){]}. Because of that, the existence
of a NE is in jeopardy. The rest of this section is then devoted to
provide a detailed solution analysis of the game; we derive sufficient
conditions for the existence and the uniqueness of a NE. 

Mathematically, a NE can be interpreted as a fixed-point of the players'
best-response map. When this map is a continuous single valued function,
the existence of a fixed-point can be proved by using the renowned
Brouwer fixed-point theorem%
\footnote{Brouwer fixed-point theorem states that every continuous (vector-valued)
function $\Phi:\mathcal{C\mapsto C}$ defined over a nonempty convex
compact set $\mathcal{C}\subseteq\mathbb{R}^{n}$ has a fixed point
in $\mathcal{C}$.%
} (see, e.g., \cite[Th. 2.1.18]{Facchinei-Pang_FVI03}), provided that
one can identify a convex compact set for the application of the theorem.
Our goal is then to derive a set of sufficient conditions under which
the best-response map associated with $\mathcal{G}(\mathcal{X},\,{\boldsymbol{{\theta}}})$
is a \emph{single-valued continuous} map over a proper \emph{compact}
and \emph{convex} set; this is a nontrivial task, because of the nonconvexity
of the players' optimization problems and the potential unboundedness
of the price. The new line of analysis we propose is based on the
following three steps:
\begin{description}
\item [{$\mbox{Step\,\textbf{1}}:$}] To deal with the unboundedness of
the price, we introduce an auxiliary price-truncated game $\mathcal{G}_{t}(\mathcal{X},\,{\boldsymbol{{\theta}}})$,
where the price $\pi$ is constrained to be upper bounded by a given
positive constant $t$;
\item [{$\mathbf{Step}\,\textbf{2}:$}] We derive sufficient conditions
for the nonconvex players' optimization problems in the game $\mathcal{G}_{t}(\mathcal{X},\,{\boldsymbol{{\theta}}})$
to have unique optimal solutions; building on such solutions we introduce
a continuous single-value map$-$the best-response associated with
the game $\mathcal{G}_{t}(\mathcal{X},\,{\boldsymbol{{\theta}}})$$-$defined
on a convex and compact set, whose fixed-points are the Nash equilibria
of the game $\mathcal{G}_{t}(\mathcal{X},\,{\boldsymbol{{\theta}}})$.
We can then apply the Brouwer fixed-point theorem to deduce that $\mathcal{G}_{t}(\mathcal{X},\,{\boldsymbol{{\theta}}})$
has a NE; 
\item [{$\mathbf{Step}\,\textbf{3}:$}] The final step is to demonstrate
that there exists a sufficiently large $t$ such that the price truncation
in the game $\mathcal{G}_{t}(\mathcal{X},\,{\boldsymbol{{\theta}}})$
is not binding. This will allow us to deduce that a NE of $\mathcal{G}_{t}(\mathcal{X},\,{\boldsymbol{{\theta}}})$
is also a NE of the original, un-truncated, game $\mathcal{G}(\mathcal{X},\,{\boldsymbol{{\theta}}})$.\medskip{}

\end{description}

\subsubsection*{Step 1: The price-truncated game $\mathcal{G}_{t}(\mathcal{X},\,{\boldsymbol{{\theta}}})$ }

To motivate the price-truncated game, observe first that the price
complementarity condition in (\ref{eq:side constraint-1}) is equivalent
to 
\begin{equation}
\pi^{\star}\in{\operatornamewithlimits{\mbox{argmax}}_{\pi\geq0}}\,\left\{ \pi\cdot I(\mathbf{x}^{\star})\right\} .\label{eq:price_CC}
\end{equation}
In order to bound the price $\pi$ in (\ref{eq:price_CC}), let us
introduce the price interval defined as: given $t>0$, 
\begin{equation}
\mathcal{S}_{t}\,\triangleq\,\left\{ \,\pi\,\mid\,{\displaystyle 0\leq\pi\leq t}\right\} ,\label{eq:simplex_price}
\end{equation}
and truncate in (\ref{eq:price_CC}) the nonnegative axis $\pi\geq0$
by $\mathcal{S}_{t}$. We then replace (\ref{eq:price_CC}) with the
following price-truncated optimization problem: 
\begin{equation}
\pi_{t}^{\star}\in{\operatornamewithlimits{\mbox{argmax}}_{\pi_{t}\in\mathcal{S}_{t}}}\,\left\{ \pi_{t}\cdot I(\mathbf{x}^{\star})\right\} ,\label{eq:price_truncated_max}
\end{equation}
where instead of $\pi$ we used $\pi_{t}$ to make explicit the dependence
of the optimal solution of (\ref{eq:price_truncated_max}) on $t$.
Using (\ref{eq:price_truncated_max}), the price-truncated game $\mathcal{G}_{t}(\mathcal{X},\,{\boldsymbol{{\theta}}})$
can be defined as follows.

\vspace{0.3cm}

\hspace{-0.7cm}%
\framebox{\begin{minipage}[t]{1\columnwidth}%
\textbf{Game }$\mathcal{G}_{t}(\mathcal{X},\,{\boldsymbol{{\theta}}})$.
The game is composed of $Q+1$ players' optimization problems: the
following nonconvex optimization problems for the $Q$ players 
\begin{equation}
{\displaystyle {\operatornamewithlimits{\mbox{maximize}}_{\mathbf{x}_{q}\in\mathcal{X}_{q}}}\,\,}\theta_{q}\left(\mathbf{x}_{q},\,\mathbf{x}_{-q}\right)-\pi_{t}\cdot I(\mathbf{x}),\quad q=1,\ldots,Q,\label{eq:game_G_t_2}
\end{equation}
and the price-truncated optimization problem for the $(Q+1)$-st player
\begin{equation}
{\operatornamewithlimits{\mbox{maximize}}_{\pi_{t}\in\mathcal{S}_{t}}}\,\,\pi_{t}\cdot I(\mathbf{x}).\label{eq:price_CC_simplex}
\end{equation}
\end{minipage}} \vspace{0.2cm}

Note that in the game $\mathcal{G}_{t}(\mathcal{X},\,{\boldsymbol{{\theta}}})$
there are no side constraints, but the price complementarity condition
in (\ref{eq:side constraint-1}) is treated as an additional player
of the game, at the same level of the other $Q$ players. In fact,
this formulation facilitates the solution analysis of the game, as
detailed next. 

Let us start our analysis by rewriting the NE of $\mathcal{G}_{t}(\mathcal{X},\,{\boldsymbol{{\theta}}})$
as fixed-points of a proper best-response map defined on a convex
and compact set, which allows us to apply standard fixed-point arguments.
Given $t\geq0$, suppose that each optimization problem in (\ref{eq:game_G_t_2})
has a unique optimal solution for every fixed $\mathbf{x}_{-q}\in\mathcal{Y}_{-q}$
and $\pi_{t}\in\mathcal{S}_{t}$ (we derive shortly conditions for
this assumption to hold; see Proposition \ref{proposition_uniqueness_opt_sol}
below); let denote such a solution by $\mathbf{x}_{q}^{\star}(\mathbf{x}_{-q},\,\pi_{t})$,
i.e., 
\begin{equation}
\mathbf{x}_{q}^{\star}(\mathbf{x}_{-q},\,\pi_{t})\triangleq{\displaystyle {\operatornamewithlimits{argmax}_{\mathbf{z}_{q}\in\mathcal{X}_{q}}}\,}\left\{ \theta_{q}\left(\mathbf{z}_{q},\,\mathbf{x}_{-q}\right)-\pi_{t}\cdot I(\mathbf{z}_{q},\mathbf{x}_{-q})\right\} ,\label{eq:optimal_solution_of_q_problem}
\end{equation}
where in (\ref{eq:optimal_solution_of_q_problem}) we made explicit
the dependence of $\mathbf{x}_{q}^{\star}(\mathbf{x}_{-q},\,\pi_{t})$
on the strategy profile $\mathbf{x}_{-q}$ of the other players and
the price $\pi_{t}.$ In order to have a unique solution also of the
price-truncated linear optimization problem (\ref{eq:price_CC_simplex}),
we introduce the following proximal-based regularization in (\ref{eq:price_CC_simplex}):
given $t\geq0$, $\mathbf{x}\in\mathcal{Y}$, and $\pi_{t}\in\mathcal{S}_{t}$,
let 
\begin{equation}
\pi_{t}^{\star}(\mathbf{x},\,\pi_{t})\triangleq{\operatornamewithlimits{argmax}_{{\mu}_{t}\in\mathcal{S}_{t}}}\,\left\{ \mu_{t}\cdot I(\mathbf{x})-\frac{{1}}{2}\,\left({\mu}_{t}-{\pi}_{t}\right)^{2}\right\} .\label{eq:price_cc_regularized}
\end{equation}
Note that, thanks to the proximal regularization, the optimization
problem in (\ref{eq:price_cc_regularized}) becomes strongly convex
for any given $(\mathbf{x},\,\pi_{t})$, and thus has a unique solution
$\pi_{t}^{\star}(\mathbf{x},\,\pi_{t})$, which depends on $(\mathbf{x},\,\pi_{t})$.
Building on (\ref{eq:optimal_solution_of_q_problem}) and (\ref{eq:price_cc_regularized}),
we can introduce the following best-response map $\mathcal{B}:\mathcal{Y}\times\mathcal{S}_{t}\rightarrow\mathcal{Y}\times\mathcal{S}_{t}$
associated with the price-truncated game $\mathcal{G}_{t}(\mathcal{X},\,{\boldsymbol{{\theta}}})$:\vspace{-0.3cm}
\begin{equation}
\mathcal{Y}\times\mathcal{S}_{t}\ni(\mathbf{x},\,\pi_{t})\triangleq\left(\begin{array}{c}
\mathbf{x}_{1}\\
\vdots\\
\mathbf{x}_{Q}\\
\pi_{t}
\end{array}\right)\mapsto\mathcal{B}(\mathbf{x},\,\pi_{t})\triangleq\left(\begin{array}{c}
\mathbf{x}_{1}^{\star}(\mathbf{x}_{-1},\,\pi_{t})\\
\vdots\\
\mathbf{x}_{Q}^{\star}(\mathbf{x}_{-q},\,\pi_{t})\\
\pi_{t}^{\star}(\mathbf{x},\,\pi_{t})
\end{array}\right).\label{eq:best-response_map}
\end{equation}
 Note that, even though the feasible sets $\mathcal{X}_{q}$ of the
players' optimization problems in (\ref{eq:game_G_t_2}) are nonconvex,
the map $\mathcal{B}(\bullet)$ is defined over the \emph{convex}
and \emph{compact} set $\mathcal{Y}\times\mathcal{S}_{t}$; which
is a key point to apply the Brouwer fixed-point theorem. Moreover,
the set of fixed-points of $\mathcal{B}(\bullet)$ coincides with
that of the NE of the game $\mathcal{G}_{t}(\mathcal{X},\,{\boldsymbol{{\theta}}})$,
establishing thus the desired connection between the map (\ref{eq:best-response_map})
and the game $\mathcal{G}_{t}(\mathcal{X},\,{\boldsymbol{{\theta}}})$.
More formally, we have the following.\textbf{ }

\begin{lemma}\label{fixed-point_NE_Gt}Suppose that each optimization
problem in (\ref{eq:optimal_solution_of_q_problem}) has a unique
optimal solution for every given $\mathbf{x}_{-q}\in\mathcal{Y}_{-q}$
and $\pi_{t}\in\mathcal{S}_{t}$. A tuple $\left(\mathbf{x}^{\star},\,\pi_{t}^{\star}\right)$
is a NE of $\mathcal{G}_{t}(\mathcal{X},\,{\boldsymbol{{\theta}}})$
if and only if it is a fixed-point of the map $\mathcal{B}(\bullet)$;
that is $\left(\mathbf{x}^{\star},\,\pi_{t}^{\star}\right)=\mathcal{B}\left(\mathbf{x}^{\star},\,\pi_{t}^{\star}\right)$.\end{lemma}

Based on Lemma \ref{fixed-point_NE_Gt}, we can now study the existence
of a NE of $\mathcal{G}_{t}(\mathcal{X},\,{\boldsymbol{{\theta}}})$
by focusing on the fixed-points of the map $\mathcal{B}$.

\subsubsection*{Step 2: Existence of a NE of $\mathcal{G}_{t}(\mathcal{X},\,{\boldsymbol{{\theta}}})$}

We provide now sufficient conditions guaranteeing that each nonconvex
problem (\ref{eq:game_G_t_2}) has a unique optimal solution, for
every given $\mathbf{x}_{-q}\in\mathcal{Y}_{-q}$ and $\pi_{t}\in\mathcal{S}_{t}$.
Then, we show that these conditions are also sufficient for the existence
of a fixed-point of the map $\mathcal{B}$ in (\ref{eq:best-response_map}),
and thus a NE of the game $\mathcal{G}_{t}(\mathcal{X},\,{\boldsymbol{{\theta}}})$.

It is well-known that, under some Constraint Qualification (CQ), a
locally/globally optimal solution of a (possibly nonconvex) nonlinear
program satisfies the Karush-Kuhn-Tucker (KKT) conditions associated
with the optimization problem; such solutions are called stationary
solutions of the optimization problem. It turns out that to establish
the single-valuedness of the players' best-response map it is enough
to derive conditions guaranteeing the uniqueness of the stationary
solutions, provided that a suitable CQ holds. The classical approach
to write the KKT conditions of each player's optimization problem
would be introducing multipliers associated with \emph{all} the constraints
in the set $\mathcal{X}_{q}$$-$both the convex part $\mathcal{Y}_{q}$
and the nonconvex part $I_{q}(\mathbf{x}_{q})\leq0$ {[}cf. (\ref{eq:set_Xq}){]}$-$and
then maximizing the resulting Lagrangian function over\emph{ the whole
}space (i.e., considering an unconstrained optimization problem for
the Lagrangian maximization). The study of the uniqueness of the stationary
solutions based on the ``standard'' KKT conditions is however not
an easy task. To simplify the analysis, we propose here a different
approach: instead of explicitly accounting all the multipliers as
variables of the KKT system, for each player's optimization problem,
we introduce multipliers \emph{only for the nonconvex} constraints
$I_{q}(\mathbf{x}_{q})\leq0$, and retain the convex part $\mathcal{Y}_{q}$
as explicit constraints in the maximization of the resulting Lagrangian
function. More specifically, denoting by $\lambda_{q}$ the multiplier
associated with the nonconvex constraint $I_{q}(\mathbf{x}_{q})\leq0$
of player $q$, the Lagrangian function associated with the optimization
problem (\ref{eq:game_G_t_2}) of player $q$ (rewritten as a minimization)
is 
\begin{equation}
\mathcal{L}_{q}{\displaystyle \left(\left(\mathbf{x}_{q},\lambda_{q}\right),\,\mathbf{x}_{-q},\pi_{t}\right)}\triangleq-\theta_{q}(\mathbf{x}_{q},\mathbf{x}_{-q})+\lambda_{q}\cdot I_{q}(\mathbf{x}_{q})+\pi_{t}\cdot I(\mathbf{x}_{q},\mathbf{x}_{-q}),\label{eq:Lagrangian-1}
\end{equation}
which depends also on the strategies $\mathbf{x}_{-q}$ of the other
players and the price $\pi_{t}$. Given $\mathbf{x}_{-q}$ and $\pi_{t}$,
it is not difficult to see that if $\mathbf{x}_{q}^{\star}$ is an
optimal solution of the $q$-th player's optimization problem in (\ref{eq:player q transformed 1})
and some CQ holds at $\mathbf{x}_{q}^{\star}$, there exists a multiplier
${\lambda}_{q}^{\star}$ associated with the local nonconvex constraint
$I_{q}(\mathbf{x}_{q})\leq0$ such that the tuple $\left(\mathbf{x}_{q}^{\star},{\lambda}_{q}^{\star}\right)$
satisfies
\begin{equation}
\begin{array}{lc}
\mbox{(i)}:\quad & \mathbf{x}_{q}^{\star}\,\in\,\underset{\mathbf{x}_{q}\,\in\,\mathcal{Y}_{q}}{\mbox{argmin}}{\displaystyle \,\left\{ \mathcal{L}_{q}{\displaystyle \left((\mathbf{x}_{q},\lambda_{q}^{\star}),\mathbf{x}_{-q},\pi_{t}\right)}\right\} }\bigskip\\
\mbox{(ii)}: & 0\,\leq\,{\lambda}_{q}^{\star}\,\perp\,-\, I_{q}(\mathbf{x}_{q}^{\star})\geq0.
\end{array}\label{eq:KKT_player_q_1}
\end{equation}
Note that each Lagrangian minimization in (i) is constrained over
the convex part $\mathcal{Y}_{q}$ of the player's local constraints
$\mathcal{X}_{q}$. Since $\mathcal{Y}_{q}$ is a convex set, we can
invoke the variational principle for the optimality of $\mathbf{x}_{q}^{\star}$
in (i), and obtain the following necessary conditions for (\ref{eq:KKT_player_q_1})
to hold: 
\begin{equation}
\begin{array}{lc}
\mbox{(}\mbox{i}^{'}\mbox{)}:\quad & \left(\mathbf{x}_{q}-\mathbf{x}_{q}^{\star}\right)^{T}\,\nabla_{\mathbf{x}_{q}}\mathcal{L}_{q}{\displaystyle \left((\mathbf{x}_{q}^{\star},\lambda_{q}^{\star}),\,\mathbf{x}_{-q},\pi_{t}\right)}\geq0{\displaystyle \epc\forall\mathbf{x}_{q}\in\mathcal{Y}_{q}}\bigskip\\
\mbox{(ii\ensuremath{^{'}})}: & (\lambda_{q}-{\lambda}_{q}^{\star})\cdot\left(-\, I_{q}(\mathbf{x}_{q}^{\star})\right)\geq0,\quad\forall\lambda_{q}\in\mathbb{R}_{+}
\end{array}\label{eq:KKT_player_q_2}
\end{equation}
where $\mbox{(}\mbox{i}^{'}\mbox{)}$ is just the aforementioned first-order
(necessary) optimality condition of the (nonconvex) optimization problem
in $\mbox{(}\mbox{i}\mbox{)}$, albeit with a convex feasible set
${\cal Y}_{q}$; and $\mbox{(}\mbox{ii}^{'}\mbox{)}$ is equivalent
to $\mbox{(}\mbox{ii}\mbox{)}$. Finally, since there is no coupling
in the constraints involving the variables $\mathbf{x}_{q}$ and $\lambda_{q}$
in $(\mbox{i}\ensuremath{^{'}})$-$(\mbox{ii}\ensuremath{^{'}})$,
we can equivalently rewrite the two separated inequalities $(\mbox{i}\ensuremath{^{'}})$-$(\mbox{ii}\ensuremath{^{'}})$
as one inequality, obtaining 
\begin{equation}
\left(\begin{array}{c}
\mathbf{x}_{q}-\mathbf{x}_{q}^{\star}\smallskip\\
\lambda_{q}-{\lambda_{q}}^{\star}
\end{array}\right)^{T}\underset{\triangleq\mathbf{F}_{q}{\displaystyle \left((\mathbf{x}_{q}^{\star},\lambda_{q}^{\star});\,\mathbf{x}_{-q},\pi_{t}\right)}}{\underbrace{\left(\begin{array}{c}
\nabla_{\mathbf{x}_{q}}\mathcal{L}_{q}{\displaystyle {\displaystyle \left((\mathbf{x}_{q}^{\star},\lambda_{q}^{\star}),\,\mathbf{x}_{-q},\pi_{t}\right)}}\smallskip\\
-\, I_{q}(\mathbf{x}_{q}^{\star})\smallskip
\end{array}\right)}}\geq0,\quad\forall\left(\mathbf{x}_{q},\lambda_{q}\right)\in\underset{\triangleq\mathcal{K}_{q}}{\underbrace{\mathcal{Y}_{q}\times\mathbb{R}_{+}}}.\label{eq:VI_ref}
\end{equation}

The above system of inequalities defines the so-called VI problem
in the variables $\left(\mathbf{x}_{q},\lambda_{q}\right)$ for fixed
$(\mathbf{x}_{-q},\pi_{t})$, whose defining vector function is $\mathbf{F}_{q}\left(\bullet;\,\mathbf{x}_{-q},{\pi}_{t}\right)$
and feasible set is $\mathcal{K}_{q}$, both defined in (\ref{eq:VI_ref});%
\footnote{Given a set $\mathcal{Q}\subseteq\mathbb{R}^{n}$ and a vector-valued
function $\boldsymbol{\Psi}:\mathcal{Q}\rightarrow\mathbb{R}^{n}$,
the VI($\mathcal{Q},\boldsymbol{\Psi}$) problem is to find a point
$\mathbf{z}^{\star}\in\mathcal{Q}$, termed a solution of the VI,
such that $(\mathbf{z}-\mathbf{z}^{\star})^{T}\boldsymbol{\Psi}(\mathbf{z}^{\star})\geq0$
for all $\mathbf{z}\in\mathcal{Q}$ \cite{Facchinei-Pang_FVI03}. %
} such a VI is denoted by VI$(\mathcal{K}_{q},\mathbf{F}_{q})$. According
to the implications (\ref{eq:KKT_player_q_1})$\Rightarrow$(\ref{eq:VI_ref}),
the VI$(\mathcal{K}_{q},\mathbf{F}_{q})$ is an equivalent reformulation
of the KKT conditions of the $q$-th player's optimization problem
in (\ref{eq:player q transformed 1}), wherein the convex constraints
$\mathcal{Y}_{q}$'s (and thus the associated multipliers) have been
absorbed in the VI set $\mathcal{K}_{q}$, which is thus convex. It
turns out that the nonconvex problem in (\ref{eq:player q transformed 1})
has a unique optimal solution for any given $\mathbf{x}_{-q}$ and
$\pi_{t}$$-$the best-response of (\ref{eq:best-response_map}) is
unique, and thus $\mathbf{x}_{q}^{\star}(\mathbf{x}_{-q},\,\pi_{t})$
is well-defined$-$if the VI$(\mathcal{K}_{q},\mathbf{F}_{q})$ has
a unique $x_{q}$-component solution and some CQ holds. Proposition
\ref{proposition_uniqueness_opt_sol} below shows that Abadie CQ \cite[Ch. 3.2]{Facchinei-Pang_FVI03}
is satisfied by any nontrivial optimal solution of (\ref{eq:player q transformed 1})
and establishes the uniqueness of the $\mathbf{x}_{q}$-component
under the positive definiteness of the Hessian matrix $\nabla_{\mathbf{x}_{q}}^{2}\mathcal{L}_{q}\left((\mathbf{x}_{q},\lambda_{q}),\,\mathbf{x}_{-q},\pi_{t}\right)$
of $\mathcal{L}_{q}{\displaystyle \left(\left(\mathbf{x}_{q},\lambda_{q}\right),\,\mathbf{x}_{-q},\pi_{t}\right)}$,
for all $(\mathbf{x}_{q},\lambda_{q})\in\mathcal{K}_{q}$ and any
given $\mathbf{x}_{-q}\in\mathcal{Y}_{-q}$ and $\pi_{t}\geq0$. The
matrix $\nabla_{\mathbf{x}_{q}}^{2}\mathcal{L}_{q}\left((\mathbf{x}_{q},\lambda_{q}),\,\mathbf{x}_{-q},\pi_{t}\right)$
{[}interpreted as a function of $(\mathbf{x}_{q},\lambda_{q})$, for
fixed $\mathbf{x}_{-q}$ and $\pi_{t}${]} is given by
\begin{equation}
\nabla_{\mathbf{x}_{q}}^{2}\mathcal{L}_{q}\left((\mathbf{x}_{q},\lambda_{q}),\,\mathbf{x}_{-q},\pi_{t}\right)\triangleq-\nabla_{\mathbf{x}_{q}}^{2}\theta_{q}(\mathbf{x}_{q},\mathbf{x}_{-q})+\lambda_{q}\cdot\nabla_{\mathbf{x}_{q}}^{2}I_{q}(\mathbf{x}_{q})+\pi_{t}\cdot\nabla_{\mathbf{x}_{q}}^{2}I(\mathbf{x}_{q},\mathbf{x}_{-q}).\label{eq:L_q_2_matrix}
\end{equation}

Lemma \ref{Lemma_bounded_multipliers} in Appendix \ref{app:Proof-of-Proposition_uniqueness_NE}
shows that all the $\lambda_{q}$-solutions of the VI$(\mathcal{K}_{q},\mathbf{F}_{q})$
are bounded from above, for every given $\mathbf{x}_{-q}\in\mathcal{Y}_{q}$
and $\pi_{t}\in\mathcal{S}_{t}$. Specifically, it holds that any
$\lambda_{q}^{\star}$ satisfies $\lambda_{q}^{\star}\in[0,\,\lambda^{\max}]$
(see Lemma \ref{Lemma_bounded_multipliers} in Appendix \ref{app:Proof-of-Proposition_uniqueness_NE}),
with 
\begin{equation}
\lambda^{\max}\triangleq{\displaystyle {\sum_{q=1}^{Q}}\,{\displaystyle {\frac{1/\left[{\displaystyle {\min_{1\leq q\leq Q}}\,\left\{ \,{\displaystyle {I}_{q}^{\text{\text{{max}}}}},\,{\displaystyle {\min_{1\leq k\leq N}}\, p_{q,k}^{\max}\,}\right\} }\right]}{\left[\,{\displaystyle {\min_{1\leq k\leq N}}\,\left\{ \,\log\left(\,1+{\displaystyle {\frac{p_{q,k}^{\max}}{{\sigma}_{q,k}^{2}+{\displaystyle {\sum_{r\neq q}}\,|{H}_{qr}(k)|^{2}\, p_{r,k}^{\max}}}}\,}\right)\,\right\} \,}\right]\,{\displaystyle {\min_{1\leq k\leq N}}\,\left\{ {\sigma}_{q,k}^{2}\right\} }}}.}}\label{eq:t_star_def}
\end{equation}
$ $This allows us to restrict the requirement on the positive definiteness
of $\nabla_{\mathbf{x}_{q}}^{2}\mathcal{L}_{q}\left((\mathbf{x}_{q},\lambda_{q}),\,\mathbf{x}_{-q},\pi_{t}\right)$
on all $\mathbf{x}_{q}\in{\mathcal{Y}}_{q}$ and ${\lambda}_{q}\in[0,\,\lambda^{\max}]$.
The above discussion is made formal in the following proposition. 

\begin{proposition}\label{proposition_uniqueness_opt_sol}Let $\mathbf{x}_{-q}\in\mathcal{Y}_{-q}$
and $\pi_{t}\in\mathcal{S}_{t}$ for some $t>0$. Suppose that $\nabla_{\mathbf{x}_{q}}^{2}\mathcal{L}_{q}\left((\mathbf{x}_{q},\lambda_{q}),\,\mathbf{x}_{-q},\pi_{t}\right)$
in (\ref{eq:L_q_2_matrix}) is positive definite for all $\mathbf{x}_{q}\in{\mathcal{Y}}_{q}$
and ${\lambda}_{q}\in[0,\,\lambda^{\max}]$. Then, the $q$-th nonconvex
optimization problem in (\ref{eq:game_G_t_2}) has a unique optimal
solution $\mathbf{x}_{q}^{\star}\in\mathcal{X}_{q}$ that is necessarily
nontrivial. \end{proposition}\vspace{-0.5cm}

\begin{proof}See Appendix \ref{app:Proof-of-Proposition_uniqueness_NE}.\end{proof}

Note that under conditions in the above proposition, the optimization
problems in (\ref{eq:game_G_t_2}) remain nonconvex (the constraint
set $\mathcal{X}_{q}$ is indeed nonconvex). To shed light on the
physical interpretation of the obtained result, we provide in Corollary~\ref{corollary_sf_cond_uniqueness_opt_sol}
below easier conditions to be checked (but more restrictive) under
which Proposition \ref{proposition_uniqueness_opt_sol} is true. To
state the corollary, we use as weights $w_{q,k}$'s involved in the
interference constraints (\ref{eq:individual_overal_interference_constraint})
and (\ref{eq:global_interference_constraints_2}) the cross-channels
between secondary and primary users, i.e., $w_{q,k}=G_{P,q}(k)$,
for all $q=1,\ldots,Q$ and $k=1,\ldots,Q$ (more general conditions
are given in Appendix \ref{app:Proof-of-Proposition_uniqueness_NE}).

\begin{corollary}\label{corollary_sf_cond_uniqueness_opt_sol} Proposition
\ref{proposition_uniqueness_opt_sol} holds if the following sufficient
condition is satisfied: 
\begin{equation}
\gamma_{q}^{(1)}\cdot{\displaystyle {\displaystyle {\max_{k=1,\ldots,N}}\left\{ \dfrac{{|G_{P,q}(k)|^{2}}}{I^{\,{\rm tot}}}\right\} }<1,}\label{eq:diagonal_dominance_SF_cond}
\end{equation}
where $\gamma_{q}^{(1)}$ is a positive constant that depends only
on system/sensing parameters and it is defined in (\ref{eq:omega_q})
(cf. Appendix \ref{sec:Proof-of-Corollary_existence}) \end{corollary}\vspace{-0.5cm}\begin{proof}See
Appendix \ref{sec:Proof-of-Corollary_existence}.\end{proof}\vspace{-0.2cm}

The condition in (\ref{eq:diagonal_dominance_SF_cond}) has an interesting
physical interpretation: the nonconvex problem in (\ref{eq:game_G_t_2})
has a unique solution provided that the (normalized) cross-channels
between the secondary and the primary users are ``sufficiently''
small, meaning that there is not ``too much'' interference at the
primary receivers; see Sec. \ref{remark_cond} for more details on
the physical interpretation of the above conditions. 

Based on Proposition \ref{proposition_uniqueness_opt_sol} and Lemma
\ref{fixed-point_NE_Gt}, we can now establish the existence of a
NE of the game $\mathcal{G}_{t}(\mathcal{X},\,{\boldsymbol{{\theta}}})$
invoking the existence of a fixed-point of the single-valued mapping
$\mathcal{B}(\bullet)$ defined in (\ref{eq:best-response_map}).

\begin{proposition}\label{Proposition_existence of a NE of G_t}Given
$t>0$, suppose that each matrix $\nabla_{\mathbf{x}_{q}}^{2}\mathcal{L}_{q}\left((\mathbf{x}_{q},\lambda_{q}),\,\mathbf{x}_{-q},\pi_{t}\right)$
in (\ref{eq:L_q_2_matrix}) is positive definite for all $(\mathbf{x}_{q},\lambda_{q})\in{\mathcal{Y}}_{q}\times[0,\,\lambda^{\max}]$,
$\mathbf{x}_{-q}\in\mathcal{Y}_{q},$ and $\pi_{t}\in\mathcal{S}_{t}$.
Then, the game $\mathcal{G}_{t}(\mathcal{X},\,{\boldsymbol{{\theta}}})$
has a (nontrivial) NE.\end{proposition}

\begin{proof}Under the positive definiteness of each matrix $\nabla_{\mathbf{x}_{q}}^{2}\mathcal{L}_{q}\left((\mathbf{x}_{q},\lambda_{q}),\,\mathbf{x}_{-q},\pi_{t}\right)$,
the optimization problems (\ref{eq:optimal_solution_of_q_problem})
and (\ref{eq:price_cc_regularized}) have a unique optimal solutions
$\mathbf{x}_{q}^{\star}(\mathbf{x}_{-q},\,{\pi}_{t})$'s and ${\pi}_{t}^{\star}(\mathbf{x},\,{\pi}_{t})$,
respectively, for any given $\mathbf{x}\in\mathcal{Y}$ and $\pi_{t}\in S_{t}$.
Since these optimal solutions are unique, it is not difficult to show
that they are continuous functions of the parameters $(\mathbf{x},\,{\pi}_{t})$
(see, e.g., \cite{Zlobec_book-StableProgramming}), implying that
the single-valued map $\mathcal{B}$ in (\ref{eq:best-response_map})
is a continuous function on the convex and compact set ${\mathcal{Y}}\times\mathcal{S}_{t}$.
It follows from the Brouwer fixed-point theorem, that $\mathcal{B}$
has a fixed-point, which is a NE of the game $\mathcal{G}_{t}(\mathcal{X},\,{\boldsymbol{{\theta}}})$
(Lemma \ref{fixed-point_NE_Gt}). It follows from Proposition \ref{proposition_uniqueness_opt_sol}
that such a NE must be nontrivial. \end{proof}

\subsubsection*{Step 3: Existence and uniqueness of a NE of the game $\mathcal{G}(\mathcal{X},\,{\boldsymbol{{\theta}}})$ }

To pass from a NE of the price-truncated game $\mathcal{G}_{t}(\mathcal{X},\,{\boldsymbol{{\theta}}})$
to a NE of the original game $\mathcal{G}(\mathcal{X},\,{\boldsymbol{{\theta}}})$,
we argue that there exists a sufficiently large $t>0$ such that the
truncation constraint $\pi_{t}\leq t$ in $\mathcal{S}_{t}$ is not
binding at the optimal solution $\pi_{t}^{\star}$ of the price-truncated
optimization problem (\ref{eq:price_cc_regularized}), corresponding
to a NE of $\mathcal{G}_{t}(\mathcal{X},\,{\boldsymbol{{\theta}}})$.
This implies that a NE of $\mathcal{G}_{t}(\mathcal{X},\,{\boldsymbol{{\theta}}})$
is also a NE of $\mathcal{G}(\mathcal{X},\,{\boldsymbol{{\theta}}})$
and, as such, existence conditions given in Proposition \ref{Proposition_existence of a NE of G_t}
for the game $\mathcal{G}_{t}(\mathcal{X},\,{\boldsymbol{{\theta}}})$
apply also to $\mathcal{G}(\mathcal{X},\,{\boldsymbol{{\theta}}})$.
This is made formal in Theorem \ref{Theo_Existence-and-uniqueness_NE_G}
below, where we derive sufficient conditions for the existence and
uniqueness of a NE of $\mathcal{G}(\mathcal{X},\,{\boldsymbol{{\theta}}})$.

To introduce the theorem, we follow a similar approach as in Step
2: i) we first write the KKT conditions associated with the game $\mathcal{G}_{t}(\mathcal{X},\,{\boldsymbol{{\theta}}})$,
which under some CQ, are necessary conditions for a tuple $(\mathbf{x}^{\star},\pi_{t}^{\star})$
to be a NE of $\mathcal{G}_{t}(\mathcal{X},\,{\boldsymbol{{\theta}}})$
along with some multipliers associated with the local nonconvex constraints
$\{I_{q}(\mathbf{x}_{q})\leq0,\quad q=1,\ldots,Q\}$ and the truncation
in $\mathcal{S}_{t}$; and then ii) we rewrite this KKT system as
a proper VI problem, whose solution analysis leads to the desired
results (c.f. Theorem \ref{Theo_Existence-and-uniqueness_NE_G}). 

Under a suitable CQ, every NE $(\mathbf{x}^{\star},{\pi}_{t}^{\star})$
of $\mathcal{G}_{t}(\mathcal{X},\,{\boldsymbol{{\theta}}})$ will
satisfy the KKT conditions of the game, which are obtained by aggregating
the KKT conditions of players' optimization problems in (\ref{eq:game_G_t_2})
and (\ref{eq:price_CC_simplex}). Denoting by ${\lambda}_{q}^{\star}$
and $\eta_{t}^{\star}$ the multipliers associated with the nonconvex
constraint $I_{q}(\mathbf{x}_{q}^{\star})\leq0$ of player $q$ and
the price truncation $\pi_{t}^{\star}\leq t$ in $\mathcal{S}_{t}$,
respectively, and proceeding as in (\ref{eq:KKT_player_q_1})-(\ref{eq:VI_ref}),
the KKT conditions of $\mathcal{G}_{t}(\mathcal{X},\,{\boldsymbol{{\theta}}})$
that are necessarily satisfied by any NE $(\mathbf{x}^{\star},{\pi}_{t}^{\star})$
can be written as: 
\begin{equation}
\begin{array}{lc}
\mbox{(i)}: & \left(\begin{array}{c}
\mathbf{x}_{1}-\mathbf{x}_{1}^{\star}\\
\vdots\\
\mathbf{x}_{Q}-\mathbf{x}_{Q}^{\star}
\end{array}\right)^{T}\left(\begin{array}{c}
\nabla_{\mathbf{x}_{1}}\mathcal{L}_{1}{\displaystyle {\displaystyle \left((\mathbf{x}_{1}^{\star},\lambda_{1}^{\star}),\,\mathbf{x}_{-1}^{\star},\pi_{t}^{\star}\right)}}\\
\vdots\\
\nabla_{\mathbf{x}_{Q}}\mathcal{L}_{Q}{\displaystyle {\displaystyle \left((\mathbf{x}_{Q}^{\star},\lambda_{Q}^{\star}),\,\mathbf{x}_{-Q}^{\star},\pi_{t}^{\star}\right)}}\smallskip
\end{array}\right)\geq0,\quad\forall\mathbf{x}_{q}\in\mathcal{Y}_{q}\quad\mbox{and}\quad q=1,\ldots,Q,\bigskip\\
\mbox{(ii)}: & \left(\begin{array}{c}
(\lambda_{1}-{\lambda}_{1}^{\star})\\
\vdots\\
(\lambda_{Q}-{\lambda}_{Q}^{\star})
\end{array}\right)^{T}\left(\begin{array}{c}
-\, I_{1}(\mathbf{x}_{1}^{\star})\\
\vdots\\
-\, I_{Q}(\mathbf{x}_{Q}^{\star})\smallskip
\end{array}\right)\geq0,\quad\forall\lambda_{q}\geq0\quad\mbox{and}\quad q=1,\ldots,Q\bigskip\\
\mbox{(iii)}: & 0\leq\,{\pi}_{t}^{\star}\,\perp\,-I(\mathbf{x}^{\star})+\eta_{t}^{\star}\,\geq0\quad\mbox{and}\quad0\leq\eta_{t}^{\star}\,\perp\, t-{\pi}_{t}^{\star}\geq0.\qquad\qquad\qquad
\end{array}\label{eq:KKT_game_G_t_1}
\end{equation}
Observing that the complementarity conditions in (iii) of (\ref{eq:KKT_game_G_t_1})
are equivalent to the VI problem in the $\pi_{t}$ variable: 
\[
(\pi_{t}-\pi_{t}^{\star})\cdot(-I(\mathbf{x}^{\star}))\geq0,\quad\forall\pi_{t}\in\mathcal{S}_{t},
\]
the KKT system (\ref{eq:KKT_game_G_t_1}) can be equivalently rewritten
as 

\begin{equation}
\begin{array}{c}
\left(\begin{array}{c}
\mathbf{x}-\mathbf{x}^{\star}\medskip\\
\boldsymbol{{\lambda}-{\lambda}}^{\star}\medskip\\
\pi_{t}-\pi_{t}^{\star}
\end{array}\right)^{T}\underset{\triangleq\boldsymbol{{\Psi}}(\mathbf{x}^{\star},\,\boldsymbol{{\lambda}}^{\star},\,\pi_{t}^{\star})}{\underbrace{\left(\begin{array}{c}
\left(\nabla_{\mathbf{x}_{q}}\mathcal{L}_{q}{\displaystyle {\displaystyle \left((\mathbf{x}_{q}^{\star},\lambda_{q}^{\star}),\,\mathbf{x}_{-q}^{\star},\pi_{t}^{\star}\right)}}\right)_{q=1}^{Q}\medskip\\
\left(-I_{q}{\displaystyle (\mathbf{x}_{q}^{\star})}\right)_{q=1}^{Q}\medskip\\
-I(\mathbf{x}^{\star})\smallskip
\end{array}\right)}}\geq0,\quad\forall(\mathbf{x},\,\boldsymbol{{\lambda}},\,\pi_{t})\in\underset{\triangleq\mathcal{Z}_{t}}{\underbrace{\mathcal{Y}\times\mathbb{R}_{+}^{Q}\times\mathcal{S}_{t}}},\end{array}\label{eq:KKT_Game_Gt}
\end{equation}
which represents a VI problem in the tuple $(\mathbf{x},\,\boldsymbol{{\lambda}},\,\pi_{t})$,
i.e., VI$(\mathcal{Z}_{t},\,\boldsymbol{{\Psi}})$, with $\mathbf{x}=(\mathbf{x}_{q})_{q=1}^{Q}$
and $\boldsymbol{{\lambda}}\triangleq(\boldsymbol{\lambda}_{q})_{q=1}^{Q}$.

Based on the VI formulation (\ref{eq:KKT_Game_Gt}), in Appendix \ref{proof_Theo_Existence-and-uniqueness_NE_G}
we prove that the following two properties are satisfied by any solutions
$(\mathbf{x}^{\star},\,\boldsymbol{{\lambda}}^{\star},\,\pi_{t}^{\star})$
of VI$(\mathcal{Z}_{t},\,\boldsymbol{{\Psi}})$ and thus \emph{by
any NE} of $\mathcal{G}_{t}(\mathcal{X},\,{\boldsymbol{{\theta}}})$
(under some suitable CQ): i) at any $(\mathbf{x}^{\star},\,\boldsymbol{{\lambda}}^{\star},\,\pi_{t}^{\star})$,
$\pi_{t}^{\star}$ is bounded from above by $\pi_{t}^{\star}\leq\lambda^{\max}$,
with $\lambda^{\max}$ defined in (\ref{eq:t_star_def}); and ii)
the $\mathbf{x}$-component of $(\mathbf{x}^{\star},\,\boldsymbol{{\lambda}}^{\star},\,\pi_{t}^{\star})$
is unique if the Jacobian matrix of $\left(\nabla_{\mathbf{x}_{q}}\mathcal{L}_{q}{\displaystyle {\displaystyle \left((\mathbf{x}_{q},\lambda_{q}),\,\mathbf{x}_{-q},\pi_{t}\right)}}\right)_{q=1}^{Q}$
with respect to $\mathbf{x}$, denoted by $\mathbf{A}(\mathbf{x},\,\boldsymbol{{\lambda}},\,\pi_{t})$,
is positive definite on $\mathcal{Y}\times[0,\,\lambda^{\max}]^{Q}\times\mathcal{S}_{t}$,
with $\mathbf{A}(\mathbf{x},\,\boldsymbol{{\lambda}},\,\pi_{t})$
given by:
\begin{equation}
\mathbf{A}(\mathbf{x},\,\boldsymbol{{\lambda}},\,{\pi}_{t})\triangleq\mbox{J}_{\mathbf{x}}\left(\begin{array}{l}
\nabla_{\mathbf{x}_{1}}\mathcal{L}_{1}{\displaystyle {\displaystyle \left((\mathbf{x}_{1},\lambda_{1}),\,\mathbf{x}_{-1},\pi_{t}\right)}}\\
\vdots\\
\nabla_{\mathbf{x}_{Q}}\mathcal{L}_{Q}{\displaystyle {\displaystyle \left((\mathbf{x}_{Q},\lambda_{Q}),\,\mathbf{x}_{-Q},\pi_{t}\right)}}
\end{array}\right).\label{eq:def_A_matrix}
\end{equation}
 Building on the established connection between the NE of $\mathcal{G}_{t}(\mathcal{X},\,{\boldsymbol{{\theta}}})$
and the solutions of the VI$(\mathcal{Z}_{t},\,\boldsymbol{{\Psi}})$
and using properties i) and ii) above, we can finally obtain the desired
existence and uniqueness result: (a) It follows from property i) that
since the truncated game $\mathcal{G}_{t}(\mathcal{X},\,{\boldsymbol{{\theta}}})$
has a NE for $t>\lambda^{\max}$ (which is guaranteed under conditions
in Proposition \ref{Proposition_existence of a NE of G_t}), the original
game $\mathcal{G}(\mathcal{X},\,{\boldsymbol{{\theta}}})$ must have
a NE as well; and (b) According to property ii), if there exists a
$t>\lambda^{\max}$ such that $\mathbf{A}(\mathbf{x},\,\boldsymbol{{\lambda}},\,\pi_{t})$
is positive definite for all $(\mathbf{x},\,\boldsymbol{{\lambda}},\,\pi_{t})\in\mathcal{Y}\times\mathbb{R}_{+}^{Q}\times\mathcal{S}_{t}$,
the $\mathbf{x}$-component of the solution of the VI$(\mathcal{Z}_{t},\,\boldsymbol{{\Psi}})$$-$and
thus of the NE of $\mathcal{G}(\mathcal{X},\,{\boldsymbol{{\theta}}})$$-$is
unique. These results are collected in Theorem \ref{Theo_Existence-and-uniqueness_NE_G}
below and formally proved in Appendix \ref{proof_Theo_Existence-and-uniqueness_NE_G}. 

\begin{theorem}\label{Theo_Existence-and-uniqueness_NE_G}Given the
game $\mathcal{G}(\mathcal{X},\,{\boldsymbol{{\theta}}})$ and $\lambda^{\max}$
defined in (\ref{eq:t_star_def}), the following hold:
\begin{description}
\item [{(a)}] Suppose that there exists a $t>\lambda^{\max}$ such that
\emph{each} matrix $\nabla_{\mathbf{x}_{q}}^{2}\mathcal{L}_{q}\left((\mathbf{x}_{q},\lambda_{q}),\,\mathbf{x}_{-q},\pi_{t}\right)$
in (\ref{eq:L_q_2_matrix}) is positive definite for all $(\mathbf{x}_{q},\lambda_{q})\in{\mathcal{Y}}_{q}\times[0,\,\lambda^{\max}]$,
$\mathbf{x}_{-q}\in\mathcal{Y}_{q},$ and $\pi_{t}\in\mathcal{S}_{t}$.
Then, every NE $(\mathbf{x}^{\star},{\pi}_{t}^{\star})$ of $\mathcal{G}_{t}(\mathcal{X},\,{\boldsymbol{{\theta}}})$
is a NE of $\mathcal{G}(\mathcal{X},\,{\boldsymbol{{\theta}}})$;
therefore $\mathcal{G}(\mathcal{X},\,{\boldsymbol{{\theta}}})$ has
a NE; 
\item [{(b)}] If the condition in (a) is strengthened by the following:
the matrix $\mathbf{A}(\mathbf{x},\,\boldsymbol{{\lambda}},\,{\pi}_{t})$
in (\ref{eq:def_A_matrix}) is positive definite for all $\mathbf{x}\in\mathcal{Y}$,
$\boldsymbol{{\lambda}}\in[0,\,\lambda^{\max}]^{Q}$, and $\pi_{t}\in\mathcal{S}_{t}$,
then the $\mathbf{x}$-component of the NE of the game $\mathcal{G}(\mathcal{X},\,{\boldsymbol{{\theta}}})$
is unique.
\end{description}
\end{theorem}

\begin{proof}See Appendix \ref{proof_Theo_Existence-and-uniqueness_NE_G}.\end{proof}

Sufficient conditions for the matrix $\mathbf{A}(\mathbf{x},\,\boldsymbol{{\lambda}},\,{\pi}_{t})$
to be positive definite are given in the following. 

\begin{corollary}\label{corollary_sf_cond_uniqueness_NE} Statement
(b) \emph{{[}}and thus also (a)\emph{{]}} of Theorem \ref{Theo_Existence-and-uniqueness_NE_G}
true if the following sufficient conditions are satisfied: for all
$q=1,\ldots,Q$,
\begin{equation}
\gamma_{q}^{(1)}\cdot{\displaystyle {\displaystyle {\max_{k=1,\ldots,N}}\left\{ \dfrac{{|{G}_{Pq}(k)|^{2}}}{I^{\,{\rm tot}}}\right\} }+\gamma_{q}^{(2)}\cdot\sum_{r\neq q}\left({\max_{k=1,\ldots,N}}\left\{ {\displaystyle \frac{|{H}_{qr}(k)|^{2}}{{\sigma}_{q,k}^{2}}}\right\} +{\max_{k=1,\ldots,N}}\left\{ {\displaystyle {\frac{|{H}_{rq}(k)|^{2}}{{\sigma}_{r,k}^{2}}}}\right\} \right)<1,}\label{eq:diagonal_dominance_A_pd}
\end{equation}
where $\gamma_{q}^{(1)}$ and $\gamma_{q}^{(2)}$ are positive constants
depending only on system/sensing parameters and are defined in (\ref{eq:omega_q})
and (\ref{eq:gamma_2}), respectively (cf. Appendix \ref{sec:Proof-of-corollary_sf_cond_uniqueness_NE}). 

\end{corollary}\begin{proof}See Appendix \ref{sec:Proof-of-corollary_sf_cond_uniqueness_NE}.\end{proof}

\subsection{Discussion on the existence/uniqueness conditions \label{remark_cond}}

Corollary \ref{corollary_sf_cond_uniqueness_opt_sol} and Corollary
\ref{corollary_sf_cond_uniqueness_NE} suggest an intuitive physical
interpretation of the equilibrium existence/uniqueness conditions:
existence of an equilibrium and uniqueness of the $\mathbf{x}$-component
are ensured if the MUI in the network is sufficiently small (compared
to the background noise). More specifically, existence results in
(\ref{eq:diagonal_dominance_SF_cond}) impose a limit (only) on the
maximum interference that the the SUs are allowed to generate at the
primary receivers, measured by ${\displaystyle {\max_{k=1,\ldots,N}}\left\{ {|{G}_{Pq}(k)|^{2}}/I^{\,{\rm tot}}\right\} }$.
Uniqueness conditions in (\ref{eq:diagonal_dominance_A_pd}) impose
instead a limit on the maximum MUI experienced at \emph{both} primary
and secondary receivers. This is clear looking at the LHS of (\ref{eq:diagonal_dominance_A_pd}):
the first term on the LHS, ${\displaystyle {\max_{k=1,\ldots,N}}\left\{ {|{G}_{Pq}(k)|^{2}}/I^{\,{\rm tot}}\right\} },$
coincides with that of (\ref{eq:diagonal_dominance_SF_cond}), imposing
thus a limit on the MUI at the PU, whereas the second term, $\sum_{r\neq q}{\underset{k=1,\ldots,N}{\max}}\left\{ |{H}_{qr}(k)|^{2}/{\sigma}_{q,k}^{2}\right\} +\sum_{r\neq q}{\underset{k=1,\ldots,N}{\max}}\left\{ |{H}_{rq}(k)|^{2}/{\sigma}_{r,k}^{2}\right\} $,
limits the overall MUI in the secondary network; indeed, the quantity
$\sum_{r\neq q}$ ${\underset{k=1,\ldots,N}{\max}}\left\{ |{H}_{rq}(k)|^{2}/{\sigma}_{r,k}^{2}\right\} $
is an estimate of the maximum interference generated by each SU $q$
against all the other SUs $r$'s, and $\sum_{r\neq q}{\underset{k=1,\ldots,N}{\max}}\left\{ {\displaystyle |{H}_{qr}(k)|^{2}/{\sigma}_{q,k}^{2}}\right\} $
can be interpreted as a limit on the maximum MUI tolerable by each
secondary receiver $q$ and generated by all the other secondary transmitters
$r$'s. These two sources of MUI affect the uniqueness through the
constants $\gamma_{q}^{(1)}$ and $\gamma_{q}^{(2)}$, which depend
on the fixed sensing/device-level parameters as well as on the SU/PUs'
QoS requirements (e.g., maximum false alarm rate/minimum detection
probability, and maximum sensing time constraints).

Interestingly, conditions in (\ref{eq:diagonal_dominance_A_pd}) are
of the same genre as those obtained in the literature to guarantee
the uniqueness of the NE of \emph{convex} games modeling the power
control problem in ad-hoc networks \cite{Luo-Pang_IWFA-Eurasip,Huang-Cendrillon-Chiang-Moonen_SP07,Scutari-Palomar-Barbarossa_SP08_PI,Scutari-Palomar-Barbarossa_AIWFA_IT08}
and CR systems \cite{Scutari-Palomar-Facchinei-Pang_SPMag09,Pang-Scutari-Palomar-Facchinei_SP_10}.
The main difference is that, because of the nonconvexity of some constraints
and the joint optimization of sensing and transmission strategies,
in (\ref{eq:diagonal_dominance_A_pd}), there is an extra term, ${\displaystyle {\max_{k=1,\ldots,N}}\left\{ {|{G}_{Pq}(k)|^{2}}/I^{\,{\rm tot}}\right\} },$
limiting the interference generated also against the PUs and the two
weights $\gamma_{q}^{(1)}$ and $\gamma_{q}^{(2)}$ capturing the
sensing/QoS requirements.

\section{Distributed Algorithms\textmd{\normalsize \label{sub:Game_with_fixed_prices}}{\normalsize \vspace{-0.2cm}}}

This section is devoted to the design of distributed algorithms that
solve the proposed class of games and the study of their convergence.
Before analyzing the most general game $\mathcal{G}(\mathcal{X},\,\boldsymbol{\theta})$,
we focus on solution methods for the game where the price $\pi$ is
a fixed exogenous parameter (and thus there are only local interference
constraints). The resulting algorithms will be used as a subroutine
in an extended iterative algorithm solving the more complex game $\mathcal{G}(\mathcal{X},\,\boldsymbol{\theta})$
wherein the prices are endogenous variables to optimize.

\subsection{Game with exogenous price\label{sub:Game-with-exogenous}}

When the price $\pi$ is an exogenous fixed parameter, game $\mathcal{G}(\mathcal{X},\,\boldsymbol{\theta})$
reduces to the following game. 

\vspace{0.2cm}\hspace{-0.5cm}%
\framebox{\begin{minipage}[t]{0.98\columnwidth}%
\textbf{Game }$\mathcal{G}_{{\pi}}(\mathcal{X},\,\boldsymbol{{\theta}})$.
The optimization problem of player $q$ is: given $\mathbf{x}_{-q}\in\mathcal{X}_{-q}$
and $\pi\geq0$, 
\begin{equation}
{\displaystyle {\operatornamewithlimits{\mbox{maximize}}_{\mathbf{x}_{q}\in\mathcal{X}_{q}}}\,\,}\theta_{q}\left(\mathbf{x}_{q},\,\mathbf{x}_{-q}\right)-\pi\cdot I(\mathbf{x}_{q},\mathbf{x}_{-q})\quad q=1,\ldots,Q.\label{eq:game_G_t_pi}
\end{equation}
\end{minipage}}\vspace{0.2cm}

We have denoted such a game by $\mathcal{G}_{{\pi}}(\mathcal{X},\,\boldsymbol{{\theta}})$,
making explicit the fact that $\pi$ is an exogenous fixed parameter.
Note that $\mathcal{G}_{{\pi}}(\mathcal{X},\,\boldsymbol{{\theta}})$
contains as special cases the game with zero price (and thus no global
interference constraints) as introduced in Sec. \ref{sub:GT_local},
and the equisensing game with constant price $\pi$ (and local interference
constraints only), which is an instance of the game $\mathcal{G}(\mathcal{X},\,\boldsymbol{\theta})$
introduced in Sec. \ref{sub:GT_global}. Therefore, Algorithms for
$\mathcal{G}_{{\pi}}(\mathcal{X},\,\boldsymbol{{\theta}})$ apply
also to the aforementioned special cases. 

We are interested in iterative schemes based on the best-response
mapping: according to a given scheduling (e.g., sequentially, simultaneously,
or asynchronously), each SU solves his own optimization problem in
(\ref{eq:game_G_t_pi}), given the strategies of the others. If this
procedure converges and some suitable conditions are satisfied, it
will converge to a NE of the game $\mathcal{G}_{{\pi}}(\mathcal{X},\,\boldsymbol{{\theta}})$.
The Jacobi version of the proposed class of algorithms wherein the
users update their strategies simultaneously is formally described
in Algorithm \ref{async_best-response_algo}.

\begin{algo}{Jacobi Best-Response-Consensus Algorithm for $\mathcal{G}_{{\pi}}(\mathcal{X},\,\boldsymbol{{\theta}})$}
S\texttt{$\mbox{(\mbox{S.0})}:$} Choose any feasible $\mathbf{x}^{(0)}\in\mathcal{X}$
and set $n=0$.

\texttt{$\mbox{(S.1)}:$} \texttt{If} $\mathbf{x}^{(n)}$ satisfies
a suitable termination criterion: \texttt{STOP}. \\
\texttt{$\mbox{(S.2)}:$} Run a consensus algorithm to locally compute
the average $\dfrac{{1}}{Q}\,{\displaystyle {\sum_{r=1}^{Q}}}\,\dfrac{{\wh{\tau}_{r}^{(n)}}}{\sqrt{{f_{r}}}}$.\\
 \texttt{$\mbox{(S.3)}:$}\noun{ }\texttt{for} $q=1,\ldots,Q$,$\,$compute
\begin{equation}
\mathbf{x}_{q}^{(n+1)}\in\underset{\mathbf{x}_{q}\in\mathcal{X}_{q}}{\text{{argmax}}}\,\,\left\{ \theta_{q}\left(\mathbf{x}_{q},\,\mathbf{x}_{-q}^{(n)}\right)-\pi\cdot I(\mathbf{x}_{q},\mathbf{x}_{-q})\right\} .\vspace{-0.2cm}\label{eq:Async_update}
\end{equation}
\\
 $\qquad\,$$\quad\,\,\,$\texttt{$\mbox{(S.4)}:$} $n\leftarrow n+1$;
go to \texttt{$\mbox{(S.1)}$}.\label{async_best-response_algo}\end{algo}

In order to relax constraints on the synchronization of the players\textquoteright{}
updates, totally asynchronous schemes (in the sense specified in \cite{Bertsekas_Book-Parallel-Comp})
can be considered, where some SUs may update their strategy profile
more frequently than others and they may even use an outdated measurement
of the interference generated by the others (we refer to \cite{Bertsekas_Book-Parallel-Comp}
and \cite{Scutari-Palomar-Barbarossa_AIWFA_IT08} for a formal description
of asynchronous algorithms). The analysis of this general class of
algorithms is addressed in Appendix \ref{sec:Convergence-of-Asynchronous_BR_local_constraints},
where we provide sufficient conditions for their convergence; see
Theorem \ref{Theo-async_best-response_NEP} and Corollary \ref{Corollary_SF_Cond_convergence_algo_zero_pricing}.
Since Algorithm \ref{async_best-response_algo} is an instance of
these asynchronous schemes, it converges under the same aforementioned
conditions. It is worth remarking that the obtained convergence conditions
have the same physical interpretation of that given for the existence/uniqueness
of the NE (cf. Sec. \ref{remark_cond}). Roughly speaking, they require
``low'' interference in the network, meaning ``small'' values
of the (normalized) secondary cross-channels $|{H}_{qr}(k)|^{2}/{\sigma}_{q,k}^{2}$
as well as secondary-primary cross-channels ${|{G}_{Pq}(k)|^{2}}/I^{\,{\rm tot}}$.
Interestingly, they do not depend on the specific updating scheduling
used by the users, meaning that the whole class of asynchronous algorithms
converges under the same set of unified conditions. The main implication
of this result is that all the algorithms obtained as special case
of the asynchronous scheme, such as the \emph{sequential} (Gauss-Seidel
scheme) and the \emph{simultaneous} (Jacobi scheme) best-response
algorithms, are robust against missing or outdated updates of the
players.

\subsubsection{Discussion on the implementation\label{Rmk_implementation_issues} }

We discuss now some implementation issues related to the proposed
algorithms; for notational simplicity, we will focus only on Algorithm
\ref{async_best-response_algo}, but similar conclusions can be drawn
also for the asynchronous implementation. 

In Step 3 of the algorithm, each user $q$ needs to compute its best-response,
knowing the information on the strategies of the others $\mathbf{x}_{-q}^{(n)}=(\mathbf{x}_{r}^{(n)})_{r\neq q=1}^{Q},$
with each $\mathbf{x}_{r}=(\wh{\tau}_{r},\mathbf{p}_{r},P_{r}^{\text{{fa}}})$.
Given the structure of the feasible set $\mathcal{X}_{q}$ {[}specifically,
the presence of local interference constraints (\ref{eq:individual_overal_interference_constraint}){]}
and the functional dependence of the objective function in (\ref{eq:game_G_t_pi})
on $\mathbf{x}_{-q}$ {[}see (\ref{eq:player q transformed 1}){]},
this knowledge requires each SU $q$ to estimate: i) the overall Power
Spectral Density (PSD) of the MUI at each subcarrier, $\sum_{r\neq q}|{H}{}_{qr}(k)|^{2}p_{r,k}$;
ii) the primary-secondary cross-channel function $\left(G_{Pq}(k)\right)_{k=1}^{N}$
{[}if the weights $w_{q,k}$'s in the local interference constraints
(\ref{eq:individual_overal_interference_constraint}) are chosen as
$w_{q,k}=G_{Pq}(k)${]}; and iii) the average of the (normalized)
sensing times $(1/Q)\,\sum_{r=1}^{Q}(\wh{\tau}_{r}/\sqrt{{f_{s}^{(r)}}})$
of all the SUs. Among other remarks, we discuss next alternative distributed
protocols to obtain these estimates, each of them being characterized
by a different level (albeit limited) of signaling among the SUs and
computational complexity.

\subsubsection*{Estimate of the MUI\emph{ }and the primary-secondary cross-channels}

To measure the MUI in a totally distributed way, it is enough for
the SUs to perform a preliminary noise calibration of their receivers
(during this phase of course the SUs must stay silent). After this
noise calibration phase, to acquire the MUI, the SUs just need to
locally measure the global interference experienced at their receivers.
Note that this procedure does not require the SUs to be able to distinguish
between primary and secondary signaling. 

Because of the presence of the individual interference constraints
in the set $\mathcal{X}_{q}$, each SU needs to estimate also the
secondary-primary cross-channel transfer function $\left(G_{Pq}(k)\right)_{k=1}^{N}$
{[}if in (\ref{eq:individual_overal_interference_constraint}) one
uses $w_{q,k}=G_{Pq}(k)${]}. This knowledge can be acquired by each
SU in advance by using classical channel estimation techniques, and
updated at the rate of the channel coherence time. In the CR scenarios
where the PUs cannot communicate with the SUs (e.g., when the PUs
are legacy systems) and thus cannot be involved in the (cross-)channel
estimation, and the primary receivers have a fixed geographical location,
it may be possible to install some monitoring devices close to each
primary receiver having the functionality of (cross-)channel/interference
measurement. 

In scenarios where the above options are not feasible and the channel
state information cannot be acquired, a different choice of the weights
coefficients $w_{q,k}$'s and the interference threshold $I_{q}^{\max}$
in (\ref{eq:individual_overal_interference_constraint}) can be made,
based on worst-case channel/interference statistics. More specifically,
one can replace the instantaneous value of the secondary-primary cross-channel
transfer function $\left(G_{Pq}(k)\right)_{k=1}^{N}$ with its expected
value; the expected value of each $G_{Pq}(k)$ is 
\begin{equation}
\text{{E}}\left\{ \left|G_{Pq}(k)\right|^{2}\right\} =\dfrac{{\sigma_{g}}}{1+\left(d_{Pq}/d_{0}\right)^{\varsigma}},\label{eq:pathloss}
\end{equation}
where $\sigma_{g}$ is a positive constant depending on the number
of resolvable paths and their variance; $\varsigma$ is the path loss
exponent, which generally is $2\leq\varsigma\leq6$; $d_{Pq}$ is
the distance between the SU $q$ and the PU; and $d_{0}$ is the Fraunhofer
distance. The interference constraints imposed to each SU $q$ become
then 
\begin{equation}
\sum_{k=1}^{N}P_{q,k}^{\text{{miss}}}\left({\tau}_{q},\, P_{q}^{\,\text{{fa}}}\right)\cdot\dfrac{{\sigma_{g}}}{1+\left(d_{Pq}/d_{0}\right)^{\varsigma}}\cdot p_{q,k}\leq I_{q}^{\text{{max}}},\label{eq:Int_costr_2}
\end{equation}
which is still in the form of (\ref{eq:individual_overal_interference_constraint}),
with weights coefficients $w_{q,k}={\sigma_{g}}/(1+(d_{Pq}/d_{0})^{\varsigma})$. 

When the distance $d_{Pq}$ in (\ref{eq:Int_costr_2}) is unknown,
one can instead consider a probabilistic (conservative) version of
(\ref{eq:Int_costr_2}), based on the worst-case interference scenario,
as proposed in \cite{Barbarossa-Sardellitti_CIP12}. Modeling $d_{Pq}$
as a random variable, we can impose 
\begin{equation}
\text{{Prob}}\left\{ \sum_{k=1}^{N}P_{q,k}^{\text{{miss}}}\left({\tau}_{q},\, P_{q}^{\,\text{{fa}}}\right)\cdot\dfrac{{\sigma_{g}}}{1+\left(d_{Pq}/d_{0}\right)^{\varsigma}}\cdot p_{q,k}\leq I_{q}^{\text{{max}}}\right\} \geq P_{I},\label{eq:Int_costr_3}
\end{equation}
where $0\leq P_{I}\leq1$ is a given positive constant guaranteeing
the desired QoS at the primary receiver. To obtain an explicit expression
of the probability above, we consider next a more conservative constraint
implying (\ref{eq:Int_costr_3}). More specifically, denoting by $d_{\min}\triangleq\min_{q}d_{Pq}$
the distance between the PU and the nearest SU $q$, the following
interference constraint implies (\ref{eq:Int_costr_3}): 
\begin{equation}
\text{{Prob}}\left\{ \sum_{k=1}^{N}P_{q,k}^{\text{{miss}}}\left({\tau}_{q},\, P_{q}^{\,\text{{fa}}}\right)\cdot\dfrac{{\sigma_{g}}}{1+\left(d_{\min}/d_{0}\right)^{\varsigma}}\cdot p_{q,k}\leq I_{q}^{\text{{max}}}\right\} \geq P_{I}.\label{eq:Int_costr_4}
\end{equation}
Assuming that the SUs are randomly distributed according to a homogeneous
Poisson point process with spatial density $\rho$, $d_{\min}\triangleq\min_{q}d_{Pq}$
is Rayleigh distributed; the probability in (\ref{eq:Int_costr_4})
can be then evaluated in closed form and we obtain \cite{Barbarossa-Sardellitti_CIP12}
\begin{equation}
\sum_{k=1}^{N}P_{q,k}^{\text{{miss}}}\left({\tau}_{q},\, P_{q}^{\,\text{{fa}}}\right)\cdot p_{q,k}\leq\bar{{I}}_{q}^{\max},\quad\mbox{with}\quad\bar{{I}}_{q}^{\max}\triangleq\dfrac{{I}_{q}^{\max}}{{\sigma_{g}}}\cdot\left(1+\dfrac{|{\ln}(P_{I})|}{\pi\rho r_{0}^{2}}\right)\label{eq:Int_constr_5}
\end{equation}
which is still in the form of (\ref{eq:individual_overal_interference_constraint}),
where $w_{q,k}=1$ and the interference threshold ${I}_{q}^{\max}$
is replaced by $\bar{{I}}_{q}^{\max}$.

\subsubsection*{Estimate of the average sensing time {[}Step 2{]}}

The average of the sensing times can be locally computed by each SU
by running a consensus based algorithm that requires the interaction
only between nearby secondary nodes, as stated in Step 2. Consensus
algorithms have become popular over the past few decades since \cite{Tsitsiklis_PhD84}
as a practical scheme for the in-network distributed calculation of
general functions of the node values; several protocols suitable for
different applications and working under different network settings
have been proposed and their properties analyzed; see, e.g., \cite{OlfatiSaber-Fax-Murray_ProcIEEE,Scutari-Barbarossa-Pescosolido_ConsDelay_SP08}
for a good overview of recent results. In order to minimize the running
time of the consensus iterates and thus the amount of signaling to
be exchange in Step 2 by the SUs, we suggest here to use the \emph{finite-time
}distributed convergence linear scheme proposed in \cite{Sundaram-Hadjicostis_JSAC08}.
The main advantage of this scheme with respect to the more classical
consensus/gossip algorithms whose convergence is only asymptotic (i.e.,
exact consensus is not reached in a finite number of times) is that,
at no extra signaling, each node can immediately calculate the consensus
value after observing the evolution of its own value over a \emph{finite
} number of time-iterations (specifically, upper bounded by the size
of the network). 

The consensus scheme we consider in Step 2 of Algorithm \ref{consensus_algorithm}
makes use of the following liner iterations: given the (normalized)
sensing times $\wh{\tau}_{q}^{(n)}$'s obtained as output of Step
3 at iterations $n$, and setting $z_{q}^{(0)}=\wh{\tau}_{q}^{(n)}/\sqrt{{f_{q}}}$,
each SU $q$ updates at each (inner) time-iteration $i$ its value
as 
\begin{equation}
z_{q}^{(i+1)}=a_{qq}\, z_{q}^{(i)}+\sum_{r\in\mathcal{N}_{q}}a_{qr}\,\left(z_{r}^{(i)}-z_{q}^{(i)}\right)\label{consensus-iterates}
\end{equation}
where $\mathcal{N}_{q}$ is the set of neighbors of user $q$, which
are the nodes that interfere with node $q$ (the SUs' network is modeled
as a directed graph); the cardinality of $\mathcal{N}_{q}$, the number
of neighbors of node $q$, is denoted by $\text{{deg}}_{q}^{\text{{in}}}\triangleq\left|\mathcal{N}_{q}\right|$
(also called in the graph theory jargon the \emph{in-degree} of node
$q$); and the $a_{qr}$'s are a set of given coefficients. These
weights represent a degree of freedom in the algorithm design; here
we focus on the following choice that can be made locally by each
SU $q$: 
\begin{equation}
a_{qr}=\left\{ \begin{array}{lll}
1, &  & \mbox{if }r\in\mathcal{N}_{q}\\
0, &  & \mbox{if }r\notin\mathcal{N}_{q}\\
F-\text{{deg}}_{q}^{\text{{in}}} &  & \mbox{if }r=q,
\end{array}\right.\label{weight_matrix}
\end{equation}
where $F$ is any integer number. Associated with the SUs' network
topology, there are some absolute quantities that play a role in the
stopping criterion of the iterates (\ref{consensus-iterates}) and
the computation of the final consensus value. More specifically, for
each node $q,$ there exist a scalar $0\leq L_{q}\leq Q-\mbox{deg}_{q}$
and a ($L_{q}+1$)-length vector $\mathbf{m}_{q}\in\mathbb{R}^{L_{q}+1}$
having the following properties \cite{Sundaram-Hadjicostis_JSAC08}:
given the samples $z_{q}^{(0)},\dots,z_{q}^{(L_{q})}$ collected by
the SU $q$ in the first $L_{q}+1$ iterations of (\ref{consensus-iterates}),
it holds that 
\begin{equation}
\mathbf{m}_{q}^{T}\,\left[\begin{array}{c}
z_{q}^{(0)}\\
\vdots\\
z_{q}^{(L_{q})}
\end{array}\right]=\dfrac{{1}}{Q}\,\sum_{r=1}^{Q}z_{r}^{(0)}=\dfrac{{1}}{Q}\,{\displaystyle {\sum_{r=1}^{Q}}}\,\dfrac{{\wh{\tau}_{r}^{(n)}}}{\sqrt{{f_{r}}}}.\label{eq:consensus}
\end{equation}
According to (\ref{eq:consensus}), each SU $q$ can obtain locally
the desired average of the sensing times after running the linear
iterates (\ref{consensus-iterates}) for $L_{q}+1$ time-steps; this
will require at most $Q-\mbox{deg}_{q}+1$ time-iterations. Note that,
to calculate the quantity in (\ref{eq:consensus}), the SUs do not
need to store the entire set of samples $z_{q}^{(0)},\dots,z_{q}^{(L_{q})}$;
instead one can compute the scalar product in (\ref{eq:consensus})
incrementally, as the iterations progress. 

To implement the above protocol distributively, each SU $q$ has to
preliminarily estimate his own $L_{q}$ and $\mathbf{m}_{q}$; for
time-invariant topologies this can be done just once; the cost of
this computation will then be amortized over the number of times the
consensus algorithm is performed. In \cite{Sundaram-Hadjicostis_JSAC08},
the authors proposed a \emph{decentralized} protocol still based on
the updating (\ref{consensus-iterates}) to perform such a computation
in (at most) $Q(Q-1)$ iterations; we refer the interested reader
to \cite[Sec. V]{Sundaram-Hadjicostis_JSAC08} for details. The consensus
protocol discussed above is formally described in Algorithm \ref{consensus_algorithm}
below, which represents the subroutine to implement Step 2 of Algorithm
\ref{async_best-response_algo}. 

\begin{algo}{Finite-time Consensus Algorithm in Step 2 of Algorithm
\ref{async_best-response_algo}} S\texttt{$\mbox{Data}:$} $\wh{\tau}_{q}^{(n)}$
{[}from Step 2 of Algorithm \ref{async_best-response_algo}{]}, $L_{q}$,
$\mathbf{m}_{q}$, and $(a_{qr})_{r=1}^{Q}$, for all $q=1,\ldots,Q$.\\
\texttt{$\mbox{(S.2a)}:$} Set $z_{q}^{(0)}=\wh{\tau}_{q}^{(n)}$,
for $q=1,\ldots,Q$ and set $i=0$.\\
\texttt{$\mbox{(S.2b)}:$}\noun{ }\texttt{for} $i=1,\ldots,\max_{q}L_{q}$,$ $ 

\quad\quad\quad\quad $-$ Each SU $ $$q$ updates $z_{q}^{(i)}$
according to (\ref{consensus-iterates}) 

$\qquad\,\,\,$$\quad\,\,\,$$-$ \texttt{if} $i==L_{q}$ for some
$q$, \texttt{then} SU $ $$q$ computes (\ref{eq:consensus}) and
gets $\dfrac{{1}}{Q}\,{\displaystyle {\sum_{r=1}^{Q}}}\,\dfrac{{\wh{\tau}_{r}^{(n)}}}{\sqrt{{f_{r}}}}$; 

\quad\quad\quad\quad\texttt{end }(\texttt{for}).\label{consensus_algorithm}\end{algo}\medskip{}

In Algorithm \ref{consensus_algorithm}, the number of iterations
$i$ required to propagate the consensus over the whole network is
$\max_{q}\{L_{q}\}+1\leq Q-\min_{q}\{\mbox{deg}_{q}\}+1$. One can
reduce such a number by slightly changing the above protocol: SU $q$
runs the iteration (\ref{consensus-iterates}) for $L_{q}+1$ consecutive
time-steps, or until he receives the consensus value from a neighbor.
If $L_{q}+1$ iterations passes without receiving the consensus value,
SU $q$ calculates that value and broadcast it to his neighbors, along
with a flag indicating that it is the consensus value (and not just
an intermediate value). In this way, ``slower'' SUs $r$'s will
receive the final value at most one iteration after node $q$.

\subsubsection*{On the time-complexity and communication costs }

We quantify now the complexity of Algorithm \ref{async_best-response_algo}
(whose Step 2 is implemented using Algorithm \ref{consensus_algorithm})
in terms of the minimum number of iterations required to reach the
desired convergence accuracy and communication costs (number of message
passing among the SUs). Both results come readily from the following
two facts.

\noindent \emph{Fact 1}. The convergence conditions of Algorithm
\ref{async_best-response_algo} as given in Theorem \ref{Theo-async_best-response_NEP}
in Appendix \ref{sec:Convergence-of-Asynchronous_BR_local_constraints}
are based on the contraction properties of the best-response mapping
$\mathcal{B}_{\pi_{t}}(\mathbf{x})\triangleq\left(\mathbf{x}_{q}^{\star}(\mathbf{x}_{-q},\,\pi_{t})\right)_{q=1}^{Q}$
associated with the game $\mathcal{G}_{{\pi}}(\mathcal{X},\,\boldsymbol{{\theta}})$
in (\ref{eq:game_G_t_pi}), with each $\mathbf{x}_{q}^{\star}(\mathbf{x}_{-q},\,\pi_{t})$
defined in (\ref{eq:optimal_solution_of_q_problem}): under assumptions
in Theorem \ref{Theo-async_best-response_NEP}, there exists a constant
$c_{\mathcal{B}}\in(0,1)$ such that {[}see (\ref{eq:contraction_constant})
in Appendix \ref{sec:Convergence-of-Asynchronous_BR_local_constraints}{]}
\begin{equation}
\left\Vert \mathcal{B}_{\pi_{t}}(\mathbf{x})-\mathcal{B}_{\pi_{t}}(\mathbf{y})\right\Vert \leq c_{\mathcal{B}}\,\left\Vert \mathbf{x}-\mathbf{y}\right\Vert ,\qquad\forall\mathbf{x},\mathbf{y}\in\mathcal{X},\label{eq:contraction}
\end{equation}
where an explicit expression of the contraction constant $c_{\mathcal{B}}$
is given in (\ref{eq:contraction_constant}) (cf. Appendix \ref{sub:Contraction-properties-of_BR_of_G_pi}).
If the ``suitable termination criterion'' in Step 2 of Algorithm
\ref{async_best-response_algo} is chosen as the smallest iteration
$n=n_{\min}$ at which the relative error $\left\Vert \mathcal{B}_{\pi_{t}}(\mathbf{x}^{(n)})-\mathbf{x}^{\star}\right\Vert /\left\Vert \mathbf{x}^{(0)}-\mathbf{x}^{\star}\right\Vert $
is less than a prescribed tolerance $\epsilon_{\max}>0$ {[}with $\mathbf{x}^{\star}$
being the NE of $\mathcal{G}_{{\pi}}(\mathcal{X},\,\boldsymbol{{\theta}})${]},
(\ref{eq:contraction}) leads to 
\begin{equation}
n_{\min}\geq\dfrac{{\ln}\left(1/\epsilon_{\max}\right)}{\ln\left|c_{\mathcal{B}}\right|},\label{eq:number_of_iterations}
\end{equation}
which provides the number of iterations $n$ required for Algorithm
\ref{async_best-response_algo} to reach convergence (within the accuracy
$\epsilon_{\max}$). 

\noindent \emph{Fact 2}. The consensus algorithm described in Algorithm
\ref{consensus_algorithm} was shown to converge in at most $\max_{q}\{L_{q}\}+1\leq Q-\min_{q}\{\mbox{deg}_{q}\}+1$
iterations. The communication cost incurred by the protocol can be
characterized as follows. Given the directed graph modeling the network
topology (the outgoing edges from each node $q$ link the nodes associated
with the SUs who receive interference from SU $q$), each SU $q$
transmits a scalar value on each outgoing edge at each time-step $i$;
since there are at most $\max_{q}\{L_{q}\}$ runs, each SU $q$ will
have in principle to transmit $(\max_{q}\{L_{q}\}+1)\cdot\text{{deg}}_{q}^{\text{{out}}}$
messages, where $\text{{deg}}_{q}^{\text{{out}}}$ is the out-degree
of node $q$ (i.e., the number of SUs having user $q$ as interferer).
Thanks to the broadcast nature of the wireless channel, however, a
single transmission of each user $q$ will be equivalent to communicating
a message along each of $\text{{deg}}_{q}^{\text{{out}}}$ outgoing
edges, and thus each node would only have to transmit $\max_{q}\{L_{q}\}+1$
messages. Summing over all nodes in the network, there will be $\sum_{q=1}^{Q}(\max_{q}\{L_{q}\}+1)$
overall messages that have to be transmitted to run the consensus
protocol. 

Using Facts 1 and 2 above, one can conclude that Algorithm \ref{async_best-response_algo}
(whose Step 2 is implemented by Algorithm \ref{consensus_algorithm})
converges (within the accuracy $\epsilon_{\max}$) in $\dfrac{{\ln}\left(1/\epsilon_{\max}\right)}{\ln\left|c_{\mathcal{B}}\right|}\cdot(\max_{q}\{L_{q}\}+1)$
(outer plus inner-loop) iterations, which is also the number of per/user
message passing.

\subsubsection*{A special case: fixed equi-sensing times}

In the scenarios where no coordination is allowed among the SUs to
run a consensus algorithm, one can implement a special case of Algorithm
\ref{async_best-response_algo}, where the SUs' sensing times are
fixed a-priori and thus not optimized. This would correspond to solving
the game $\mathcal{G}_{{\pi}}(\mathcal{X},\,\boldsymbol{{\theta}})$
in (\ref{eq:game_G_t_pi}) where the sensing times $\wh{\tau}_{q}$
are fixed and equal to a common value $\tau$; the resulting solution
scheme will be like Algorithm \ref{async_best-response_algo} where
there is no Step 2 and the optimization problems in (\ref{eq:game_G_t_pi})
are solved only with respect to the tuple $(\mathbf{p}_{q},P_{q}^{\text{{fa}}})$,
given $\wh{\tau}_{q}=\tau$. The time and communication complexity
of such an algorithm is of the same order of that required by the
well-known iterative waterfilling algorithm proposed and studied in
many papers \cite{Luo-Pang_IWFA-Eurasip,Scutari-Palomar-Barbarossa_SP08_PI,Scutari-Palomar-Barbarossa_AIWFA_IT08,Leshem-Zehavip_SPMag_sub09}
to distributively solve the rate maximization game over interference
channels (there is no optimization of the sensing part in any formulation
of that game). The price in the reduction of signaling obtained with
the fixing of sensing times may be paid in terms of overall performance;
in Sec. \ref{sec:Numerical-Results}, we numerically quantify the
loss in using a fixed sensing time rather than optimizing it. This
sheds some light on the trade-off between performance and signaling
in the proposed games.

\subsubsection*{On the best-response computation}

A last comment deals with the computation of the best-response of
each optimization problem (\ref{eq:game_G_t_pi}), which would require
the capability of solving a nonconvex problems. This is not a difficult
task under the assumption of Theorem \ref{Theo-async_best-response_NEP}
(cf. Appendix \ref{sec:Convergence-of-Asynchronous_BR_local_constraints}),
which ensures that each of such (nonconvex) optimization problems
has a unique stationary point (cf. Proposition \ref{proposition_uniqueness_opt_sol})
that can be computed by any of nonlinear programming solvers, provided
that each SU $q$ has the information on the strategies $\mathbf{x}_{-q}$
of the other SUs.

Finally, observe that, when conditions in Theorem \ref{Theo-async_best-response_NEP}
are not satisfied, every limit point of the sequence generated by
the proposed algorithms, wherein the best-response solution is replaced
by a stationary solution, has still some optimality properties: it
is guaranteed to be a QNE of the game, whose properties have been
studied in our companion paper \cite{Pang-Scutari-NNConvex_PI}.

\subsection{Game with endogenous prices\textmd{\normalsize{} \label{sub:Game_G_with_side_constraints}}}

We focus now on distributed algorithms for solving the general game
$\mathcal{G}(\mathcal{X},\,{\boldsymbol{{\theta}}})$. The main challenge
here is to obtain \emph{distributed} algorithms in the presence of
\emph{coupling nonconvex} constraints. The proposed approach is to
reduce the solution of the nonconvex game $\mathcal{G}(\mathcal{X},\,{\boldsymbol{{\theta}}})$
with side constraints to a solution of a \emph{sequence }of (compact
and) \emph{convex}%
\footnote{According to the terminology introduced in \cite{Rosen_econ65}, a
game is said to be compact and convex if: i) the feasible set of each
player is a convex and compact set; and ii) the cost function of each
player (to be minimized) is a convex and continuously differentiable
function of the strategy of that player, for any given strategy profile
of the other players. The desired properties of such games are: i)
each player optimization problem is a convex problem and thus it can
be solved using efficient numerical algorithms; and i) they always
have a NE. %
} games of a particular structure \emph{with no side constraints}.
The advantage of this method is that we can efficiently solve each
of the\emph{ }convex games with convergence guarantee using the best-response
algorithms introduced in Sec. \ref{sub:Game-with-exogenous} for the
game $\mathcal{G}_{\pi}(\mathcal{X},\,{\boldsymbol{{\theta}}})$ with
exogenous price; the disadvantage is that, to recover the solution
of the original game $\mathcal{G}(\mathcal{X},\,{\boldsymbol{{\theta}}})$,
we have to solve a (possibly infinite) number of convex games. However,
it is important to remark from the outset that this potential drawback
is greatly mitigated by the fact that, as we discuss shortly, (i)
one only needs to solve these games inaccurately; (ii) the (inaccurate)
solution of the NEPs usually requires little computational effort;
and (iii) in practice, a fairly accurate solution of the original
game $\mathcal{G}(\mathcal{X},\,{\boldsymbol{{\theta}}})$ is obtained
after the solution of a limited number of games in the sequence.

Before introducing the formal description of the algorithm, let us
begin with some informal observations and intermediate results motivating
how the sequence of convex games is built; the mathematical details
can be found in Appendix \ref{app:Convergence-of-Best-Response}.
At the basis of our analysis there are two results, namely: i) an
equivalence (under some conditions) between the game $\mathcal{G}(\mathcal{X},\,{\boldsymbol{{\theta}}})$
and the VI$(\mathcal{Z}_{t},\,\boldsymbol{{\Psi}})$ introduced in
(\ref{eq:KKT_Game_Gt}); and ii) the reformulation of the VI$(\mathcal{Z}_{t},\,\boldsymbol{{\Psi}})$
as a \emph{convex} game with \emph{no side constraint}. The former
connection, which is made formal in Lemma \ref{Lemma_G_VI_equivalence}
below, allows us to remove side constraints from the game $\mathcal{G}(\mathcal{X},\,{\boldsymbol{{\theta}}})$,
whereas the latter, given in Lemma \ref{Lemma_VI_game_noprox} below,
paves the way to the use of best-response algorithms for\emph{ convex}
games with \emph{no side constraints, }as introduced in Sec. \ref{sub:Game-with-exogenous}.

\begin{lemma} \label{Lemma_G_VI_equivalence} Given the game $\mathcal{G}(\mathcal{X},\,{\boldsymbol{{\theta}}})$,
suppose that there exists some $t>\lambda^{\max}$ such that the matrix
$\mathbf{A}(\mathbf{x},\,\boldsymbol{{\lambda}},\,{\pi}_{t})$ in
(\ref{eq:def_A_matrix}) is positive definite for all $\mathbf{x}\in\mathcal{Y}$,
$\boldsymbol{{\lambda}}\in[0,\,\lambda^{\max}]^{Q}$, and $\pi_{t}\in\mathcal{S}_{t}$.
Then $\mathcal{G}(\mathcal{X},\,{\boldsymbol{{\theta}}})$ is equivalent
to the VI$(\mathcal{Z}_{t},\,\boldsymbol{{\Psi}})$, which always
has a solution. The equivalence is in the following sense: for any
solution $\left(\mathbf{x}^{\text{{VI}}},\,\boldsymbol{{\lambda}}^{\text{{VI}}},\,{\pi}_{t}^{\text{{VI}}}\right)\in\mathcal{Z}_{t}$
of the VI, the tuple $\left(\mathbf{x}^{\text{{VI}}},\,{\pi}_{t}^{\text{{VI}}}\right)$
is a NE of $\mathcal{G}(\mathcal{X},\,{\boldsymbol{{\theta}}})$;
conversely, the game $\mathcal{G}(\mathcal{X},\,{\boldsymbol{{\theta}}})$
has a NE $(\mathbf{x}^{\star},{\pi}_{t}^{\star})$, and for any such
a NE there exist multipliers ${\boldsymbol{{\lambda}}}^{\star}\in[0,\,\lambda^{\max}]^{Q}$
associated with the nonconvex constraints $\{I_{q}(\mathbf{x}_{q}^{\star}),\, q=1,\ldots,Q\}$
such that $\left(\mathbf{x}^{\star},\,\boldsymbol{{\lambda}}^{\star},\,{\pi}_{t}^{\star}\right)$
is a solution of the $\text{{VI}}(\mathcal{Z}_{t},\,\boldsymbol{{\Psi}})$.\end{lemma}

Sufficient conditions for $\mathbf{A}(\mathbf{x},\,\boldsymbol{{\lambda}},\,{\pi}_{t})$
to be positive definite along with their physical interpretation are
given in Sec. \ref{remark_cond} (cf. Corollary \ref{corollary_sf_cond_uniqueness_NE}).
Under conditions of Lemma \ref{Lemma_G_VI_equivalence}, one can solve
the $\text{{VI}}(\mathcal{Z}_{t},\,\boldsymbol{{\Psi}})$ and obtain
the NE of the original game $\mathcal{G}(\mathcal{X},\,{\boldsymbol{{\theta}}})$.
Since we are interested in using best-response algorithms as those
developed in Sec. \ref{sub:Game-with-exogenous} for games with exogenous
price and no side constraints, we rewrite next the $\text{{VI}}(\mathcal{Z}_{t},\,\boldsymbol{{\Psi}})$
as a game, and then use best-response algorithms to solve that game.
More formally, let us introduce the following game with \emph{no side
constraints} wherein the players, anticipating rivals' strategies,
solve\emph{
\begin{equation}
\begin{array}{clc}
(i):\quad\,\,\,{\displaystyle {\operatornamewithlimits{minimize}_{\mathbf{x}_{q}\in{\mathcal{Y}}_{q}}}} & \mathcal{L}_{q}\left((\mathbf{x}_{q},\lambda_{q}),\,\mathbf{x}_{-q},\pi_{t}\right), & q=1,\ldots,Q\medskip\\
(\mbox{\mbox{ii}}):\quad\,\,{\displaystyle {\operatornamewithlimits{minimize}_{\lambda_{q}\in[0,\,\lambda^{\max}]}}} & -\lambda_{q}\cdot I_{q}(\mathbf{x}_{q}), & \,\,\, q=1,\ldots,Q,\medskip\\
(iii):\quad{\displaystyle {\operatornamewithlimits{minimize}_{\pi_{t}\in\mathcal{S}_{t}}}} & -\pi_{t}\cdot I(\mathbf{x}).
\end{array}\label{eq:G_VI_augmented}
\end{equation}
}The following connection holds between the above game and the $\text{{VI}}(\mathcal{Z}_{t},\,\boldsymbol{{\Psi}})$.

\begin{lemma} \label{Lemma_VI_game_noprox}Under the setting of Lemma
\ref{Lemma_G_VI_equivalence}, the $\text{{VI}}(\mathcal{Z}_{t},\,\boldsymbol{{\Psi}})$
is equivalent to the game in (\ref{eq:G_VI_augmented}), which always
admits a NE. \end{lemma}

Note that the game in (\ref{eq:G_VI_augmented}) is composed of $2Q+1$
players. The first $Q$ players in (i) correspond to the players of
the original game $\mathcal{G}(\mathcal{X},\,{\boldsymbol{{\theta}}})$$-$the
SUs in the system$-$that now optimize a different cost function,
which is the ``Lagrangian'' function associated with their original
cost functions in (\ref{eq:player q transformed 1}), for a given
set of price $\pi_{t}$ and multiplies $\boldsymbol{{\lambda}}$.
In addition to the $Q$ SUs, there are $Q+1$ more players solving
problems (ii) and (iii); they act as virtual players who aim to compute
the optimal multipliers $\lambda_{q}$'s associated with the nonconvex
local interference constraints $\{I_{q}(\mathbf{x}_{q}),\, q=1,\ldots,Q\}$
and the optimal price $\pi_{t}$, respectively. By introducing these
virtual players, the original game $\mathcal{G}(\mathcal{X},\,{\boldsymbol{{\theta}}})$
can be transformed (under the setting of Lemma \ref{Lemma_G_VI_equivalence})
into the desired (compact and) \emph{convex} game with only \emph{local
constraints}, which paves the way to the design of best-response algorithms
for the game $\mathcal{G}(\mathcal{X},\,{\boldsymbol{{\theta}}})$. 

We proved in Appendix \ref{sec:Convergence-of-Asynchronous_BR_local_constraints}
that the best-response algorithms introduced in Sec. \ref{sub:Game-with-exogenous}
converge under conditions implying the uniqueness of individual player's
optimization problems. 
The game in the form (\ref{eq:G_VI_augmented}) however may never
satisfy such conditions; indeed, the linear programming problems in
(ii) and (iii) have multiple optimal solutions whenever some $I_{q}(\mathbf{x}_{q})=0$
or $I(\mathbf{x})=0$. To overcome this issue, we follow a similar
idea as in Step 1 of Sec. \ref{sub:Existence-and-uniqueness} and
introduce in (ii) and (iii) of (\ref{eq:G_VI_augmented}) a proximal-based
regularization of the $\lambda$-variables and price $\pi_{t}$, so
that the resulting modified optimization problems become strongly
convex. Given the center of the regularization of the $\lambda$-variables,
say $\boldsymbol{{\lambda}}^{0}\triangleq(\lambda_{q}^{0})_{q=1}^{Q}$,
and the price $\pi_{t}$, say $\pi_{t}^{0}$, and the proximal gain
$\alpha>0$, the regularized version of the game in (\ref{eq:G_VI_augmented}),
denoted by $\mathcal{G}_{t}(\mathcal{X},\boldsymbol{{\theta}},\boldsymbol{{\lambda}}^{0},\pi_{t}^{0})$
is the following. \vspace{0.5cm}

\hspace{-0.6cm}%
\framebox{\begin{minipage}[t]{0.98\columnwidth}%
\textbf{Game} $\mathcal{G}_{t}(\mathcal{X},\boldsymbol{{\theta}},\boldsymbol{{\lambda}}^{0},\pi_{t}^{0})$.
Anticipating rivals' strategies and given $\boldsymbol{{\lambda}}^{0}\triangleq(\lambda_{q}^{0})_{q=1}^{Q}$,
$\pi_{t}^{0}$, and $\alpha>0$, the $2Q+1$ players solve the following
optimization problems: 
\begin{equation}
\begin{array}{clc}
{\displaystyle {\operatornamewithlimits{minimize}_{\mathbf{x}_{q}\in{\mathcal{Y}}_{q}}}} & \mathcal{L}_{q}\left((\mathbf{x}_{q},\lambda_{q}),\,\mathbf{x}_{-q},\pi_{t}\right), & q=1,\ldots,Q\medskip\\
{\displaystyle {\operatornamewithlimits{minimize}_{\lambda_{q}\in[0,\,\lambda^{\max}]}}} & -\lambda_{q}\cdot I_{q}(\mathbf{x}_{q})+\dfrac{{\alpha}}{2}\,\left(\lambda_{q}-\lambda_{q}^{0}\right)^{2}, & \quad\quad\,\,\, q=1,\ldots,Q,\medskip\\
{\displaystyle {\operatornamewithlimits{minimize}_{\pi_{t}\in\mathcal{S}_{t}}}} & -\pi_{t}\cdot I(\mathbf{x})+\dfrac{{\alpha}}{2}\,\left(\pi_{t}-\pi_{t}^{0}\right)^{2}
\end{array}\label{eq:G_prox}
\end{equation}
\end{minipage}}\vspace{0.5cm}

The main (desired) property of game $\mathcal{G}_{t}(\mathcal{X},\boldsymbol{{\theta}},\boldsymbol{{\lambda}}^{0},\pi_{t}^{0})$
is that, under the setting of Lemma \ref{Lemma_G_VI_equivalence},
the NE is unique and it can be computed with convergence guarantee
using best-response algorithms as those introduced in Sec. \ref{sub:Game-with-exogenous}
(we make formal this statement shortly). Nice as it is, this result
would be of no practical interest if we were not able to connect the
solutions of $\mathcal{G}_{t}(\mathcal{X},\boldsymbol{{\theta}},\boldsymbol{{\lambda}}^{0},\pi_{t}^{0})$
with those of the game in (\ref{eq:G_VI_augmented}) and thus the
original game $\mathcal{G}(\mathcal{X},\,{\boldsymbol{{\theta}}})$.
In fact, the solution of $\mathcal{G}_{t}(\mathcal{X},\boldsymbol{{\theta}},\boldsymbol{{\lambda}}^{0},\pi_{t}^{0})$
and (\ref{eq:G_VI_augmented}) are in general different but, nevertheless,
there exists a connection between them, as stated in the following
lemma. 

\begin{lemma} \label{Lemma_NE_VI_game_prox-Game}Under the setting
of Lemma \ref{Lemma_G_VI_equivalence}, a tuple $(\mathbf{x}^{\star},\boldsymbol{{\lambda}}^{\star},\pi_{t}^{\star})$
is a NE of the game in (\ref{eq:G_VI_augmented}) if and only if it
is a NE of the game $\mathcal{G}_{t}(\mathcal{X},\boldsymbol{{\theta}},\boldsymbol{{\lambda}}^{\star},\pi_{t}^{\star})$.
Therefore, such a $(\mathbf{x}^{\star},\pi_{t}^{\star})$ is a NE
of the original game $\mathcal{G}(\mathcal{X},\boldsymbol{{\theta}})$.\end{lemma}

Providing the relationship between $\mathcal{G}(\mathcal{X},\boldsymbol{{\theta}})$,
the game in (\ref{eq:G_VI_augmented}), and $\mathcal{G}_{t}(\mathcal{X},\boldsymbol{{\theta}},\boldsymbol{{\lambda}}^{0},\pi_{t}^{0})$,
Lemma \ref{Lemma_NE_VI_game_prox-Game} opens the way to the design
of best-response algorithms that solve the original game $\mathcal{G}(\mathcal{X},\boldsymbol{{\theta}})$:
instead of solving $\mathcal{G}(\mathcal{X},\boldsymbol{{\theta}})$
directly, starting from an arbitrary regularization tuple $(\boldsymbol{{\lambda}}^{0},\pi_{t}^{0})>\mathbf{0}$,
one can solve the sequence of games $\mathcal{G}_{t}(\mathcal{X},\boldsymbol{{\theta}},\boldsymbol{{\lambda}}^{0},\pi_{t}^{0})\rightarrow\cdots\rightarrow\mathcal{G}_{t}(\mathcal{X},\boldsymbol{{\theta}},\boldsymbol{{\lambda}}^{n},\pi_{t}^{n})\rightarrow\cdots$,
where the center $(\boldsymbol{{\lambda}}^{n},\pi_{t}^{n})$ of the
regularization of the game at stage $n$ is just the $(\lambda,\pi_{t})$-component
of the (unique) NE of the game $\mathcal{G}_{t}(\mathcal{X},\boldsymbol{{\theta}},\boldsymbol{{\lambda}}^{n-1},\pi_{t}^{n-1})$
in the previous stage. If this procedure converges, it must converge
to a tuple $(\mathbf{x}^{\star},\boldsymbol{{\lambda}}^{\star},\pi_{t}^{\star})$
that necessarily is a NE of the game $\mathcal{G}_{t}(\mathcal{X},\boldsymbol{{\theta}},\boldsymbol{{\lambda}}^{\star},\pi_{t}^{\star})$,
which implies by Lemma \ref{Lemma_NE_VI_game_prox-Game} that $(\mathbf{x}^{\star},\boldsymbol{{\lambda}}^{\star})$
is also a NE of the original game $\mathcal{G}(\mathcal{X},\boldsymbol{{\theta}})$.
A flow-chart with the connection of all these games along with an
informal description of the above ideas is given in Figure \ref{Fig_GT-VI_eq}.
\begin{figure}[H]
\center\includegraphics[height=9.5cm]{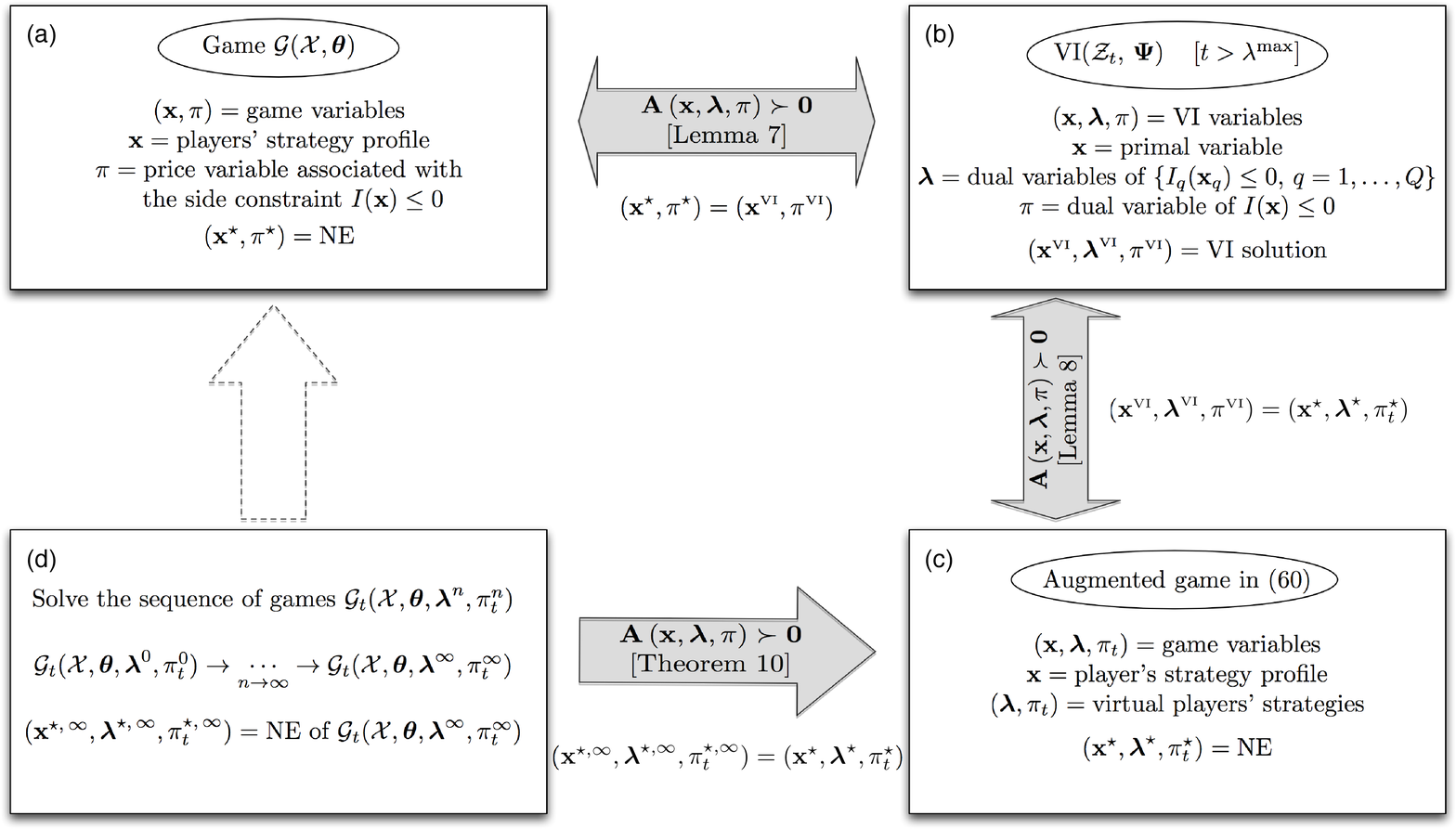}\vspace{-.2cm}

{\footnotesize \caption{{\footnotesize Connection among $\mathcal{G}(\mathcal{X},\boldsymbol{{\theta}})$,
VI$(\mathcal{Z}_{t},\boldsymbol{{\Psi}})$ , and the sequence of games
$\mathcal{G}_{t}(\mathcal{X},\boldsymbol{{\theta}},\boldsymbol{{\lambda}}^{n},\pi_{t}^{n})$.
Under the setting of Lemma \ref{Lemma_G_VI_equivalence}, we have
the following: i) $\mathcal{G}(\mathcal{X},\boldsymbol{{\theta}})$
in (a) is equivalent to the ``augmented'' VI$(\mathcal{Z}_{t},\boldsymbol{{\Psi}})$
in (b), where the local interference constraints $\{I_{q}(\mathbf{x}_{q}),\,\, q=1,\dots,Q\}$
are ``relaxed'' by introducing the multipliers $\boldsymbol{{\lambda}}\triangleq(\lambda_{q})_{q=1}^{Q}$
and $\pi$ is a variable of the VI; ii) the VI$(\mathcal{Z}_{t},\boldsymbol{{\Psi}})$
can be interpreted as a (compact) }\emph{\footnotesize convex}{\footnotesize{}
``augmented'' game with }\emph{\footnotesize no side constraints
}{\footnotesize as represented in (c) {[}see (\ref{eq:G_VI_augmented}){]},
where there are $Q$ real players, the SUs, and $Q+1$ virtual players
who aim to optimize the multipliers $\lambda_{q}$'s and the price
variable $\pi_{t}$; iii) a NE of the augmented game (\ref{eq:G_VI_augmented}),
and thus the original game $\mathcal{G}(\mathcal{X},\boldsymbol{{\theta}})$,
is computed via best-response algorithms solving the sequence of regularized
convex games with no side constraints $\mathcal{G}_{t}(\mathcal{X},\boldsymbol{{\theta}},\boldsymbol{{\lambda}}^{0},\pi_{t}^{0})\rightarrow\cdots\rightarrow\mathcal{G}_{t}(\mathcal{X},\boldsymbol{{\theta}},\boldsymbol{{\lambda}}^{\infty},\pi_{t}^{\infty})$
as shown in (d).}{\small{} }\label{Fig_GT-VI_eq}}
}
\end{figure}

A formal description of the above solution method is given in Algorithm
\ref{alg:PDA_GT_interpretation} below, which provides the desired
best-response based scheme solving the game $\mathcal{G}(\mathcal{X},\boldsymbol{{\theta}})$;
the convergence conditions are given in Theorem \ref{ProxDecAlg_viaGT_conv_theo}.
In the algorithm we use the following notation: given $\left(\boldsymbol{{\lambda}}^{n},\pi_{t}^{n}\right)$,
we denote by $(\mathbf{x}^{\star}(\boldsymbol{{\lambda}}^{n},\pi_{t}^{n}),\boldsymbol{{\lambda}}^{\star}(\boldsymbol{{\lambda}}^{n},\pi_{t}^{n}),$
$\pi_{t}^{\star}(\boldsymbol{{\lambda}}^{n},\pi_{t}^{n}))$ the NE
tuple of the game $\mathcal{G}_{t}(\mathcal{X},\boldsymbol{{\theta}},\boldsymbol{{\lambda}}^{n},\pi_{t}^{n})$,
where we make explicit the dependence on the regularization offset
$\left(\boldsymbol{{\lambda}}^{n},\pi_{t}^{n}\right)$.

\begin{algo}{Best-Response Algorithm for $\mathcal{G}(\mathcal{X},\boldsymbol{{\theta}})$}
S\texttt{$\mbox{(S.0)}:$} Choose any tuple $\left(\boldsymbol{{\lambda}}^{0},\pi_{t}^{0}\right)>\mathbf{0}$,
with $\boldsymbol{{\lambda}}^{0}\triangleq(\lambda_{q}^{0})_{q=1}^{Q}$,
and some $\epsilon\in(0,1)$; set $n=0$.

\texttt{$\mbox{(S.1)}:$} \texttt{If} $\left(\mathbf{x}^{\star}(\boldsymbol{{\lambda}}^{n},\pi_{t}^{n}),\,\boldsymbol{{\lambda}}^{\star}(\boldsymbol{{\lambda}}^{n},\pi_{t}^{n}),\,\pi_{t}^{\star}(\boldsymbol{{\lambda}}^{n},\pi_{t}^{n})\right)$
satisfies a suitable termination criterion: STOP.

\texttt{$\mbox{(S.2)}:$} Solve the game $\mathcal{G}_{t}(\mathcal{X},\boldsymbol{{\theta}},\boldsymbol{{\lambda}}^{n},\pi_{t}^{n})$;
let $\left(\mathbf{x}^{\star}(\boldsymbol{{\lambda}}^{n},\pi_{t}^{n}),\,\boldsymbol{{\lambda}}^{\star}(\boldsymbol{{\lambda}}^{n},\pi_{t}^{n}),\,\pi_{t}^{\star}(\boldsymbol{{\lambda}}^{n},\pi_{t}^{n})\right)$
be the NE. 

\texttt{$\mbox{(S.3)}:$} Update the center of the regularization:
\begin{equation}
\begin{array}{l}
\lambda_{q}^{n+1}\triangleq(1-\epsilon)\cdot\lambda_{q}^{n}+\epsilon\cdot\lambda_{q}^{\star}(\boldsymbol{{\lambda}}^{n},\pi_{t}^{n}),\quad q=1,\ldots,Q,\medskip\\
\pi_{t}^{n+1}\triangleq(1-\epsilon)\cdot\pi_{t}^{n}+\epsilon\cdot\pi_{t}^{\star}(\boldsymbol{{\lambda}}^{n},\pi_{t}^{n}).\vspace{-0.3cm}
\end{array}\label{eq:price_multipliers_update_in_Algorithm_prox}
\end{equation}
.

\texttt{$\mbox{(S.4)}:$} $n\leftarrow n+1$; go to \texttt{$\mbox{(S.1)}.$}
\label{alg:PDA_GT_interpretation}\end{algo}\medskip 

\begin{theorem}\label{ProxDecAlg_viaGT_conv_theo} Under the setting
of Lemma \ref{Lemma_G_VI_equivalence}, the sequence $\left\{ \left(\mathbf{x}^{\star}(\boldsymbol{{\lambda}}^{n},\pi_{t}^{n}),\,\pi_{t}^{\star}(\boldsymbol{{\lambda}}^{n},\pi_{t}^{n})\right)\right\} {}_{n=0}^{\infty}$
generated by Algorithm \ref{alg:PDA_GT_interpretation} globally converges
to a NE of $\mathcal{G}(\mathcal{X},\boldsymbol{{\theta}})$. \end{theorem}

\begin{proof}See Appendix \ref{app:Convergence-of-Best-Response}.
\end{proof}

It is interesting to observe that Algorithm \ref{alg:PDA_GT_interpretation}
converges under the same conditions introduced in Proposition \ref{Theo_Existence-and-uniqueness_NE_G}
and guaranteeing the uniqueness of the $x$-component of the NE of
$\mathcal{G}(\mathcal{X},\boldsymbol{{\theta}})$; we refer to Corollary
\ref{corollary_sf_cond_uniqueness_NE} and Sec. \ref{remark_cond}
for easier conditions to be checked as well as a detailed discussion
on their interpretation in terms of the system parameters. We discuss
next some practical implementation issues related to Algorithm \ref{alg:PDA_GT_interpretation}.

\subsubsection{Discussion on the implementation\label{sub:Implementation-issues_global_algorithm}}

Algorithm \ref{alg:PDA_GT_interpretation} is conceptually a double-loop
scheme wherein at each (outer) iteration $n$, given the current values
of the regularization parameters $(\boldsymbol{{\lambda}}^{n},\pi_{t}^{n})$,
the SUs solve the game $\mathcal{G}_{t}(\mathcal{X},\boldsymbol{{\theta}},\boldsymbol{{\lambda}}^{n},\pi_{t}^{n})$
(with $t>\lambda^{\max}$) {[}Step 2{]}, which requires an inner iterative
process. Once the NE of $\mathcal{G}_{t}(\mathcal{X},\boldsymbol{{\theta}},\boldsymbol{{\lambda}}^{n},\pi_{t}^{n})$
is reached, the regularization parameters $(\boldsymbol{{\lambda}}^{n},\pi_{t}^{n})$
are updated according to (\ref{eq:price_multipliers_update_in_Algorithm_prox})
{[}Step 3{]}, which represents the outer loop, and the new game $\mathcal{G}_{t}(\mathcal{X},\boldsymbol{{\theta}},\boldsymbol{{\lambda}}^{n+1},\pi_{t}^{n+1})$
is played again (if the convergence criterion in Step 1 is not met).
In practice, however, Algorithm \ref{alg:PDA_GT_interpretation} is
implementable as a single-loop scheme: the SUs play the game $\mathcal{G}_{t}(\mathcal{X},\boldsymbol{{\theta}},\boldsymbol{{\lambda}}^{n},\pi_{t}^{n})$,
wherein from ``time to time'' (more precisely, when a NE is reached
within the required accuracy) the objective functions of the virtual
players are changed by updating the regularization terms from $\frac{\alpha}{2}(\lambda_{q}-\lambda_{q}^{n})$
and $\frac{\alpha}{2}(\pi_{t}-\pi_{t}^{n})$ to $\frac{\alpha}{2}(\lambda_{q}-\lambda_{q}^{n+1})$
and $\frac{\alpha}{2}(\pi_{t}-\pi_{t}^{n+1})$, respectively. 

In order to implement the aforementioned single-scale scheme, the
following issues need to be addressed: 1) How to solve each inner
game $\mathcal{G}_{t}(\mathcal{X},\boldsymbol{{\theta}},\boldsymbol{{\lambda}}^{n},\pi_{t}^{n})$
via distributed best-response algorithms? 2) How to update the regularization
parameters in a distributed way? and 3) How to check the terminations
of the inner process in Step 2$-$the SUs have reached a NE of the
game $\mathcal{G}_{t}(\mathcal{X},\boldsymbol{{\theta}},\boldsymbol{{\lambda}}^{n},\pi_{t}^{n})$
within the desired accuracy? We provide an answer to these questions
next.

\subsubsection*{On the inner game and price/multipliers update {[}Steps 2 and 3{]}}

Capitalizing on the solution methods that we developed in Sec. \ref{sub:Game-with-exogenous}
for games with exogenous price and no side constraints, a natural
choice for computing a NE of each $\mathcal{G}_{t}(\mathcal{X},\boldsymbol{{\theta}},\boldsymbol{{\lambda}}^{n},\pi_{t}^{n})$
in Step 2 of Algorithm \ref{alg:PDA_GT_interpretation} is applying
those best-response asynchronous algorithms to $\mathcal{G}_{t}(\mathcal{X},\boldsymbol{{\theta}},\boldsymbol{{\lambda}}^{n},\pi_{t}^{n})$.
For instance, if a Jacobi scheme is chosen (cf. Algorithm \ref{async_best-response_algo}),
Algorithm \ref{alg:PDA_GT_interpretation} reduces to Algorithm \ref{algo_bi_level}
below, which sheds light on the signaling and complexity requirements
of the proposed class of algorithms. In Algorithm \ref{algo_bi_level}
we use the following notation: $\left(\mathbf{x}^{\star}(\overline{{\boldsymbol{{\lambda}}}},\overline{\pi}_{t}),\,\boldsymbol{{\lambda}}^{\star}(\overline{{\boldsymbol{{\lambda}}}},\overline{\pi}_{t}),\,\pi_{t}^{\star}(\overline{{\boldsymbol{{\lambda}}}},\overline{\pi}_{t})\right)$
denotes the NE of $\mathcal{G}_{t}(\mathcal{X},\boldsymbol{{\theta}},\overline{{\boldsymbol{{\lambda}}}},\overline{\pi}_{t})$,
and $[x]_{0}^{\lambda^{\max}}$ in (\ref{eq:Async_update_NE_sel})
is the Euclidean projection onto the interval $[0,\,\lambda^{\max}]$,
i.e., $ $ $[x]_{0}^{\lambda^{\max}}\triangleq\max(0,\min(x,\,\lambda^{\max}))$.

\begin{algo}{Jacobi Best-Response-Consensus Algorithm for $\mathcal{G}(\mathcal{X},\boldsymbol{{\theta}})$}
S\texttt{$\mbox{(S.0)}:$} Choose i) any arbitrary starting point
$(\mathbf{x}^{(0)},\boldsymbol{{\lambda}}^{(0)},\pi_{t}^{(0)})$,
with $\mathbf{x}^{(0)}\triangleq(\wh{\tau}_{q}^{(0)},\mathbf{p}_{q}^{(0)},P_{q}^{\text{{fa}}(0)})\in\mathcal{Y}$
and $(\boldsymbol{{\lambda}}^{(0)},\pi_{t}^{(0)})>\mathbf{0}$; ii)
any regularization tuple $\left(\overline{\boldsymbol{{\lambda}}},\overline{\pi}_{t}\right)>\mathbf{0}$
, and iii) some $\epsilon\in(0,1)$; set $n=0$.\\[1pt] \texttt{$\mbox{(\mbox{S.1}):}$}
\texttt{If} $\left(\mathbf{x}^{\star}(\overline{{\boldsymbol{{\lambda}}}},\overline{\pi}_{t}),\,\boldsymbol{{\lambda}}^{\star}(\overline{{\boldsymbol{{\lambda}}}},\overline{\pi}_{t}),\,\pi_{t}^{\star}(\overline{{\boldsymbol{{\lambda}}}},\overline{\pi}_{t})\right)$
satisfies a suitable termination criterion: STOP.\\[1pt] 
 \texttt{$\mbox{(\mbox{S.2}a):}$} Run a (vector) consensus algorithm
to locally compute the current values of $\dfrac{{1}}{Q}\,{\displaystyle {\sum_{q=1}^{Q}}}\,\dfrac{{\wh{\tau}_{q}^{(n)}}}{\sqrt{{f_{q}}}}$
and $I\left(\mathbf{x}^{(n)}\right)={\displaystyle {\sum_{q=1}^{Q}}}I_{q}(\mathbf{x}_{q}^{(n)})$
{[}cf. Algorithm \ref{consensus_algorithm}{]};\\[1pt]\texttt{$\mbox{(\mbox{S.2}b):}$
}Update the players' strategies simultaneously:\texttt{ }

\begin{equation}
\begin{array}{lll}
\mathbf{x}_{q}^{(n+1)}\in\underset{\mathbf{x}_{q}\in\mathcal{Y}_{q}}{\text{{argmin}}}\left\{ \mathcal{L}_{q}\left((\mathbf{x}_{q},\lambda_{q}^{(n)}),\,\mathbf{x}_{-q}^{(n)},\pi_{t}^{n}\right)\right\} , &  & \forall q=1,\ldots,Q\\
\\
\lambda_{q}^{(n+1)}=\left[\overline{\lambda}_{q}+\dfrac{I_{q}\left(\mathbf{x}_{q}^{(n)}\right)}{{\alpha}}\right]_{0}^{\lambda^{\max}}, &  & \forall q=1,\ldots,Q\\
\\
\pi_{t}^{(n+1)}=\left[\overline{\pi}_{t}+\dfrac{I\left(\mathbf{x}^{(n)}\right)}{{\alpha}}\right]_{0}^{\lambda^{\max}}.
\end{array}\label{eq:Async_update_NE_sel}
\end{equation}
\\[1pt]$\quad\,\,\,$ \texttt{$\mbox{(\mbox{S.3})}:$} If $\left(\mathbf{x}^{(n+1)},\,\boldsymbol{{\lambda}}^{(n+1)},\,\pi_{t}^{(n+1)}\right)$
is a NE of $\mathcal{G}_{t}(\mathcal{X},\boldsymbol{{\theta}},\overline{{\boldsymbol{{\lambda}}}},\overline{{\pi}}_{t})$,
then 

\hspace{1.5cm}1) update the regularization tuple $\left(\overline{\boldsymbol{{\lambda}}},\overline{\pi}_{t}\right)$:
\begin{equation}
\overline{{\lambda}}_{q}={\lambda}_{q}^{(n+1)},\,\,\forall q=1,\ldots,Q\quad\mbox{and}\quad\overline{\pi}_{t}=\pi_{t}^{(n+1)};\label{eq:step_size_and_centroid_update}
\end{equation}
\hspace{1.5cm}2) set $\left(\mathbf{x}^{\star}(\overline{{\boldsymbol{{\lambda}}}},\overline{\pi}_{t}),\,\boldsymbol{{\lambda}}^{\star}(\overline{{\boldsymbol{{\lambda}}}},\overline{\pi}_{t}),\,\pi_{t}^{\star}(\overline{{\boldsymbol{{\lambda}}}},\overline{\pi}_{t})\right)=\left(\mathbf{x}^{(n+1)},\,\boldsymbol{{\lambda}}^{(n+1)},\,\pi_{t}^{(n+1)}\right)$;

\hspace{1.5cm}3) $n\leftarrow n+1$ and return to \texttt{$\mbox{(\mbox{S.1})}$}. 

\hspace{1cm}else: $n\leftarrow n+1$ and return to \texttt{$\mbox{(\mbox{S.2}a)}$}.
\label{algo_bi_level} \end{algo}

The convergence analysis of the algorithm follows from that of Algorithm
\ref{alg:PDA_GT_interpretation} (the outer loop) and Algorithm \ref{async_best-response_algo}
(the inner loop) and thus is omitted. It is worth mentionig that Algorithm
\ref{algo_bi_level} converges under similar conditions obtained for
Algorithm \ref{async_best-response_algo}, provided that a sufficiently
large proximal gain $\alpha$ is chosen; this is not surprising, since
the core of Algorithm \ref{algo_bi_level} is the updating rule used
in Algorithm \ref{async_best-response_algo}, whose convergence conditions
imply those of the outer loop (cf. Theorem \ref{ProxDecAlg_viaGT_conv_theo}).
We refer to Sec. \ref{sub:Game-with-exogenous} for an interpretation
of these convergence conditions.

Algorithm \ref{algo_bi_level} is mainly composed of two-subroutines:
a consensus-based scheme {[}Step 2a{]} and a best-response update
{[}Step 2b{]}, both implemented locally by the SUs. More specifically,
the inner game $\mathcal{G}_{t}(\mathcal{X},\boldsymbol{{\theta}},\overline{{\boldsymbol{{\lambda}}}},\overline{\pi}_{t})$
is solved in a fairly distributed way by following a two-steps procedure.
Fist, in Step 2a, the SUs run a consensus algorithm to locally acquire
the global information required to perform the update of their sensing/transmission
variables as well as the multipliers ${\lambda}_{q}$'s and the price
$\pi_{t}$, which is represented by the average sensing time $(1/Q)\,\sum_{q=1}^{Q}\,({\wh{\tau}_{q}^{(n)}}/\sqrt{{f_{q}}})$
and the global level of interference $I(\mathbf{x}^{(n)})=\sum_{q=1}^{Q}I_{q}(\mathbf{x}_{q}^{(n)})$
generated at the primary receiver; this procedure requires an exchange
of information among neighboring nodes, as already discussed in Sec.
\ref{sub:Game-with-exogenous}, where we refer for details. Once the
aforementioned information is available at the secondary transmitters,
each SU $q$ \emph{locally} updates his own sensing/transmission strategy
$\mathbf{x}_{q}$ as well as the multiplier $\lambda_{q}$ and the
price $\pi_{t}$, according to (\ref{eq:Async_update_NE_sel}) {[}Step2b{]};
he just needs to measure the MUI experienced at his receiver and solve
his own optimization problem. Note that: i) the updates of the multipliers
$\lambda_{q}$'s and the price $\pi_{t}$ have an explicit closed
form expression, and thus are computationally inexpensive; and ii)
there is no need of a centralized authority for the optimization of
the price $\pi_{t}$, which is instead updated locally by each SU.

\subsubsection*{On the inner termination criterium {[}Step 2{]}}

The only issue left to discuss is how to check the termination criterion
of the inner process in Step 2 of Algorithm \ref{alg:PDA_GT_interpretation};
similar discussion applies to Algorithm \ref{algo_bi_level}. In practice,
Step 2 is terminated when the NE $\left(\mathbf{x}^{\star}(\boldsymbol{{\lambda}}^{n},\pi_{t}^{n}),\,\boldsymbol{{\lambda}}^{\star}(\boldsymbol{{\lambda}}^{n},\pi_{t}^{n}),\,\pi_{t}^{\star}(\boldsymbol{{\lambda}}^{n},\pi_{t}^{n})\right)$
of $\mathcal{G}_{t}(\mathcal{X},\boldsymbol{{\theta}},{\boldsymbol{{\lambda}}}^{n},\pi_{t}^{n})$
is reached within the prescribed accuracy,%
\footnote{Recall that, under the convergence conditions in Theorem \ref{ProxDecAlg_viaGT_conv_theo},
each $\mathcal{G}_{t}(\mathcal{X},\boldsymbol{{\theta}},\overline{{\boldsymbol{{\lambda}}}},\overline{\pi}_{t})$
has a unique equilibrium,%
} say $\varepsilon^{(n)}$, where we let $\varepsilon^{(n)}$ to depend
on the (outer) iteration index $n$. Stated in mathematical terms,
this means that the players leave Step 2 as soon as their current
strategy profile $(\mathbf{x},\,\boldsymbol{{\lambda}},\,\pi_{t})$
satisfies the following inequality: 
\begin{equation}
\left\Vert \left[\begin{array}{c}
\mathbf{x}\\
\boldsymbol{{\lambda}}\\
\pi_{t}
\end{array}\right]-\left[\begin{array}{c}
\mathbf{x}^{\star}(\boldsymbol{{\lambda}}^{n},\pi_{t}^{n})\\
\boldsymbol{{\lambda}}^{\star}(\boldsymbol{{\lambda}}^{n},\pi_{t}^{n})\\
\pi_{t}^{\star}(\boldsymbol{{\lambda}}^{n},\pi_{t}^{n})
\end{array}\right]\right\Vert ^{2}\leq\varepsilon^{(n)},\label{eq:error_bound}
\end{equation}
where $\left\Vert \bullet\right\Vert $ is any vector norm. Denoting
by $\mathbf{z}\triangleq(\mathbf{x},\,\boldsymbol{{\lambda}},\,\pi_{t})$
the players' strategy profile and by $\mathbf{S}_{\mathcal{G}_{t}}(\boldsymbol{{\lambda}}^{n},\pi_{t}^{n})$
the (unique) NE of $\mathcal{G}_{t}(\mathcal{X},\boldsymbol{{\theta}},{\boldsymbol{{\lambda}}}^{n},\pi_{t}^{n})$,
which depends on the regularization parameters $({\boldsymbol{{\lambda}}}^{n},\pi_{t}^{n})$,
the stopping criterium in (\ref{eq:error_bound}) can be equivalently
written as $\|\mathbf{z}-\mathbf{S}_{\mathcal{G}_{t}}(\boldsymbol{{\lambda}}^{n},\pi_{t}^{n})\|^{2}\leq\varepsilon^{(n)}$.
Using the above notation/terminology, Step 2 of Algorithm \ref{alg:PDA_GT_interpretation}
reads as 

\begin{center}
\texttt{$\mbox{(S.2a)}:$} Solve the game $\mathcal{G}_{t}(\mathcal{X},\boldsymbol{{\theta}},\boldsymbol{{\lambda}}^{n},\pi_{t}^{n})$
within the accuracy $\varepsilon^{(n)}$: find a $\mathbf{z}=(\mathbf{x},\,\boldsymbol{{\lambda}},\,\pi_{t})$
such that\vspace{-.3cm} 
\begin{equation}
\|\mathbf{z}-\mathbf{S}_{\mathcal{G}_{t}}(\boldsymbol{{\lambda}}^{n},\pi_{t}^{n})\|^{2}\leq\varepsilon^{(n)};\vspace{-.2cm}\label{eq:stopping_rule_Step2}
\end{equation}

\par\end{center}

\begin{center}
\texttt{$\mbox{(S.2b)}:$} Set $\left(\mathbf{x}^{\star}(\boldsymbol{{\lambda}}^{n},\pi_{t}^{n}),\,\boldsymbol{{\lambda}}^{\star}(\boldsymbol{{\lambda}}^{n},\pi_{t}^{n}),\,\pi_{t}^{\star}(\boldsymbol{{\lambda}}^{n},\pi_{t}^{n})\right)=\mathbf{z}$.\hspace{7.3cm} 
\par\end{center}

\noindent In general, the test in (\ref{eq:stopping_rule_Step2})
would require some coordination among the players; nevertheless, we
suggest next two simple distributed protocols to do that, building
on the error-bound analysis of VIs \cite[Ch. 6]{Facchinei-Pang_FVI03}. 

Observe preliminarily that an error bound on the distance of the current
strategy profile $\mathbf{z}$ from the NE $\mathbf{S}_{\mathcal{G}_{t}}(\boldsymbol{{\lambda}}^{n},\pi_{t}^{n})$
can be obtained by solving a convex (quadratic) problem (see, e.g.,
\cite[Prop. 6.3.1]{Facchinei-Pang_FVI03}, \cite[Prop. 6.3.7]{Facchinei-Pang_FVI03}).
Indeed, under the convergence conditions of Algorithm \ref{alg:PDA_GT_interpretation}
{[}cf. Theorem \ref{ProxDecAlg_viaGT_conv_theo}{]}, one can write
each game $\mathcal{G}_{t}(\mathcal{X},\boldsymbol{{\theta}},\boldsymbol{{\lambda}}^{n},\pi_{t}^{n})$
as a (strongly monotone) VI problem, for which  the following error
bound holds \cite[Prop. 6.3.1]{Facchinei-Pang_FVI03}: a (finite and
absolute) constant $\eta>0$%
\footnote{An explicit expression of $\eta$ can be obtained as a function of
the system parameters, based on \cite[Prop. 6.3.1]{Facchinei-Pang_FVI03}.%
} exists such that for every $\mathbf{z}$, 
\begin{equation}
\|\mathbf{z}-\mathbf{S}_{\mathcal{G}_{t}}(\boldsymbol{{\lambda}}^{n},\pi_{t}^{n})\|^{2}\,\leq\eta\,\|\boldsymbol{{\Psi}}_{\text{{nat}}}^{n}\left(\mathbf{z}\right)\|^{2},\label{eq:error_bound_1}
\end{equation}
with 
\begin{equation}
\boldsymbol{{\Psi}}_{\text{{nat}}}^{n}\left(\mathbf{z}\right)\triangleq\left(\begin{array}{c}
\left(\begin{array}{c}
\mathbf{x}_{q}-\Pi_{\mathcal{Y}_{q}}\left(\mathbf{x}_{q}-\nabla_{\mathbf{x}_{q}}\mathcal{L}_{q}{\displaystyle {\displaystyle ((\mathbf{x}_{q},\lambda_{q}),\,\mathbf{x}_{-q},\pi_{t})}}\right)\\
\lambda_{q}-\left[\lambda_{q}^{n}+\dfrac{{I_{q}{\displaystyle (\mathbf{x}_{q})}}}{\alpha}\right]_{0}^{\lambda^{\max}}
\end{array}\right)_{q=1}^{Q}\medskip\\
\pi_{t}-\left[\pi_{t}^{n}+\dfrac{{I{\displaystyle (\mathbf{x})}}}{\alpha}\right]_{0}^{\lambda^{\max}}
\end{array}\right)\triangleq\left(\begin{array}{c}
\left(\left[\boldsymbol{{\Psi}}_{\text{{nat}}}^{n}\left(\mathbf{z}\right)\right]_{q}\right)_{q=1}^{Q}\medskip\\
\left[\boldsymbol{{\Psi}}_{\text{{nat}}}^{n}\left(\mathbf{z}\right)\right]_{Q+1}
\end{array}\right),\label{eq:natural map}
\end{equation}
and $\Pi_{\mathcal{Y}_{q}}\left(\mathbf{a}\right)$ denoting the Euclidean
projection of the vector $\mathbf{a}$ onto the closed and convex
set $\mathcal{Y}_{q}$, where in (\ref{eq:natural map}) we made explicit
the partition of $\boldsymbol{{\Psi}}_{\text{{nat}}}^{n}\left(\mathbf{z}\right)$
in $Q+1$ (vector) components, $\left([\boldsymbol{{\Psi}}_{\text{{nat}}}^{n}\left(\mathbf{z}\right)]_{q}\right)_{q=1}^{Q+1}$,
each of the first $Q$ being associated with one different player
$q$. The important result here is that each SU $q$ can compute his
own component $[\boldsymbol{{\Psi}}_{\text{{nat}}}^{n}\left(\mathbf{z}\right)]_{q}$
(as well as the last component $[\boldsymbol{{\Psi}}_{\text{{nat}}}^{n}\left(\mathbf{z}\right)]_{Q+1}$)\emph{
efficiently} and \emph{locally}. Indeed, capitalizing on the information
already acquired for the computation of the best-response, he just
needs to solve a quadratic programming {[}corresponding to the evaluation
of the projection $\Pi_{\mathcal{Y}_{q}}\left(\bullet\right)${]},
for which no extra signaling/coordination with the others is required. 

A simple application of the error bound (\ref{eq:error_bound_1})
for the test in (\ref{eq:stopping_rule_Step2}) is to let each SU
$q$ to choose a local termination error $\varepsilon_{q}\leq\eta\cdot\varepsilon/Q$,
with $\varepsilon=\varepsilon^{(n)}$ being the desired accuracy in
(\ref{eq:error_bound_1}), and perform the termination criterion $\left\Vert [\boldsymbol{{\Psi}}_{\text{{nat}}}^{n}\left(\mathbf{z}\right)]_{q}\right\Vert ^{2}+\left\Vert [\boldsymbol{{\Psi}}_{\text{{nat}}}^{n}\left(\mathbf{z}\right)]_{Q+1}\right\Vert ^{2}\leq\varepsilon_{q}$;
which is locally implementable, provided that an estimate of the absolute
constant $\eta$ in (\ref{eq:error_bound_1}) and the number of the
active SUs can be preliminary obtained. 

When this information is not available, one can consider a variation
(inexact version) of Algorithm \ref{alg:PDA_GT_interpretation}. Instead
of solving\emph{ }each game $\mathcal{G}_{t}(\mathcal{X},\boldsymbol{{\theta}},\boldsymbol{{\lambda}}^{n},\pi_{t}^{n})$
\emph{exactly}, the players compute at every stage $n$ only an\emph{
approximated solution} of $\mathcal{G}_{t}(\mathcal{X},\boldsymbol{{\theta}},\boldsymbol{{\lambda}}^{n},\pi_{t}^{n})$
that becomes tighter and tighter as the iteration in $n$ proceeds.
Stated in mathematical terms, we have that the sub-iterations in Step
2a are terminated according to a prescribed error sequence $\{\varepsilon^{(n)}\}_{n}$
that progressively becomes tighter as the iteration in $n$ proceeds.
For instance, a suitable termination sequence in (\ref{eq:stopping_rule_Step2})
is any $\{\varepsilon^{(n)}\}_{n}\subset[0,\infty)$ satisfying $\sum_{n=1}^{\infty}\varepsilon^{(n)}<\infty$;
since the latter condition implies $\varepsilon^{(n)}\downarrow0$,
when the iterations $n$ progress the NE $\mathbf{S}_{\mathcal{G}_{t}}(\boldsymbol{{\lambda}}^{n},\pi_{t}^{n})$
will be estimated with an increasing accuracy. One can show that the
aforementioned inexact version of Algorithm \ref{alg:PDA_GT_interpretation}
converges under the same conditions given in Theorem \ref{ProxDecAlg_viaGT_conv_theo};
we omit the details because of the space limitation, and we refer
to \cite{ScutariPalomarFacchineiPang-Monotone_bookCh} for a similar
approach valid for \emph{convex} games. The termination protocol for
the inexact version of Algorithm \ref{alg:PDA_GT_interpretation}
is then the following. Each player $q$ choses preliminarily a suitable
local termination sequence $\{\varepsilon_{q}^{(n)}\}_{n}\subset[0,\infty)$
such that $\sum_{n=1}^{\infty}\varepsilon_{q}^{(n)}<\infty$; the
termination criterion of each player $q$ becomes then $\left\Vert [\boldsymbol{{\Psi}}_{\text{{nat}}}^{n}\left(\mathbf{z}\right)]_{q}\right\Vert ^{2}+\left\Vert [\boldsymbol{{\Psi}}_{\text{{nat}}}^{n}\left(\mathbf{z}\right)]_{Q+1}\right\Vert ^{2}\leq\varepsilon_{q}^{(n)}$,
which can be locally implemented. Once the desired local accuracy
is reached by all the players, they can all update the center of their
regularization, according to (\ref{eq:price_multipliers_update_in_Algorithm_prox}).
This protocol guarantees that the resulting sequence $\varepsilon^{(n)}\triangleq\sum_{q=1}^{Q}\varepsilon_{q}^{(n)}$
in (\ref{eq:stopping_rule_Step2}) will satisfy the required condition
$\sum_{n=1}^{\infty}\varepsilon^{(n)}<\infty$, without the need of
any information exchange among the players. 

The last issue to address for a practical implementation of the two
protocols above is to understand how the players can know that also
the others have reached the desired termination criterion. This can
be done by exchanging one bit of information; otherwise each user
can just update his regularization after experiencing no changes in
$\left\Vert [\boldsymbol{{\Psi}}_{\text{{nat}}}^{n}\left(\mathbf{z}\right)]_{q}\right\Vert $
and $\left\Vert [\boldsymbol{{\Psi}}_{\text{{nat}}}^{n}\left(\mathbf{z}\right)]_{Q+1}\right\Vert $
for a prescribed number of iterations. 

Two last comments about the proposed class of algorithms solving $\mathcal{G}(\mathcal{X},\boldsymbol{{\theta}})$
are in order. To obtain decentralize algorithms even in the presence
of global (nonconvex) interference constraints, we have introduced
multipliers and relaxed the global constraints. As a side effect of
the proposed approach, we have that global interference constraints
are met only at the equilibrium of the game; implying that during
the iterations of the algorithms they might not be satisfied. This
issue is alleviated in practice by a fast convergent behavior of the
proposed algorithms, as shown in Sec. \ref{sec:Numerical-Results}.
Note that this issue is quite common to many power control algorithms
subject to QoS or coupling interference constraints (see, e.g., \cite{Chiang_Hande_Lan_Tan_book_PC}
and references therein). Finally, we wish to point out that when the
sufficient conditions for the convergence of the proposed algorithms
are not satisfied, still we can claim some optimality property for
the proposed algorithms, namely: every limit point of the sequence
generated by the our algorithms is a \emph{quasi-NE} of the game under
consideration; the analysis of such relaxed equilibrium concept along
with its main properties is addressed in the companion paper \cite{Pang-Scutari-NNConvex_PI}.

\subsection{A bird's-eye view }

In the previous three sections we proposed several distributed algorithms
to solve the general game $\mathcal{G}(\mathcal{X},\boldsymbol{{\theta}})$
and its special cases. The algorithms differ from computational complexity,
performance, and level of signaling among the SUs; making them applicable
to several different scenarios. It is useful to summarize the results
obtained so far, showing that, in spite of apparent diversities, all
the algorithms belong to a same family; Figure \ref{fig1_Roadmap}
provides the roadmap of the proposed distributed solution methods
along with the signaling required for their implementation.

\section{Numerical Results\label{sec:Numerical-Results} \vspace{-0.3cm}}

In this section, we provide some numerical results to illustrate our
theoretical findings. More specifically, we first compare the performance
of our games with those of state-of-the-art decentralized \cite{Pang-Scutari-Palomar-Facchinei_SP_10}
and centralized \cite{SchmidtShiBerryHonigUtschick-SPMag} schemes
proposed in the literature for similar problems; such schemes \emph{do
not perform any sensing optimization} using thus all the frame length
for the transmission, and the QoS of the PUs is preserved by imposing
(deterministic) interference constraints (we properly modified the
algorithms in \cite{SchmidtShiBerryHonigUtschick-SPMag} to include
the interference constraints in the feasible set of the optimization
problem). Interestingly, the proposed design of CR systems based on
the distributed\emph{ }joint optimization of the sensing and transmission
strategies is shown to outperform \emph{both centralized and decentralized
current CR designs}, which validates our new formulation. Then, we
provide an example of signaling/performance trade-off, showing the
throughput gains achievable by the SUs if the sensing time is included
in the optimization. Finally, we focus on the convergence properties
of the proposed algorithms. 

\begin{figure}[t]
\vspace{-2cm}\center\includegraphics[height=11cm]{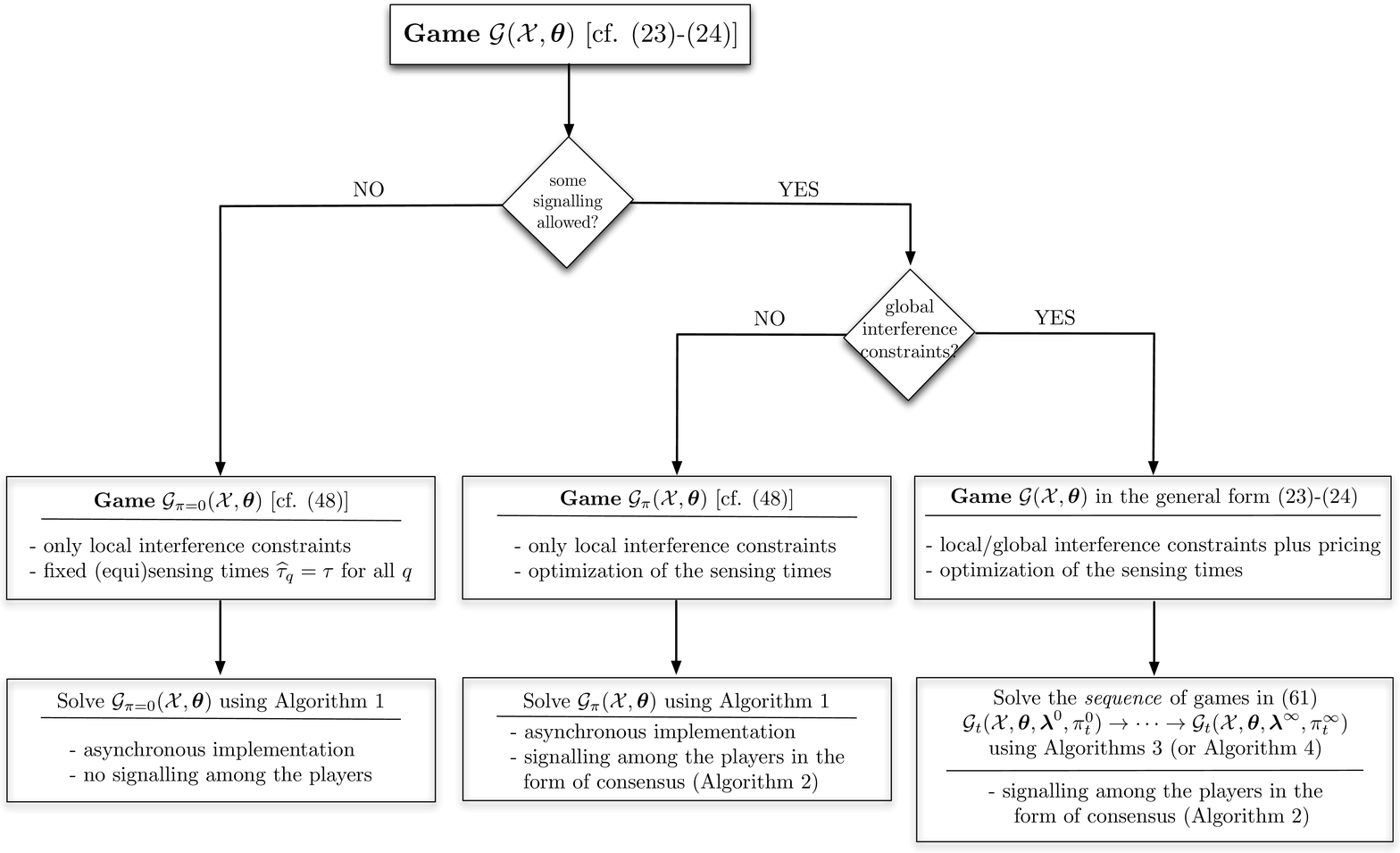}\vspace{-.2cm}

{\footnotesize \caption{{\footnotesize Road-map of the proposed algorithms solving $\mathcal{G}(\mathcal{X},\boldsymbol{{\theta}})$
and its special cases along with the resulting signalling/optimization
tradeoff. }{\small \label{fig1_Roadmap}}}
}
\end{figure}
\medskip{}

\noindent \textbf{Example \#1: Comparison with state-of-the-art algorithms.}
In Fig. \ref{Fig_comparison}, we compare the performance achievable
by the proposed joint optimization of the sensing and the transmission
strategies with those achievable using the sum-rate NUM-based approach
subject to interference constraints \cite{SchmidtShiBerryHonigUtschick-SPMag}
and the game theoretical formulation in \cite{Pang-Scutari-Palomar-Facchinei_SP_10}.
More specifically, we plot the (\%) ratio $(SR_{QE}-SR)/SR$ versus
the (normalized) interference constraint bound $P/I^{\text{{max}}}$
($P_{q}=P_{r}=P$ for all $q\neq r$ and $I_{q}^{\text{{max}}}=I^{\text{{max}}}$
for all $q$), for different values of the SNR detection $\texttt{snr}_{d}=\sigma_{I_{q,k}}^{2}/\sigma_{q,k}^{2}$,
where $SR_{QE}$ is the sum-throughput achievable at the (Q)NE of
the game $\mathcal{G}_{\pi=0}(\mathcal{X},\boldsymbol{{\theta}})$
(local interference constraints only), whereas $SR$ is either the
sum-rate achievable using the scheme in \cite{SchmidtShiBerryHonigUtschick-SPMag}
(red line curves) or the sum-rate at the NE of the game in \cite{Pang-Scutari-Palomar-Facchinei_SP_10}
(black line curves). We simulated a hierarchical CR network composed
of two PUs (the base stations of two cells) and ten SUs, randomly
distributed in the cells. The (cross-)channels among the secondary
links and between the primary and the secondary links are FIR filters
of order $L=10$, where each tap has variance equal to $1/L^{2}$;
the available bandwidth is divided in $N=1024$ subchannels. From
Fig. \ref{Fig_comparison}, we clearly see that the proposed joint
optimization of the sensing and transmission parameters yields a considerable
performance improvement over the current state-of-the-art CR \emph{centralized
and decentralized} designs, especially when the interference constraints
are stringent. 

\smallskip{}

\noindent \textbf{Example \#2: Sensing time optimization. }Fig. \ref{Fig_sens_through_tradeoff}
shows an example of the achievable throughput of the SUs when the
sensing time is included in the optimization. More specifically, in
the picture, we plot the (normalized) sum-throughput achieved at a
(Q)NE by one player of the game versus the (normalized) \emph{common}
sensing time, for different values of the (normalized) total interference
constraint (the setup is the same as in Fig. \ref{Fig_comparison}).
In the same figure, we plot also the sum-throughput achieved at the
(Q)NE of the game $\mathcal{G}_{\pi=0}(\mathcal{X},\boldsymbol{{\theta}})$
(square markers in the plot), where $c$ is set to $c=100$. According
to the picture, the following comments are in order. There exists
an optimal duration for the (common) sensing time at which the throughput
of each SU is maximized, implying that the SUs can achieve better
performance if some (limited) signaling is exchanged in order to optimize
also the sensing time. Second, as expected, more stringent interference
constraints impose lower missed detection probabilities as well as
false-alarm rates; requirement that is met by increasing the sensing
time (i.e., making the detection more accurate). This is clear in
the picture where one can see that the optimal sensing time duration
increases as the interference constraints increase. Third, the proposed
approach based on a penalty function leads to performance comparable
with those achievable by a centralized approach that computes the
optimal common sensing time based on a grid search.\smallskip{}

\noindent \textbf{Example \#3: Algorithms for} $\mathcal{G}_{\pi=0}(\mathcal{X},\boldsymbol{{\theta}})$
\textbf{(local constraints only).} In Fig. \ref{Fig_convergence_noprices},
we plot an instance of the sequential and simultaneous best-response
based algorithms, proposed in Sec. \ref{sub:Game-with-exogenous}
to solve the game $\mathcal{G_{\pi}}(\mathcal{X},\boldsymbol{{\theta}})$
in (\ref{eq:game_G_t_pi}), with ${\pi}=0$ (cf. Algorithm \ref{async_best-response_algo}).
We considered the same setup as in Fig. \ref{Fig_sens_through_tradeoff},
but with 15 active SUs; the SNR detection $\texttt{snr}_{d}\triangleq\sigma_{I_{q,k}}^{2}/\sigma_{q,k}^{2}$
is set to $\texttt{snr}_{d}=0$dB, for all $q$ and $k$; the SNR
of the SUs $ $ $\texttt{snr}_{q,k}\triangleq P_{q}/\sigma_{q}^{2}(k)$
is $\texttt{snr}_{q,k}=2$dB for all $q$ and $k$, and the (normalized)
inter-pair distances $d_{qr}/d_{qq}\geq3$ for all $q\neq r$, with
$d_{qr}$ denoting the distance between the receiver of SU $q$ and
the transmitter of SU $r$, which corresponds to a ``low/medium''
level of interference among the SUs; the bounds $\alpha_{q,k}$ and
$\beta_{q,k}$ are both equal to $0.5$ for all $q$ and $k$; and
$ $  the constant $c$ is set to $c=100$. In Fig. \ref{Fig_convergence_noprices}(a),
we plot the opportunistic throughput evolution of the SUs\textquoteright{}
links as a function of the iteration index, achieved using the sequential
best-response algorithm (solid line curves) and the simultaneous best-response
algorithm (dashed line curves); whereas in Fig. \ref{Fig_convergence_noprices}(b)
we plot the evolution of the optimal (normalized) sensing times of
the SUs versus the iteration index. To make the figures not excessively
overcrowded, we report only the curves of 3 out of 15 links. As expected,
the sequential best-response algorithm is slower than the simultaneous
version, especially if the number of active links is large, since
each SU is forced to wait for all the users scheduled in advance,
before updating his own strategy. However, both algorithms converge
in a few iterations (this desired feature has been observed for different
channel realizations), which makes them appealing in practical CR
scenarios. Observe also that, thanks to the penalty term on the sensing
times in the objective function of each SU, the algorithms converge
to the same optimal sensing time for all the SUs {[}cf. Fig. \ref{Fig_convergence_noprices}(b){]}.
Roughly speaking, these algorithms share the same features of the
well-known iterative waterfilling algorithms solving the power control
game over ICs \cite{Luo-Pang_IWFA-Eurasip,Scutari-Palomar-Barbarossa_SP08_PI,Scutari-Palomar-Barbarossa_AIWFA_IT08,Pang-Scutari-Palomar-Facchinei_SP_10,ScutariPalomarFacchineiPang-Monotone_bookCh}. 

Finally, observe that, even when the theoretical convergence conditions
we obtained are not satisfied, still we can claim that every limit
point of the sequence generated by our algorithms is a QNE of the
game.\smallskip{}

\noindent \textbf{Example \#4: Algorithms for} $\mathcal{G}(\mathcal{X},\boldsymbol{{\theta}})$
\textbf{(global constraints).} In Fig. \ref{Fig_exAlgorithm1_rate_and_interference}
we tested the convergence speed of Algorithm 1 applied to the game
$\mathcal{G}(\mathcal{X},\boldsymbol{{\theta}})$ in the presence
of global interference constraints. The system setup is the same as
the one considered in Fig. \ref{Fig_convergence_noprices} for the
low/medium interference regime, with the only difference that now,
instead of the overall bandwidth interference constraints (\ref{eq:individual_overal_interference_constraint}),
we assume that the PUs impose the global interference constraint (\ref{eq:global_interference_constraints_2});
for the sake of simplicity we considered the same interference threshold
for both the PUs. In Fig. \ref{Fig_exAlgorithm1_rate_and_interference},
we plot the opportunistic throughput evolution of $4$ (out of $15$)
SUs\textquoteright{} links and the worst-case average violation of
the interference constraints as a function of the iteration index
(counted considering both the inner and the outer iterations), achieved
using Algorithm \ref{algo_bi_level}. As expected, Fig. \ref{Fig_exAlgorithm1_rate_and_interference}
shows that the algorithms proposed to solve the game $\mathcal{G}(\mathcal{X},\boldsymbol{{\theta}})$
with side constraints require more iterations to converge that those
used to solve the game $\mathcal{G}_{\pi=0}(\mathcal{X},\boldsymbol{{\theta}})$.
On the other hand, global interference constraints impose less stringent
conditions on the transmit power of the SUs than those imposed by
the individual interference constraints, implying better throughput
performance of the SUs (at the price however of more signaling among
the SUs) \cite{Pang-Scutari-NNConvex_PI}. 
\begin{figure}[H]
\center\includegraphics[height=9cm]{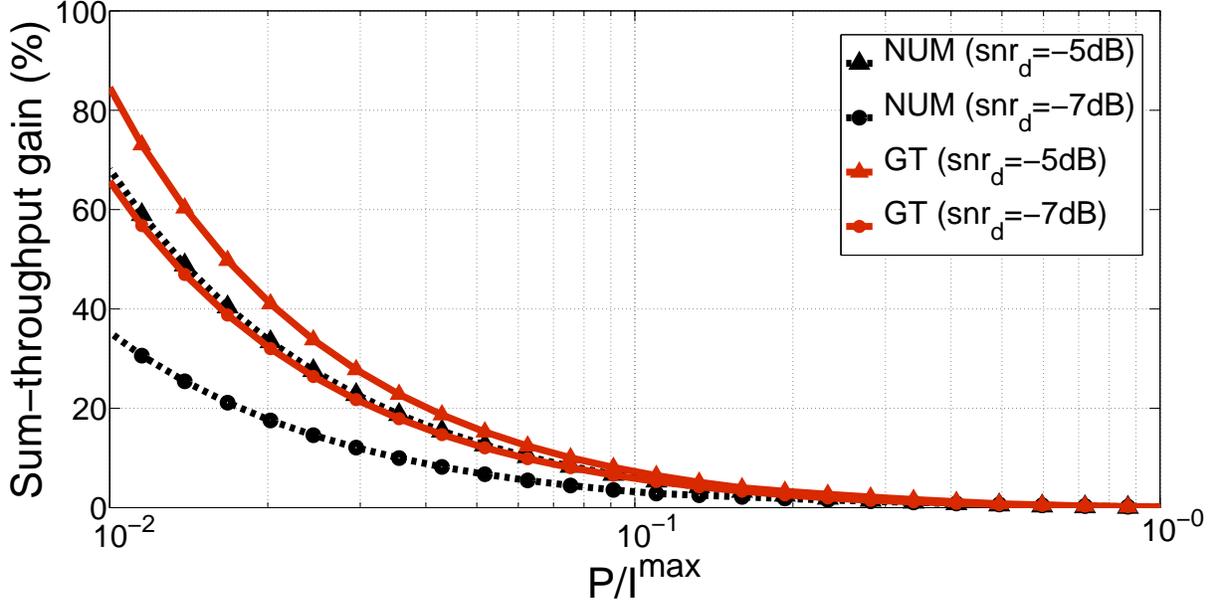}\vspace{-.3cm}

{\footnotesize \caption{{\small Comparison of proposed joint sensing/transmission optimization
with state-of-the-art NUM (cooperative) and game theoretical (noncooperative)
schemes where no sensing is optimized: (\%) ratio $(SR_{QE}-SR)/SR$
versus the (normalized) interference constraint bound $P/I^{\text{{max}}}$.}
\label{Fig_comparison}}
}{\footnotesize \par}

\vspace{-0.5cm}
\end{figure}
\begin{figure}[H]
\center\includegraphics[height=9cm]{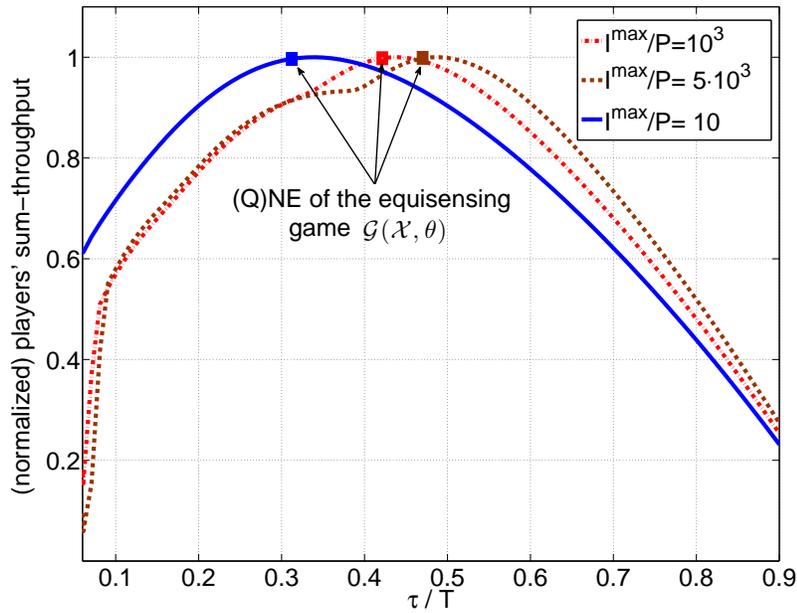}\vspace{-.6cm}

\caption{{\small Normalized throughput versus the normalized sensing time,
for different values of the (normalized) interference threshold. The
square markers correspond to the (Q)NE of the game $\mathcal{G}(\mathcal{X},\boldsymbol{{\theta}})$,
achieved with $c=100$.}\label{Fig_sens_through_tradeoff}}
\end{figure}
 
\begin{figure}[H]
\vspace{0.4cm}\hspace{-0.8cm}\includegraphics[height=7.5cm]{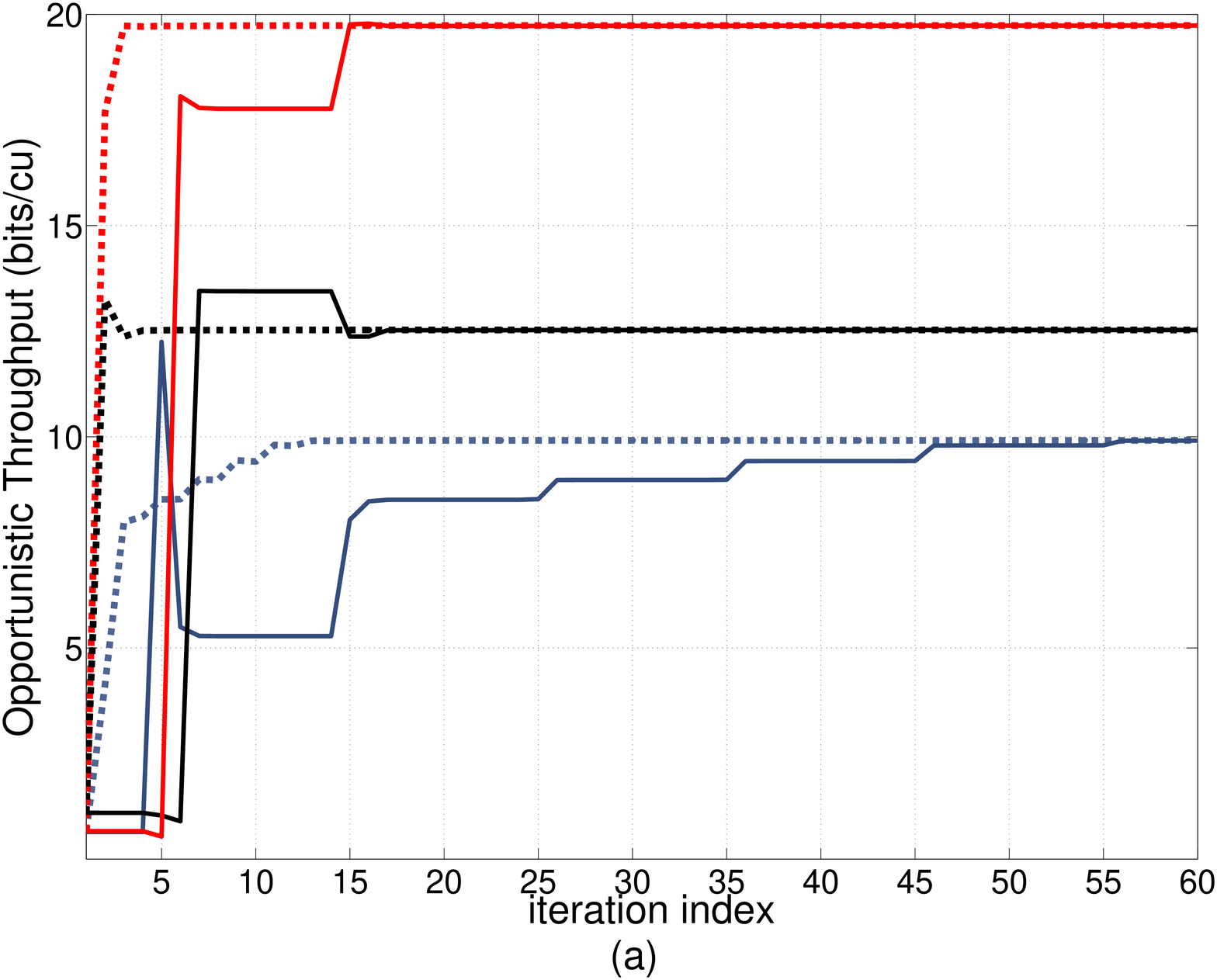}\hspace{-0.7cm}\includegraphics[height=7.5cm]{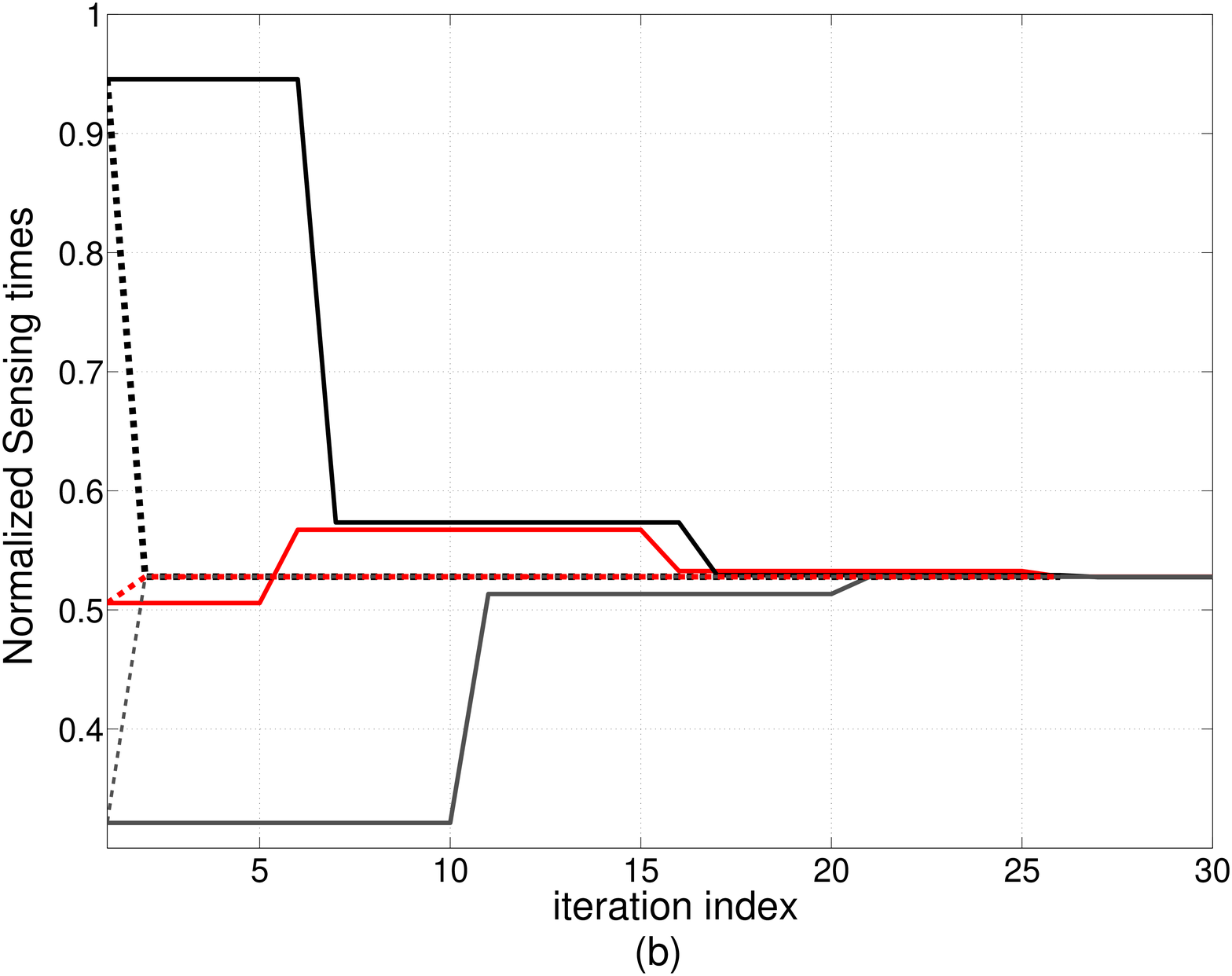}\vspace{-.3cm}

{\footnotesize \caption{{\small Example of convergence speed of sequential (solid line curves)
and simultaneous (dashed line curves) best-response based algorithms
applied to the game }$\mathcal{G}_{\pi}(\mathcal{X},\boldsymbol{{\theta}})${\small :
Secondary users\textquoteright{} opportunistic throughput (subplot
a) and normalized sensing times (subplot b) versus the iteration index.}\label{Fig_convergence_noprices}}
}
\end{figure}
\begin{figure}[H]
\vspace{-.4cm}\center\includegraphics[height=11cm]{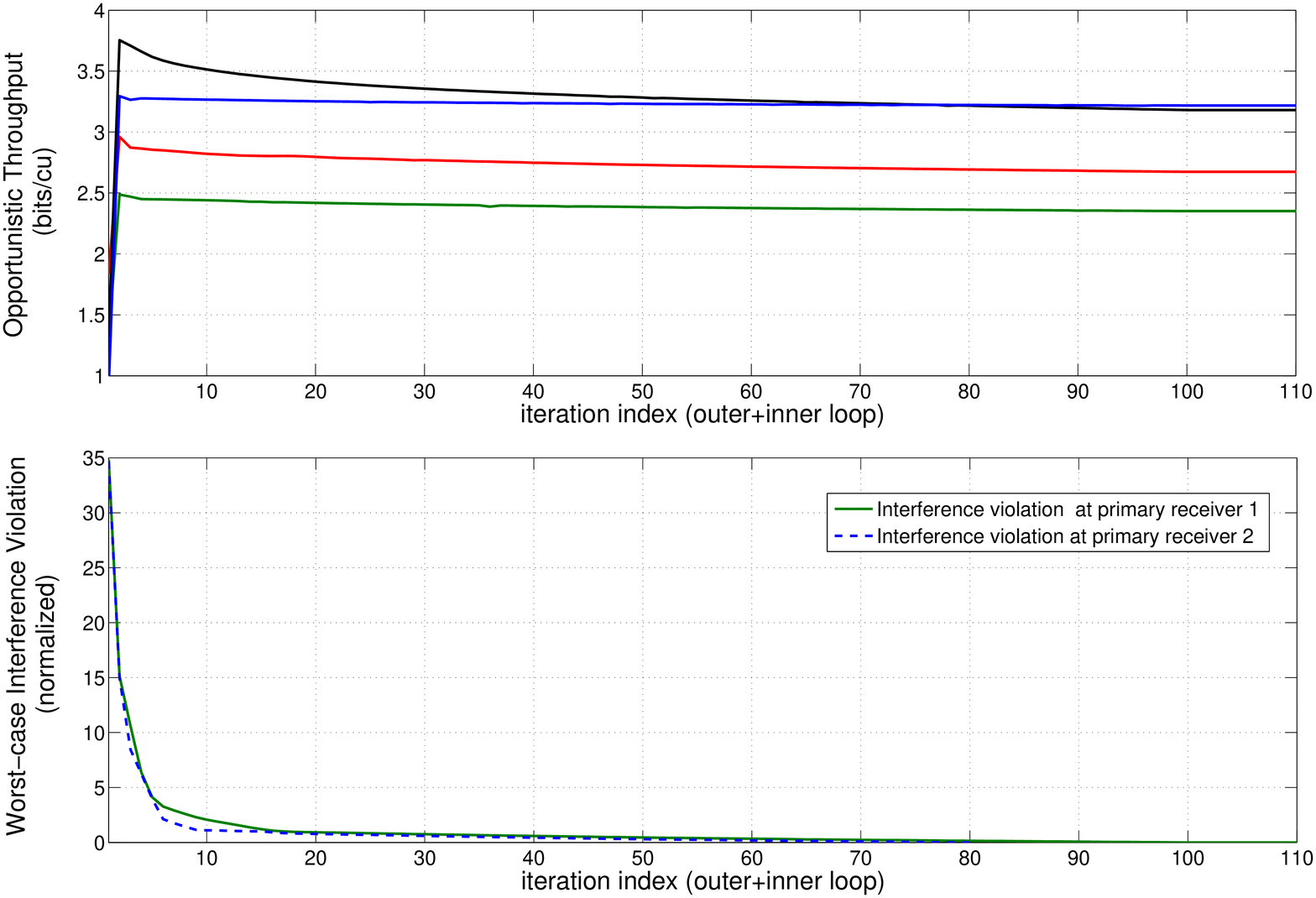}\vspace{-.6cm}

{\footnotesize \caption{{\small Algorithm 1 applied to game }\textbf{$\mathcal{G}(\mathcal{X},\,\boldsymbol{{\theta}})$}{\small :
Opportunistic throughput and average interference violation versus
iterations (outer plus inner loop) }.\label{Fig_exAlgorithm1_rate_and_interference}}
}{\footnotesize \par}

\vspace{-0.3cm}
\end{figure}

\section{Conclusions \label{sec:Conclusions}}

\noindent In this paper, we proposed a novel class of noncooperative
games with (possibly) side constraints, where each SU aims to maximize
his own opportunistic throughput by choosing jointly the sensing duration,
the detection thresholds, and the vector power allocation over SISO
frequency-selective interference channels, under local and (possibly)
global average probabilistic interference constraints. In particular,
to enforce global interference constraints while keeping the optimization
as decentralized as possible, we proposed a pricing mechanism that
penalizes the SUs in violating the global interference constraints.
The proposed games belong to the class of nonconvex games and lack
boundedness in some of the optimization variables, which makes the
analysis quite involved. A major contribution of this paper was to
introduce a new methodology for studying the existence and the uniqueness
of the solution of nonconvex games with side constraints and design
distributed solution algorithms. The proposed class of algorithms
spans from noncooperative settings modeling selfish users to cooperative
scenarios where the users are willing to exchange limited signaling
(in the form of consensus algorithms) in favor of better performance.
Numerical results showed the superiority of the proposed design (in
terms of achievable system throughput) with respect to the state-of-the-art
centralized and decentralized resource allocation algorithms for CR
systems. Together with their fast convergence behavior, this makes
them appealing in many practical CR scenarios. 

\appendix

\section*{Appendix }

\section{Proof of Proposition \ref{proposition_uniqueness_opt_sol} \label{app:Proof-of-Proposition_uniqueness_NE}}

\subsection{Intermediate results\label{sub:Intermediate-results_AppA}}

To prove the proposition we need two intermediate results, stated
in Lemma \ref{Lemma_ACQ} and Lemma \ref{Lemma_bounded_multipliers}
below. Lemma \ref{Lemma_ACQ} proves that the Abadie  Constraint Qualification
(ACQ) holds true at every (nontrivial) optimal solution of (\ref{eq:game_G_t_2}),
which implies that any of such solutions must satisfy the KKT conditions
associated with (\ref{eq:game_G_t_2}). Lemma \ref{Lemma_bounded_multipliers}
proves the boundedness of the multipliers $\lambda_{q}^{\star}$ associated
with the local nonconvex constraint $I(\mathbf{x}_{q}^{\star})\leq0$
at any solution $\mathbf{x}_{q}^{\star}$ of (\ref{eq:game_G_t_2}). 

\begin{lemma} \label{Lemma_ACQ} The ACQ holds at every feasible
solution of problem (\ref{eq:game_G_t_2}).\end{lemma}\begin{proof}
The proof follows similar steps of \cite[Prop. 8]{Pang-Scutari-NNConvex_PI}
and thus is omitted. \end{proof}

\begin{lemma}\label{Lemma_bounded_multipliers}Let $\mathbf{x}_{-q}\in\mathcal{Y}_{-q}$
and ${\pi}_{t}\in\mathcal{S}_{t}$ for some $t>0$. At every solution
$\mathbf{x}_{q}^{\star}$ of (\ref{eq:game_G_t_2}), any optimal multiplier
${\lambda}_{q}^{\star}$ associated with the constraint $I(\mathbf{x}_{q}^{\star})\leq0$
satisfies ${\lambda}\leq\lambda^{\max}$, with $\lambda^{\max}$ defined
in (\ref{eq:t_star_def}).\end{lemma}\begin{proof} First of all,
observe that the nonconvex problem (\ref{eq:game_G_t_2}) admits a
solution $\mathbf{x}_{q}^{\star}=\left(\wh{\tau}{}_{q}^{\star},\,\mathbf{p}_{q}^{\star},\, P_{q}^{\text{{fa}}\star}\right)$,
for every given $\mathbf{x}_{-q}\in\mathcal{Y}_{-q}$ and $\pi{}_{t}\in\mathcal{S}_{t}$;
by Lemma \ref{Lemma_ACQ}, $\mathbf{x}_{q}^{\star}$ must satisfy
the KKT conditions of the problem, for some multipliers ${\lambda}_{q}^{\star}$
associated with the constraint $I(\mathbf{x}_{q}^{\star})\leq0$.
Given the KKT conditions (which are omitted here), starting from the
complementarity of the $p_{q,k}$-variables, summing over $k$, and
invoking the orthogonality condition, we obtain: denoting by $\chi_{q}^{\star}$
and $\xi_{q,k}^{\star}$ the multipliers associated to the power budget
and the spectral mask constraints, respectively, 
\begin{equation}
\begin{array}{l}
{\displaystyle {\displaystyle {\displaystyle \left(\,\lambda_{q}^{\star}+{\pi}_{t}\right)\,{\sum_{k=1}^{N}}P_{q,k}^{\text{{miss}}}(\hat{{\tau}}_{q}^{\star},\, P_{q}^{\text{{fa}}\star})\,|G_{P,q}(k)|^{2}\, p_{q,k}^{\star}+{\displaystyle \chi_{q}^{\star}\,{\displaystyle {\sum_{k=1}^{N}}\, p_{q,k}^{\star}+{\displaystyle {\displaystyle {\sum_{k=1}^{N}}\,\xi_{q,k}^{\star}\, p_{q,k}^{\star}}}}}}}}\\[0.3in]
={\displaystyle \,{\displaystyle {\sum_{k=1}^{N}}\,{\displaystyle {\frac{p_{q,k}^{\star}}{\left({\displaystyle {\sum_{k=1}^{N}}}r_{q,k}(p_{q,k}^{\star},\mathbf{p}_{-q})\right)\,\left(\,{\sigma}_{q,k}^{2}+{\displaystyle {\sum_{r\neq q}}\,|{H}_{qr}(k)|^{2}\, p_{r,k}\,+|{H}_{qq}(k)|^{2}p_{q,k}^{\star}}\right)}}}}}\bigskip\\
\leq\,{\displaystyle \frac{1}{\left[\,{\displaystyle {\min_{1\leq k\leq N}}\,\left\{ \,\log\left(\,1+{\displaystyle {\frac{|{H}_{qq}(k)|^{2}p_{q,k}^{\max}}{{\sigma}_{q}^{2}(k)+{\displaystyle {\sum_{r\neq q}}\,|{H}_{qr}(k)|^{2}\, p_{r}^{\max}(k)}}}\,}\right)\,\right\} \,}\right]\,{\displaystyle {\min_{1\leq k\leq N}}\,\left\{ {\sigma}_{q,k}^{2}\right\} }}}\triangleq\lambda_{q}^{\max},
\end{array}\label{eq:bound_on_multipliers}
\end{equation}
where in the last inequality we used the following property of the
logarithmic function, which is an immediate consequence of its concavity:
for any scalar $a>0$ and $c>0$, it holds that $\log(1+c\, y)\,\geq\, y\,\log(1+c\, a),$
for all $y\in[0,\, a].$ Inequality in (\ref{eq:bound_on_multipliers})
together with the complementarity conditions associated to the power
constraints $\mathbf{p}_{q}^{\star}\leq\mathbf{p}_{q}^{\max}$ and
$\sum_{k=1}^{N}p_{q,k}^{\star}\leq P_{q}$, and the individual nonconvex
interference constraint $I_{q}(\mathbf{x}_{q}^{\star})\leq0$ lead
to 
\begin{equation}
\begin{array}{l}
{\displaystyle \lambda_{q}^{\star}\,{I}_{q}^{\text{{max}}}+{\displaystyle \,\chi_{q}^{\star}\, P_{q}+{\displaystyle \,{\displaystyle {\sum_{k=1}^{N}}\,\xi_{q,k}^{\,\star}\, p_{q,k}^{\max}}}}\leq\lambda_{q}^{\max}.}\\[0.3in]\end{array}\label{eq:mult bound}
\end{equation}
The desired result ${\lambda}_{q}^{\star}\leq\lambda^{\max}$ follows
from (\ref{eq:mult bound}) and $\min\left\{ P_{q},\,\min_{k}\,\{p_{q,k}^{\max}\}\right\} =\min_{k}\, p_{q,k}^{\max}$
for $q$. \end{proof}

\subsection{Proof of Proposition \ref{proposition_uniqueness_opt_sol}}

The proof is organized in the following two steps: 

\noindent \textbf{Step 1}. We show first that under the assumptions
in the proposition, each problem (\ref{eq:game_G_t_2}) has a unique
optimal solution, for any given $\mathbf{x}_{-q}\in\mathcal{Y}_{-q}$. 

\noindent \textbf{Step 2}. Then, we prove that any optimal solution
of (\ref{eq:game_G_t_2}) is nontrivial.

\noindent \textbf{Step 1}. Given $\mathbf{x}_{-q}\in\mathcal{Y}_{-q}$
and $\pi{}_{t}\in\mathcal{S}_{t}$, let $\mathbf{x}_{q}^{\star}=\left(\wh{\tau}{}_{q}^{\star},\,\mathbf{p}_{q}^{\star},\, P_{q}^{\text{{fa}}\star}\right)$
be a solution of (\ref{eq:game_G_t_2}); by Lemma \ref{Lemma_ACQ},
there exists a multiplier ${\lambda}_{q}^{\star}$ such that $ $
$(\mathbf{x}_{q}^{\star},\lambda_{q}^{\star})$ satisfies the VI$(\mathcal{K}_{q},\mathbf{F}_{q})$
in (\ref{eq:VI_ref}); by Lemma \ref{Lemma_bounded_multipliers},
it must be ${\lambda}_{q}^{\star}\leq\lambda^{\max}$. It turns out
that to prove Proposition \ref{proposition_uniqueness_opt_sol} is
sufficient to show that, under the condition in the proposition, the
VI$(\mathcal{K}_{q},\mathbf{F}_{q})$ has a unique solution in the
$x_{q}$-variables. 

Suppose by contradiction that there are two distinct solutions $ $of
the VI$(\mathcal{K}_{q},\mathbf{F}_{q})$, denoted by $\mathbf{y}_{q}^{(1)}\triangleq(\mathbf{x}_{q}^{(1)},\lambda_{q}^{(1)})\in\mathcal{Y}_{q}\times[0,\lambda^{\max}]$
and $\mathbf{y}_{q}^{(2)}\triangleq(\mathbf{x}_{q}^{(2)},\lambda_{q}^{(2)})\in\mathcal{Y}_{q}\times[0,\lambda^{\max}]$,
with $\mathbf{x}_{q}^{(1)}\neq\mathbf{x}_{q}^{(2)}$. Then, we have
\[
\begin{array}{l}
\left(\mathbf{y}_{q}^{(2)}-\mathbf{y}_{q}^{(1)}\right)^{T}\mathbf{F}_{q}\left(\mathbf{y}_{q}^{(1)};\,\mathbf{x}_{-q},\,{\pi}_{t}\right)\geq0\smallskip\\
\left(\mathbf{y}_{q}^{(1)}-\mathbf{y}_{q}^{(2)}\right)^{T}\mathbf{F}_{q}\left(\mathbf{y}_{q}^{(2)};\,\mathbf{x}_{-q},\,{\pi}_{t}\right)\geq0.
\end{array}
\]
Summing the two inequalities yields to 
\begin{equation}
0\leq-\left(\mathbf{y}_{q}^{(1)}-\mathbf{y}_{q}^{(2)}\right)^{T}\left(\mathbf{F}_{q}(\mathbf{y}_{q}^{(1)};\,\mathbf{x}_{-q},\,{\pi}_{t})-\mathbf{F}_{q}(\mathbf{y}_{q}^{(2)};\,\mathbf{x}_{-q},\,{\pi}_{t})\right).\label{eq:ineq_chain_0}
\end{equation}
 Invoking the mean-value theorem applied to to the univariate, differentiable,
scalar-valued function 
\begin{equation}
\delta\in[0,\,1]\mapsto\left(\mathbf{y}_{q}^{(1)}-\mathbf{y}_{q}^{(2)}\right)^{T}\mathbf{F}_{q}\left(\mathbf{y}_{q}(\delta);\,\mathbf{x}_{-q},\,{\pi}_{t}\right);\label{eq:univariate_function}
\end{equation}
we deduce that there exists some $0<\bar{{\delta}}<1$, such that
(\ref{eq:ineq_chain_0}) can be written as\vspace{-0.4cm}

\begin{eqnarray}
0 & \leq & -\left(\mathbf{y}_{q}^{(1)}-\mathbf{y}_{q}^{(2)}\right)^{T}\left(\mathbf{F}_{q}(\mathbf{y}_{q}^{(1)};\,\mathbf{x}_{-q},\,{\pi}_{t})-\mathbf{F}_{q}(\mathbf{y}_{q}^{(2)};\,\mathbf{x}_{-q},\,{\pi}_{t})\right)\label{eq:ineq_chain_uniq_proof_1}\\
 & = & -\left(\mathbf{y}_{q}^{(1)}-\mathbf{y}_{q}^{(2)}\right)^{T}\mbox{J}_{\mathbf{y}_{q}}\mathbf{F}_{q}\left(\mathbf{y}_{q}(\bar{{\delta}});\,\mathbf{x}_{-q},\,{\pi}_{t}\right)\,\left(\mathbf{y}_{q}^{(1)}-\mathbf{y}_{q}^{(2)}\right)\label{eq:ineq_chain_uniq_proof_2}\\
 & = & -\left(\begin{array}{c}
\mathbf{x}_{q}^{(1)}-\mathbf{x}_{q}^{(2)}\\
\lambda_{q}^{(1)}-\lambda_{q}^{(2)}
\end{array}\right)^{T}\left[\begin{array}{cc}
\nabla_{\mathbf{x}_{q}}^{2}\mathcal{L}_{q}\left((\mathbf{x}_{q}(\bar{{\delta}}),\,\mathbf{x}_{-q}),\,{\pi}_{t},\,{\lambda}_{q}(\bar{{\delta}})\right),\, & \nabla_{\mathbf{x}_{q}}I_{q}\left(\mathbf{x}_{q}(\bar{{\delta}})\right)\\
-\nabla_{\mathbf{x}_{q}}I_{q}\left(\mathbf{x}_{q}(\bar{{\delta}})\right)^{T} & 0
\end{array}\right]\!\left(\begin{array}{c}
\mathbf{x}_{q}^{(1)}-\mathbf{x}_{q}^{(2)}\\
\lambda_{q}^{(1)}-\lambda_{q}^{(2)}
\end{array}\right)\label{eq:ineq_chain_uniq_proof_3}\\
 & = & -\left(\mathbf{x}_{q}^{(1)}-\mathbf{x}_{q}^{(2)}\right)^{T}\nabla_{\mathbf{x}_{q}}^{2}\mathcal{L}_{q}\left((\mathbf{x}_{q}(\bar{{\delta}}),\,\mathbf{x}_{-q}),\,{\pi}_{t},\,{\lambda}_{q}(\bar{{\delta}})\right)\left(\mathbf{x}_{q}^{(1)}-\mathbf{x}_{q}^{(2)}\right),\label{eq:ineq_chain_uniq_proof_4}
\end{eqnarray}
where in (\ref{eq:ineq_chain_uniq_proof_1}) $\mbox{J}_{\mathbf{y}_{q}}\mathbf{F}_{q}(\,\cdot\,;\,\mathbf{x}_{-q},\,{\pi}_{t})$
denotes the Jacobian matrix of $\mathbf{F}_{q}(\cdot\,;\mathbf{x}_{-q},\,{\pi}_{t})$
with respect to $\mathbf{y}_{q}\triangleq\left(\mathbf{x}_{q},\,\lambda_{q}\right)$.
Since $\mathbf{x}_{q}(\bar{{\delta}})\in\mathcal{Y}_{q}$ (recall
that $\mathcal{Y}_{q}$ is a convex set) and ${\lambda}_{q}(\bar{{\delta}})\leq\lambda^{\max}$,
the inequality in (\ref{eq:ineq_chain_uniq_proof_4}) contradicts
the positive definiteness of $\nabla_{\mathbf{x}_{q}}^{2}\mathcal{L}_{q}\left((\mathbf{x}_{q}(\bar{{\delta}}),\,\mathbf{x}_{-q}),\,{\pi}_{t},\,{\lambda}_{q}(\bar{{\delta}})\right)$,
as assumed in Proposition \ref{proposition_uniqueness_opt_sol}.\medskip

\noindent \textbf{Step 2}. To complete the proof it is enough to
show that the $\mathbf{p}_{q}$-component of any optimal solution
$\mathbf{x}_{q}^{\star}=\left(\wh{\tau}{}_{q}^{\star},\,\mathbf{p}_{q}^{\star},\, P_{q}^{\text{{fa}}\star}\right)$
of (\ref{eq:game_G_t_2}) is such that $\sum_{k}p_{q}^{\star}(k)$
is lower bounded by a positive constant; see Lemma \ref{Lemma_lower_bound_on_the_power}
below. To state the lemma, we need the following intermediate definitions.
Let $\mathbf{p}_{q}^{\text{{ref}}}\triangleq(p_{q,k}^{\text{{ref}}})\in\mathcal{P}_{q}$
be any tuple such that 
\begin{equation}
{\displaystyle {\sum_{k}}\,|G_{P,q}(k)|^{2}\, p_{q,k}^{\text{{ref}}}\,\leq\,2\, I_{q}^{\max},}\label{eq:upper_interference}
\end{equation}
so that for all pairs $(\wh{\tau}_{q},\, P_{q}^{\text{{fa}}})$ satisfying
(\ref{eq:player q_equi-sensing})(b), the interference constraints
(\ref{eq:player q_equi-sensing})(a) evaluated at $(\wh{\tau}_{q},\,\mathbf{p}_{q}^{\text{{ref}}},\, P_{q}^{\text{{fa}}})$
hold; and let 
\begin{equation}
P_{q}^{\text{{fa}}^{\text{{ref}}}}\triangleq\max_{k}\,\left\{ \mathcal{Q}\left(\frac{{\sigma_{q,k|1}}\,\wh{\alpha}{}_{q,k}+({\mu_{q,k|1}-\mu_{q,k|0}})\,\sqrt{{f_{q}}\,\tau^{\min}}}{{\sigma_{q,k|0}}}\right)\right\} .\label{eq:Pfa_ref}
\end{equation}
Note that, under the feasibility conditions (\ref{eq:nec_suff_feasibility_cond}),
such a $P_{q}^{\text{{fa}}^{\text{{ref}}}}$ satisfies {[}see (\ref{eq:player q_equi-sensing})(b){]}
\begin{equation}
\dfrac{{\sigma_{q,k|0}}\,}{{\sigma_{q,k|1}}}\mathcal{Q}^{-1}\left(P_{q}^{\text{{fa}}^{\text{{ref}}}}\right)-\wh{\tau}_{q}\,\dfrac{{\mu_{q,k|1}-\mu_{q,k|0}}}{{\sigma_{q,k|1}}}\,\leq\,\wh{\alpha}_{q,k},\quad\forall k=1,\ldots,N,\label{eq:interference_constraint}
\end{equation}
for any $\wh{\tau}_{q}\geq\sqrt{{f_{q}}\,\tau^{\min}}$. Finally,
given $t>0$, let
\begin{align}
\eta_{q}^{\text{{ref}}}(t) & \triangleq\log\left(1-\frac{{\tau^{\min}}}{T_{q}}\right)+\log\left(1-P_{q}^{\text{{fa}}^{\text{{ref}}}}\right)+\log\left({\displaystyle {\sum_{k}}}\, r_{q,k}\left(\mathbf{p}_{q}^{\text{{ref}}},\mathbf{p}_{-q}^{\max}\right)\right)-\dfrac{{t}}{2}\,\left(\max_{k=1,\ldots,N}\left\{ |G_{P,q}(k)|^{2}\, p_{q,k}^{\text{{ref}}}\right\} \right).\label{eq:n_q_ref}
\end{align}
We can now introduce Lemma \ref{Lemma_lower_bound_on_the_power} that
provides a lower bound for the optimal sum-power allocation of each
player. 

\begin{lemma}\label{Lemma_lower_bound_on_the_power} Given $t>0$,
and feasible ${\pi}_{t}\in\mathcal{S}_{t}$, $\mathbf{p}_{-q}\in\mathcal{P}_{-q}$
and $\wh{\tau}_{r}\in\left[\sqrt{{f_{r}}\,\tau^{\min}},\,\sqrt{{f_{r}}\,\tau^{\max}}\right]$
for all $r\neq q$, the power-part $\mathbf{p}_{q}^{\star}$ of any
optimal solution of the $q$-th nonconvex optimization problem in
(\ref{eq:player q_equi-sensing}) satisfies 
\begin{equation}
\sum_{k=1}^{N}p_{q,k}^{\star}\geq\left(\min_{k=1,\ldots,N}\left\{ \wh{\sigma}{}_{q,k}^{2}\right\} \right)\,\exp\left(\eta_{q}^{\text{{ref}}}(t)\right).\label{eq:lower_bound_p_opt}
\end{equation}
\end{lemma} 

\begin{proof} Let $t>0$, $\pi_{t}\in\mathcal{S}_{t}$, $\mathbf{0}\leq\mathbf{p}_{r}\leq\mathbf{p}_{r}^{\max}$
 with $r\neq q$, and $\wh{\tau}_{r}$ for $r\neq q$ satisfying $\wh{\tau}_{r}\in\left[\sqrt{{f_{r}}\,\tau^{\min}},\,\sqrt{{f_{r}}\,\tau^{\max}}\right]$
be given. Let define $\wh{\tau}{}_{q}^{\text{{ref}}}\triangleq\dfrac{{\sqrt{{f_{q}}}}}{Q-1}\,\sum_{r\neq q}\dfrac{{\wh{\tau}}_{r}}{{\sqrt{{f_{r}}}}}$;
we then have $\sqrt{\tau^{\min}}\leq\dfrac{\wh{\tau}{}_{q}^{\text{{ref}}}}{{\sqrt{{f_{q}}}}}\leq\dfrac{1}{Q}\,\sum_{r=1}^{Q}\dfrac{\wh{\tau}{}_{r}}{{\sqrt{{f_{r}}}}}\leq\sqrt{\tau^{\max}}$.
Therefore, if $\mathbf{x}_{q}^{\star}=(\wh{\tau}_{q}^{\star},\,\mathbf{p}_{q}^{\star},\, P_{q}^{\text{{fa}\ensuremath{\star}}})$
is player $q$'s best-response corresponding to ${\pi}_{t}$, $\wh{\boldsymbol{{\tau}}}_{-q},$
and $\mathbf{p}_{-q}$, then
\begin{equation}
\begin{array}{l}
\wh{R}{}_{q}\left(\wh{\tau}{}_{q}^{\text{{ref}}},\,(\mathbf{p}_{q}^{\text{{ref}}},\mathbf{p}_{-q}),\, P_{q}^{\text{{fa}}^{\text{{ref}}}}\right)-\pi_{t}\cdot{\displaystyle {\sum_{k}}P_{q,k}^{\text{{miss}}}(\hat{{\tau}}_{q}^{\text{{ref}}},\, P_{q}^{\text{{fa}}^{\text{{ref}}}})\,|G_{P,q}(k)|^{2}\, p_{q,k}^{\text{{ref}}}\,\vspace{-0.6cm}}\\
\leq R_{q}\left(\wh{\tau}_{q}^{\star},\,(\mathbf{p}_{q}^{\star},\mathbf{p}_{-q}),\, P_{q}^{\text{{fa}}\star}\right)-\pi_{t}\cdot{\displaystyle {\sum_{k}}P_{q,k}^{\text{{miss}}}(\hat{{\tau}}_{q}^{\text{{ref}}},\, P_{q}^{\text{{fa}}^{\text{{ref}}}})\,|G_{P,q}(k)|^{2}\, p_{q,k}^{\star}-\,\dfrac{{c}}{2}\,\left(\left(1-\dfrac{{1}}{Q}\right)\dfrac{{\wh{\tau}_{q}^{\star}}}{\sqrt{{f_{q}}}}-\dfrac{{1}}{Q}\,{\displaystyle {\sum_{r\neq q}}}\,\dfrac{{\wh{\tau}_{r}^{\star}}}{\sqrt{{f_{r}}}}\right)^{2}\vspace{-0.6cm}}\\
\leq\log\left({\displaystyle {\sum_{k}}}\, r_{q,k}(\mathbf{p}_{q}^{\star},\mathbf{p}_{-q})\right)\leq\log\left({\displaystyle {\sum_{k}}}\log\left(1+\dfrac{p_{q,k}^{\star}}{\hat{{\sigma}}_{q,k}^{2}}\right)\right)\leq\log\left({\displaystyle {\sum_{k}}}\left(\dfrac{p_{q,k}^{\star}}{\hat{{\sigma}}_{q,k}^{2}}\right)\right),
\end{array}\label{eq:right_upper_bound}
\end{equation}
where $\wh{R}{}_{q}\left(\wh{\tau}_{q},\,\mathbf{p},\, P_{q}^{\text{{fa}}}\right)$
and $r_{q,k}\left(\mathbf{p}_{q},\mathbf{p}_{-q}\right)$ are defined
in (\ref{eq:Opportunistic-throughtput}) and (\ref{eq:rate_FSIC}),
respectively. On the other end, we have:
\begin{equation}
\begin{array}{l}
\wh{R}{}_{q}\left(\wh{\tau}{}_{q}^{\text{{ref}}},\,(\mathbf{p}_{q}^{\text{{ref}}},\mathbf{p}_{-q}),\, P_{q}^{\text{{fa}}^{\text{{ref}}}}\right)-{\displaystyle \pi_{t}\cdot{\displaystyle {\sum_{k}}P_{q,k}^{\text{{miss}}}(\hat{{\tau}}_{q}^{\text{{ref}}},\, P_{q}^{\text{{fa}}^{\text{{ref}}}})\,|G_{P,q}(k)|^{2}\, p_{q,k}^{\text{{ref}}}}\,\geq\eta_{q}^{\text{{ref}}}(t)},\end{array}\label{eq:left_lower_bound}
\end{equation}
with $\eta_{q}^{\text{{ref}}}(t)$ defined in (\ref{eq:n_q_ref}),
and in (\ref{eq:left_lower_bound}) we used ${\pi}_{t}\in\mathcal{S}_{t}$
and $P_{q,k}^{\text{{miss}}}\leq1/2$$ $. The desired bound in (\ref{eq:lower_bound_p_opt})
follows readily from (\ref{eq:right_upper_bound}) and (\ref{eq:left_lower_bound}).\end{proof}

\section{Proof of Corollary \ref{corollary_sf_cond_uniqueness_opt_sol}\label{sec:Proof-of-Corollary_existence}}

The proof is based on the following two steps. 

\noindent \textbf{Step 1}. We introduce a symmetric matrix, denoted
by $\overline{\nabla_{\mathbf{x}_{q}}^{2}\mathcal{L}_{q}}\in\mathbb{R}^{(N+2)\times(N+2)}$,
having the property that: given $t>0$, 
\begin{equation}
\mathbf{y}^{T}\,\left(\nabla_{\mathbf{x}_{q}}^{2}\mathcal{L}_{q}\left(\mathbf{x},\,{\pi}_{t},\,{\lambda}_{q}\right)\right)\,\mathbf{y}\geq\left|\mathbf{y}\right|^{T}\overline{\nabla_{\mathbf{x}_{q}}^{2}\mathcal{L}_{q}}\left|\mathbf{y}\right|\quad\forall\left(\mathbf{x},\,{\pi}_{t},\,{\lambda}_{q}\right)\in\mathcal{Y}\times\mathcal{S}_{t}\times[0,\lambda^{\max}],\quad\mbox{and}\quad\mathbf{y}\in\mathbb{R}^{N+2},\label{eq:copositivity_Lower_bound}
\end{equation}
which guarantees that $\nabla_{\mathbf{x}_{q}}^{2}\mathcal{L}_{q}\left(\mathbf{x},\,{\pi}_{t},\,{\lambda}_{q}\right)$
is positive definite if $\overline{\nabla_{\mathbf{x}_{q}}^{2}\mathcal{L}_{q}}$
is so. 

\noindent\textbf{Step 2}. We derive sufficient conditions for $\overline{\nabla_{\mathbf{x}_{q}}^{2}\mathcal{L}_{q}}$
to be positive definite. 

\noindent \textbf{Step 1}. It is not difficult to see that (\ref{eq:copositivity_Lower_bound})
is satisfied if $\overline{\nabla_{\mathbf{x}_{q}}^{2}\mathcal{L}_{q}}$
is built such that: for all $\left(\mathbf{x},\,{\pi}_{t},\,{\lambda}_{q}\right)\in\mathcal{Y}\times[0,t]\times[0,\lambda^{\max}]$,
\begin{equation}
\left[\overline{\nabla_{\mathbf{x}_{q}}^{2}\mathcal{L}_{q}}\right]_{ij}=\left\{ \begin{array}{lll}
\leq\left[\nabla_{\mathbf{x}_{q}}^{2}\mathcal{L}_{q}\left(\mathbf{x},\,{\pi}_{t},\,{\lambda}_{q}\right)\right]_{ij} &  & \mbox{if }i=j,\\
\\
\leq-\left|\left[\nabla_{\mathbf{x}_{q}}^{2}\mathcal{L}_{q}\left(\mathbf{x},\,{\pi}_{t},\,{\lambda}_{q}\right)\right]_{ij}\right| &  & \mbox{if }i\neq j.
\end{array}\right.\label{eq:nabla_L_bar}
\end{equation}

To construct such a matrix, we need to bound properly the entries
of $\nabla_{\mathbf{x}_{q}}^{2}\mathcal{L}_{q}\left(\mathbf{x},\,{\pi}_{t},\,{\lambda}_{q}\right)$.
Recalling that $\nabla_{\mathbf{x}_{q}}^{2}\mathcal{L}_{q}\left(\mathbf{x},\,{\pi}_{t},\,{\lambda}_{q}\right)$
has the following expression {[}cf. (\ref{eq:L_q_2_matrix}){]}: 
\begin{align}
\nabla_{\mathbf{x}_{q}}^{2}\mathcal{L}_{q}\left(\mathbf{x},\pi_{t},\lambda_{q}\right) & \triangleq-\nabla_{\mathbf{x}_{q}}^{2}\theta{}_{q}(\mathbf{x})+\lambda_{q}\cdot\nabla_{\mathbf{x}_{q}}^{2}I_{q}(\mathbf{x}_{q})+\pi_{t}\cdot\nabla_{\mathbf{x}_{q}}^{2}I(\mathbf{x})\label{eq:nabla_L}
\end{align}
we focus next on each term in (\ref{eq:nabla_L}) separately. 

\noindent$-$Matrix\emph{ $-\nabla_{\mathbf{x}_{q}}^{2}\theta{}_{q}(\mathbf{x})$}:
Introducing 
\begin{equation}
r_{q}(\mathbf{p})\triangleq\sum_{k=1}^{N}r_{q,k}(\mathbf{p})\leq\sum_{k=1}^{N}\log\left(1+\dfrac{{p_{q,k}^{\max}}}{\wh{\sigma}_{q,k}^{2}}\right)\triangleq r_{q}^{\max},\label{eq:def_r_q_max}
\end{equation}
with $r_{q,k}(\mathbf{p})$ defined in (\ref{eq:rate_FSIC}), $-\nabla_{\mathbf{x}_{q}}^{2}\theta{}_{q}(\mathbf{x})$
is given by 

\begin{equation}
-\nabla_{\mathbf{x}_{q}}^{2}\theta{}_{q}(\mathbf{x})=\left[\begin{array}{ccc}
\,{\displaystyle {\frac{{\displaystyle {\frac{2}{f_{q}\, T_{q}}}\,\left(\,1+{\displaystyle {\frac{\wh{\tau}_{q}^{2}}{f_{q}\, T_{q}}}\,}\right)}}{\left(\,1-{\displaystyle {\frac{\wh{\tau}_{q}^{2}}{f_{q}\, T_{q}}}\,}\right)^{2}}}+c\,\left(\dfrac{{1-1/Q}}{\sqrt{{f_{q}}}}\right)^{2}} & \mathbf{0}_{1\times N} & 0\\
\mathbf{0}_{N\times1} & \nabla_{\mathbf{p}_{q}}^{2}\left(-\log r_{q}(\mathbf{p})\right) & \mathbf{0}_{N\times1}\\
0 & \mathbf{0}_{1\times N} & {\displaystyle {\frac{1}{\left(\,1-P_{q}^{\text{\,{fa}}}\right)^{2}}}}
\end{array}\right],\label{eq:nabla_square_f_q_x_q}
\end{equation}
with $ ${\small{} 
\begin{align}
\nabla_{\mathbf{p}_{q}}^{2}\left(-\log r_{q}(\mathbf{p}_{q},\mathbf{p}_{-q})\right) & =\left[{\displaystyle {\frac{-\nabla_{\mathbf{p}_{q}}^{2}r_{q}(\mathbf{p})}{r_{q}(\mathbf{p}_{q})}}}+{\displaystyle {\frac{\nabla_{\mathbf{p}_{q}}r_{q}(\mathbf{p})\,\nabla_{\mathbf{p}_{q}}r_{q}(\mathbf{p})^{T}}{r_{q}(\mathbf{p}_{q})^{2}}}}\right]\label{eq:nablas_of_r_q_a}\\
\nabla_{\mathbf{p_{q}}}r_{q}(\mathbf{p}) & =\mbox{vect}\left\{ \left(\,{\displaystyle {\frac{1}{\wh{\sigma}_{q,k}^{2}+{\displaystyle {\sum_{r=1}^{Q}}\,|\wh{H}_{qr}(k)|^{2}\, p_{r}(k)}}}\,}\right)_{k=1}^{N}\right\} \label{eq:nablas_of_r_q_b}\\
\nabla_{\mathbf{p}_{q}}^{2}r_{q}(\mathbf{p}) & =\mbox{Diag}\left\{ \left(\,{\displaystyle {\frac{-1}{\left(\,\wh{\sigma}_{q,k}^{2}+{\displaystyle {\sum_{r=1}^{Q}}\,|\wh{H}_{qr}(k)|^{2}\, p_{r,k}}\right)^{2}}}\,}\right)_{k=1}^{N}\right\} \label{eq:nablas_of_r_q_c}
\end{align}
}{\small \par}

We provide now some bounds of the above quantities that will be used
to define the diagonal entries of $\overline{\nabla_{\mathbf{x}_{q}}^{2}\mathcal{L}_{q}}$.
The minimum eigenvalue of the positive definite matrix $\nabla_{\mathbf{p}_{q}}^{2}\left(-\log r_{q}(\mathbf{p})\right)$
is lower bounded by: for all $\mathbf{p}\in\mathcal{P}=\prod_{q=1}^{Q}\mathcal{P}_{q}$,
{\small 
\begin{align}
\lambda_{\min}\left(\nabla_{\mathbf{p}_{q}}^{2}\left(-\log r_{q}(\mathbf{p})\right)\right) & \geq\min_{k=1,\ldots,N}\left\{ d_{-\log r_{q},k}^{\min}\triangleq\dfrac{{1}/r_{q}^{\max}}{{\wh{\sigma}_{q,k}^{2}+{\displaystyle {\sum_{r=1}^{Q}}\,|\wh{H}_{qr}(k)|^{2}\, p_{r,k}^{\max}}}}\right\} \triangleq d_{-\log r_{q}}^{\min},\label{eq:lambda_min_Hrate}
\end{align}
}whereas a lower bound of the first and last diagonal elements in
(\ref{eq:nabla_square_f_q_x_q}) are: for all feasible $\left(\wh{\tau}_{q},\, P_{q}^{\text{{fa}}}\right)$
{[}see conditions (b) and (c) in (\ref{eq:player q_equi-sensing}){]},
\begin{equation}
{\frac{{\displaystyle {\frac{2}{f_{q}\, T_{q}}}\,\left(\,1+{\displaystyle {\frac{\wh{\tau}_{q}^{2}}{f_{q}\, T_{q}}}\,}\right)}}{\left(\,1-{\displaystyle {\frac{\wh{\tau}_{q}^{2}}{f_{q}\, T_{q}}}\,}\right)^{2}}}\geq{\displaystyle {\frac{{\displaystyle {\frac{2}{f_{q}\, T_{q}}}\,\left(\,1+{\displaystyle {\frac{\left(\,{\tau}^{\,\min}\,\right)^{2}}{T_{q}}}\,}\right)}}{\left(\,1-{\displaystyle {\displaystyle {\frac{\left(\,{\tau}^{\,\min}\,\right)^{2}}{T_{q}}}\,}}\right)^{2}}}\,\triangleq\, d_{\wh{\tau}_{q}}^{\min}\,\,\mbox{and}\,\,{\displaystyle {\displaystyle {\frac{1}{\left(\,1-P_{q}^{\text{{fa}}^{\min}}\,\right)^{2}}}}\geq\,{\displaystyle {\frac{1}{\left(\,1-P_{q}^{\text{{fa}}}\,\right)^{2}}}\triangleq\, d_{P_{q}^{\text{{fa}}}}^{\min}},}}\label{eq:lamnda_lower_tau_q_and_pfa}
\end{equation}
where we used the following lower bound of $P_{q}^{\text{{fa}}}$:
$P_{q}^{\text{{fa}}}\geq\min_{k}\mathcal{Q}\left({\displaystyle {\frac{\left(\mu_{{q,k}|1}-\mu_{{q,k}|0}\right)\,\sqrt{f_{q}\,{{\tau}^{\max}}}}{\sigma_{{q,k}|0}}}}\right)\triangleq P_{q}^{\text{{fa}}^{\min}}.$
This bounds will be used to define the diagonal entries of the candidate
matrix $\overline{\nabla_{\mathbf{x}_{q}}^{2}\mathcal{L}_{q}}$. 

\noindent$-$Matrix\emph{ $\nabla_{\mathbf{x}_{q}}^{2}I_{q}(\mathbf{x}_{q})$}.
Let introduce first the following quantities and their associated
bounds: \textcolor{black}{\small 
\begin{equation}
\omega_{\wh{\tau}_{q},k}\triangleq{\frac{\partial P_{q,k}^{\text{{miss}}}(\hat{{\tau}}_{q},\, P_{q}^{\text{{fa}}})}{\partial\wh{\tau}_{q}}}\quad\mbox{and}\quad\left|\omega_{\wh{\tau}_{q},k}\right|\leq\frac{1}{\sqrt{2\,\pi}}\left(\,{\displaystyle {\frac{\mu_{{q,k}|1}-\mu_{{q,k}|0}}{\sigma_{{q,k}|1}}}\,}\right)\triangleq\omega_{\wh{\tau}_{q},k}^{\max}\label{eq:der_Pmiss_wrt_tau}
\end{equation}
\begin{equation}
\omega_{P_{q}^{\text{{fa}}},k}\triangleq{\frac{\partial P_{q,k}^{\text{{miss}}}(\hat{{\tau}}_{q},\, P_{q}^{\text{{fa}}})}{\partial P_{q}^{\text{{fa}}}}}\quad\mbox{and}\quad\left|\omega_{P_{q}^{\text{{fa}}},k}\right|\leq\left(\,{\displaystyle {\frac{\sigma_{{q,k}|0}}{\sigma_{{q,k}|1}}}\,}\right)\,\text{{exp}}\left\{ \left({\displaystyle {\frac{\mu_{{q,k}|1}-\mu_{{q,k}|0}}{\sigma_{{q,k}|0}}}\,}\sqrt{{f_{q}}{\tau}^{\max}}\right)^{2}/2\right\} \triangleq\omega_{P_{q}^{\text{{fa}}},k}^{\max},\label{eq:der_Pmiss_wrt_Pfa}
\end{equation}
\begin{align}
\begin{array}{l}
\omega_{\wh{\tau}_{q}\, P_{q}^{\text{{fa}}},k}\triangleq{\displaystyle {\frac{\partial^{2}P_{q,k}^{\text{{miss}}}(\hat{{\tau}}_{q},\, P_{q}^{\text{{fa}}})}{\partial\wh{\tau}_{q}\partial P_{q}^{\text{{fa}}}}}\quad\mbox{and}\quad\omega_{\wh{\tau}_{q}\, P_{q}^{\text{{fa}}},k}\leq\max\left\{ \mathcal{Q}^{-1}(\alpha_{q,k}),\,{\displaystyle {\frac{\mu_{{q,k}|1}-\mu_{{q,k}|0}}{\sigma_{{q,k}|1}}}\,}\sqrt{{f_{q}}{\tau}^{\max}}\right\} }\\
\text{\hspace{1em}\hspace{3.5cm}\hspace{2cm}\hspace{2cm}\,\,\ensuremath{\cdot}\,{exp}}\left\{ \left({\displaystyle {\frac{\mu_{{q,k}|1}-\mu_{{q,k}|0}}{\sigma_{{q,k}|0}}}\,}\sqrt{{f_{s}}{\tau}^{\max}}\right)^{2}/2\right\} \triangleq\omega_{\wh{\tau}_{q}\, P_{q}^{\text{{fa}}},k}^{\max}
\end{array}\label{eq:sec_der_Pmiss_wrt_tau_pfa}
\end{align}
\begin{equation}
\omega_{\wh{\tau}_{q}\wh{\tau}_{q},k}\triangleq{\frac{\partial^{2}P_{q,k}^{\text{{miss}}}(\hat{{\tau}}_{q},\, P_{q}^{\text{{fa}}})}{\partial(\wh{\tau}_{q}^{\,})^{2}}}\quad\text{{and}\quad}\omega_{P_{q}^{\text{{fa}}}P_{q}^{\text{{fa}}},k}\triangleq{\displaystyle {\frac{\partial^{2}P_{q,k}^{\text{{miss}}}(\hat{{\tau}}_{q},\, P_{q}^{\text{{fa}}})}{\partial(P_{fa}^{\,(q)})^{2}}}},\vspace{-0.2cm}\label{eq:sec_der_Pmiss_wrt_tau}
\end{equation}
}which can be collected in the vectors $\boldsymbol{{\omega}}_{\wh{\tau}_{q}}\triangleq(\omega_{\wh{\tau}_{q},k})_{k=1}^{N}$,
$\boldsymbol{{\omega}}_{P_{q}^{\text{{fa}}}}\triangleq(\omega_{P_{q}^{\text{{fa}}},k})_{k=1}^{N}$,
$\boldsymbol{{\omega}}_{\wh{\tau}_{q}\, P_{q}^{\text{{fa}}}}\triangleq(\omega_{\wh{\tau}_{q}\, P_{q}^{\text{{fa}}},k})_{k=1}^{N}$,
$\boldsymbol{{\omega}}_{P_{q}^{\text{{fa}}}P_{q}^{\text{{fa}}}}\triangleq(\omega_{P_{q}^{\text{{fa}}}P_{q}^{\text{{fa}}},k})_{k=1}^{N},$
and $\boldsymbol{{\omega}}_{\wh{\tau}_{q}}^{\max}\triangleq(\omega_{\wh{\tau}_{q},k}^{\max})_{k=1}^{N}$,
$\boldsymbol{{\omega}}_{P_{q}^{\text{{fa}}}}^{\max}\triangleq(\omega_{P_{q}^{\text{{fa}}},k}^{\max})_{k=1}^{N}$,
$\boldsymbol{{\omega}}_{\wh{\tau}_{q}\, P_{q}^{\text{{fa}}}}^{\max}\triangleq(\omega_{\wh{\tau}_{q}\, P_{q}^{\text{{fa}}},k}^{\max})_{k=1}^{N}$.
Finally, we introduce the column vector $G_{P,q}\triangleq\left(|{G}_{P,q}(k)|^{2}\right)_{k=1}^{N}$
of the cross-channel transfer function between the secondary transmitter
$q$ and the PU, and the notation $\mathbf{a}\odot\mathbf{b}\triangleq(a_{k}\cdot b_{k})_{k=1}^{N}$
for given $\mathbf{a}\triangleq(a_{k})_{k=1}^{N}$ and $\mathbf{b}\triangleq(b_{k})_{k=1}^{N}$.
Then, matrix $\nabla_{\mathbf{x}_{q}}^{2}I_{q}(\mathbf{x}_{q})$\emph{
}can be written as\textcolor{black}{\small{} 
\begin{equation}
\begin{array}{l}
\nabla_{\mathbf{x}_{q}}^{2}I_{q}(\mathbf{x}_{q})=2\,\left[\begin{array}{ccc}
\mathbf{1}^{T}\mbox{vect}\left(\boldsymbol{{\omega}}_{\wh{\tau}_{q}\wh{\tau}_{q}}\odot\mathbf{G}_{P,q}\odot\mathbf{p}_{q}\right), & \mbox{vect}{\displaystyle {\displaystyle \left(\boldsymbol{{\omega}}_{\wh{\tau}_{q}}\odot\mathbf{G}_{P,q}\right)^{T},}} & \mathbf{1}^{T}\mbox{vect}\left(\boldsymbol{{\omega}}_{\wh{\tau}_{q}P_{fa}^{\,(q)}}\odot\mathbf{G}_{P,q}\odot\mathbf{p}_{q}\right)\\
{\displaystyle {\displaystyle {\displaystyle {\displaystyle \mbox{vect}\left(\boldsymbol{{\omega}}_{\wh{\tau}_{q}}\odot\mathbf{G}_{P,q}\right)}},}} & {\displaystyle \mathbf{0}_{N\times N}}, & \mbox{vect}\left(\boldsymbol{{\omega}}_{P_{fa}^{\,(q)}}\odot\mathbf{G}_{P,q}\right)\\
\mathbf{1}^{T}\mbox{vect}\left(\boldsymbol{{\omega}}_{\wh{\tau}_{q}P_{fa}^{\,(q)}}\odot\mathbf{G}_{P,q}\odot\mathbf{p}_{q}\right), & {\displaystyle \mbox{vect}\left(\boldsymbol{{\omega}}_{P_{fa}^{\,(q)}}\odot\mathbf{G}_{P,q}\right)^{T},} & \mathbf{1}^{T}\mbox{vect}\left(\boldsymbol{{\omega}}_{P_{fa}^{\,(q)}P_{fa}^{\,(q)}}\odot\mathbf{G}_{P,q}\odot\mathbf{p}_{q}\right)
\end{array}\right].\end{array}\vspace{-0.3cm}\label{eq:sum_nabla_q}
\end{equation}
}{\small \par}

Based on (\ref{eq:sum_nabla_q}), let us introduce the matrix $\left[\nabla_{\mathbf{x}_{q}}^{2}I_{q}(\mathbf{x}_{q})\right]_{\text{{off}}}$obtained
from $\nabla_{\mathbf{x}_{q}}^{2}I_{q}(\mathbf{x}_{q})$ by setting
to zero the diagonal terms ($\left[\mathbf{A}\right]_{\text{{off}}}$
denotes the off-diagonal part of the matrix $\mathbf{A}$) and take
an upper bound of its off-diagonal entries (the inequalities below
have to be intended component-wise): \textcolor{black}{\small 
\begin{equation}
\begin{array}{l}
\left[\nabla_{\mathbf{x}_{q}}^{2}I_{q}(\mathbf{x}_{q})\right]_{\text{{off}}}\triangleq2\,\left[\begin{array}{ccc}
0, & \mbox{vect}{\displaystyle {\displaystyle \left(\boldsymbol{{\omega}}_{\wh{\tau}_{q}}\odot\mathbf{G}_{P,q}\right)^{T},}} & \mathbf{1}^{T}\mbox{vect}\left(\boldsymbol{{\omega}}_{\wh{\tau}_{q}P_{fa}^{\,(q)}}\odot\mathbf{G}_{P,q}\odot\mathbf{p}_{q}\right)\\
{\displaystyle {\displaystyle {\displaystyle {\displaystyle \mbox{vect}\left(\boldsymbol{{\omega}}_{\wh{\tau}_{q}}\odot\mathbf{G}_{P,q}\right)}},}} & {\displaystyle \mathbf{0}_{N\times N}}, & \mbox{vect}\left(\boldsymbol{{\omega}}_{P_{fa}^{\,(q)}}\odot\mathbf{G}_{P,q}\right)\\
\mathbf{1}^{T}\mbox{vect}\left(\boldsymbol{{\omega}}_{\wh{\tau}_{q}P_{fa}^{\,(q)}}\odot\mathbf{G}_{P,q}\odot\mathbf{p}_{q}\right), & {\displaystyle \mbox{vect}\left(\boldsymbol{{\omega}}_{P_{fa}^{\,(q)}}\odot\mathbf{G}_{P,q}\right)^{T},} & 0
\end{array}\right]\\
\leq{\displaystyle \max_{k}}\left\{ |G_{P,q}(k)|^{2}\right\} \cdot\underset{\triangleq\left[\nabla_{\mathbf{x}_{q}}^{2}I_{q}\right]_{\text{{off}}}^{\text{{up}}}}{\underbrace{{\displaystyle \,2\,}\left[\begin{array}{ccc}
0, & {\displaystyle {\displaystyle \mbox{vect}\left(\boldsymbol{{\omega}}_{\wh{\tau}_{q}}^{\max}\right)^{T},}} & \mathbf{1}^{T}\mbox{vect}\left(\boldsymbol{{\omega}}_{\wh{\tau}_{q}P_{q}^{\text{{fa}}}}^{\max}\odot\mathbf{p}_{q}^{\max}\right)\\
{\displaystyle {\displaystyle {\displaystyle {\displaystyle \mbox{vect}\left(\boldsymbol{{\omega}}_{\wh{\tau}_{q}}^{\max}\right)}},}} & {\displaystyle \mathbf{0}_{N\times N}}, & \mbox{vect}\left(\boldsymbol{{\omega}}_{P_{q}^{\text{{fa}}}}^{\max}\right)\\
\mathbf{1}^{T}\mbox{vect}\left(\boldsymbol{{\omega}}_{\wh{\tau}_{q}P_{q}^{\text{{fa}}}}^{\max}\odot\mathbf{p}_{q}^{\max}\right), & {\displaystyle \mbox{vect}\left(\boldsymbol{{\omega}}_{P_{q}^{\text{{fa}}}}^{\max}\right)^{T},} & 0
\end{array}\right]}}\\
\triangleq{\displaystyle \max_{k}}\left\{ |G_{P,q}(k)|^{2}\right\} \cdot\left[\nabla_{\mathbf{x}_{q}}^{2}I_{q}\right]_{\text{{off}}}^{\text{{up}}}.
\end{array}\label{eq:H_pq_up_and_S_p_q}
\end{equation}
}{\small \par}

\noindent$-$Matrix $\nabla_{\mathbf{x}_{q}}^{2}I(\mathbf{x})$.
Following similar steps as for (\ref{eq:H_pq_up_and_S_p_q}), we obtain
\begin{equation}
\left[\left|\nabla_{\mathbf{x}_{q}}^{2}I(\mathbf{x})\right|\right]_{\text{{off}}}\leq{\displaystyle \max_{k}}\left\{ |G_{P,q}(k)|^{2}\right\} \cdot\left[\nabla_{\mathbf{x}_{q}}^{2}I_{q}\right]_{\text{{off}}}^{\text{{up}}}.\label{eq:H_pq_up_and_S_p_q_2}
\end{equation}

We are now ready to introduce the matrix $\overline{\nabla_{\mathbf{x}_{q}}^{2}\mathcal{L}_{q}}$
satisfying (\ref{eq:nabla_L_bar}). Given $t>0$, and the definitions
in (\ref{eq:der_Pmiss_wrt_tau})-(\ref{eq:sec_der_Pmiss_wrt_tau_pfa})
and (\ref{eq:H_pq_up_and_S_p_q_2}), we define\vspace{-0.6cm}

\begin{equation}
\overline{\nabla_{\mathbf{x}_{q}}^{2}\mathcal{L}_{q}}\triangleq\text{{Diag}}\left\{ (d_{\wh{\tau}_{q}}^{\min},\,(d_{-\log r_{q},k}^{\min})_{k=1}^{N},\, d_{P_{q}^{\text{{fa}}}}^{\min})\right\} -2\cdot\max\left\{ t,\,\lambda^{\max}\right\} \cdot{\displaystyle \max_{k}}\left\{ |G_{P,q}(k)|^{2}\right\} \cdot\left[\nabla_{\mathbf{x}_{q}}^{2}I_{q}\right]_{\text{{off}}}^{\text{{up}}}\label{eq:def_L_2_inf}
\end{equation}

\noindent \textbf{Step 2}. It follows from Step 1 that $\overline{\nabla_{\mathbf{x}_{q}}^{2}\mathcal{L}_{q}}$
in (\ref{eq:def_L_2_inf}) satisfies the desired property (\ref{eq:copositivity_Lower_bound}).
Condition (\ref{eq:diagonal_dominance_SF_cond}) of the corollary
is readily obtained by imposing that $\overline{\nabla_{\mathbf{x}_{q}}^{2}\mathcal{L}_{q}}$
is row-diagonal dominant, and setting 

\begin{equation}
\gamma_{q}^{(1)}=\frac{2\,\max(t,\,\lambda^{\max})}{\min\left\{ \dfrac{d_{\wh{\tau}_{q}}^{\min}}{\sum_{j}\left[[\nabla_{\mathbf{x}_{q}}^{2}I_{q}]_{\text{{off}}}^{\text{{up}}}\right]_{1j}},{\displaystyle {\displaystyle {\min_{i=1,\ldots,N}}}}\left\{ \dfrac{{d_{q,i}^{\min}}}{\sum_{j}\left[[\nabla_{\mathbf{x}_{q}}^{2}I_{q}]_{\text{{off}}}^{\text{{up}}}\right]_{ij}}\right\} ,\dfrac{d_{P_{q}^{\text{{fa}}}}^{\min}}{\sum_{j}\left[[\nabla_{\mathbf{x}_{q}}^{2}I_{q}]_{\text{{off}}}^{\text{{up}}}\right]_{N+2\, j}}\right\} }.\label{eq:omega_q}
\end{equation}
\hfill $\square$

\section{Proof of Theorem \ref{Theo_Existence-and-uniqueness_NE_G}\label{proof_Theo_Existence-and-uniqueness_NE_G}}

To prove the theorem we need the following lemma whose proof follows
the same idea of that in Lemma \ref{Lemma_bounded_multipliers} and
thus is omitted. 

\begin{lemma}\label{Lemma_bounded_multipliers_Game_t}Let $t>\lambda^{\max}$,
with $\lambda^{\max}$ defined in (\ref{eq:t_star_def}). Then, at
every solution $(\mathbf{x}^{\star},\,{\boldsymbol{{\lambda}}}^{\star},{\pi}_{t}^{\star})$
of the VI$(\mathcal{Z}_{t},\boldsymbol{{\Psi}})$ defined in (\ref{eq:KKT_Game_Gt})\emph{
{[}}stationary solution of $\mathcal{G}_{t}(\mathcal{X},\,\boldsymbol{{\theta}})$\emph{{]}},
the price constraints (\ref{eq:simplex_price}) are not binding, i.e.,
$\pi_{t}^{\star}<t.$\end{lemma}

\noindent\textbf{Proof of Theorem \ref{Theo_Existence-and-uniqueness_NE_G}.}
We prove only statement (a); the proof of the second part (b) follows
similar steps of those in the proof of Proposition \ref{proposition_uniqueness_opt_sol}
and thus is omitted. Given $t>t^{\star}$, under the assumptions in
(a), Proposition \ref{Proposition_existence of a NE of G_t} states
that the game $\mathcal{G}_{t}(\mathcal{X},\,\boldsymbol{{\theta}})$
admits a nontrivial NE $(\mathbf{x}^{\star},{\pi}_{t}^{\star})$;
by Lemma \ref{Lemma_ACQ}, there exist multipliers ${\boldsymbol{{\lambda}}}^{\star}$
such that $(\mathbf{x}^{\star},{\boldsymbol{{\lambda}}}^{\star},{\pi}_{t}^{\star})$
satisfies the VI$(\mathcal{Z}_{t},\boldsymbol{{\Psi}})$ in (\ref{eq:KKT_Game_Gt})
{[}or equivalently (\ref{eq:KKT_game_G_t_1}){]}. Lemma \ref{Lemma_bounded_multipliers_Game_t}
shows that the upper bound constraint on the price in $\mathcal{S}_{t}$
is not binding at $(\mathbf{x}^{\star},{\boldsymbol{{\lambda}}}^{\star},{\pi}_{t}^{\star})$,
implying from iii) of (\ref{eq:KKT_game_G_t_1}) that $\eta_{t}^{\star}=0$
and thus $0\leq\,{\pi}_{t}^{\star}\,\perp\,-I(\mathbf{x}^{\star})\geq0$.
Hence, $(\mathbf{x}^{\star},{\pi}_{t}^{\star})$ must be a NE of the
original un-truncated game $\mathcal{G}(\mathcal{X},\,\boldsymbol{\theta})$
{[}recall that, under the positive definiteness of the matrices $\nabla_{\mathbf{x}_{q}}^{2}\mathcal{L}_{q}(\mathbf{x},\,{\pi}_{t},\,\lambda_{q})$
on ${\mathcal{Y}}\times\mathcal{S}_{t}\times[0,\lambda^{\max}]$,
each optimization problem in (\ref{eq:game_G_t_2}), with $\mathbf{x}_{-q}=\mathbf{x}_{-q}^{\star}$
and $\pi_{t}={\pi}_{t}^{\star}$, has a unique stationary (and thus
optimal) solution, which then must be equal to $\mathbf{x}_{q}^{\star}$;
see Proposition \ref{proposition_uniqueness_opt_sol}{]}. \hfill $\square$

\section{Proof of Corollary \ref{corollary_sf_cond_uniqueness_NE}\label{sec:Proof-of-corollary_sf_cond_uniqueness_NE}}

In order to obtain more general conditions than those in Theorem \ref{Theo_Existence-and-uniqueness_NE_G},
by Lemma \ref{Lemma_lower_bound_on_the_power}, we can restrict the
check of the positive definiteness of the matrices $\nabla_{\mathbf{x}_{q}}^{2}\mathcal{L}_{q}(\mathbf{x},\,{\pi}_{t},\,\lambda_{q})$
and $\mathbf{A}(\mathbf{x},\,\boldsymbol{{\lambda}},\,{\pi}_{t})$
as required in Theorem \ref{Theo_Existence-and-uniqueness_NE_G} to
the subset of the feasible set where any solution of the game lies.
More specifically, let us introduce the restriction of the sets $\mathcal{P}_{q}$
and $\mathcal{Y}_{q}$ defined in (\ref{set_P_q}) and (\ref{eq:def_Y_q}),
respectively, to the power allocations satisfying (\ref{eq:lower_bound_p_opt}):
given $t>0$,
\begin{align}
\wh{\mathcal{P}}_{q}^{t} & \triangleq\left\{ \mathbf{p}\in\mathcal{P}_{q}\,:\,{\displaystyle {\sum_{k=1}^{N}}}\, p_{q,k}\,\geq\,\left(\min_{k}\left\{ \wh{\sigma}{}_{q,k}^{2}\right\} \right)\,\exp\left(\eta_{q}^{\text{{ref}}}(t)\right)\right\} ,\quad q=1,\ldots,Q,\label{eq:set_P_q_t_a}\\
\wh{\mathcal{Y}}{}^{t} & \triangleq\prod_{q}\wh{\mathcal{Y}}{}_{q}^{t},\label{eq:set_P_q_t_b}
\end{align}
where $\wh{\mathcal{Y}}{}_{q}^{t}$ is defined as $\mathcal{Y}_{q}$
in (\ref{eq:def_Y_q}), but with $\mathcal{P}_{q}$ replaced by $\wh{\mathcal{P}}_{q}^{t}$.
By Lemma \ref{Lemma_lower_bound_on_the_power}, instead of checking
the positive definiteness of $\nabla_{\mathbf{x}_{q}}^{2}\mathcal{L}_{q}(\mathbf{x},\,{\pi}_{t},\,\lambda_{q})$
and $\mathbf{A}(\mathbf{x},\,\boldsymbol{{\lambda}},\,{\pi}_{t})$
on the feasible set $\mathcal{Y}\times\mathcal{S}_{t}\times[0,\lambda^{\max}]$,
we can restrict this requirement to the subset $\wh{\mathcal{Y}}{}^{t}\times\mathcal{S}_{t}\times[0,\lambda^{\max}]$. 

We can now prove the corollary. We show next that (\ref{eq:diagonal_dominance_A_pd})
are sufficient conditions for the matrix $\mathbf{A}(\mathbf{x},\,\boldsymbol{{\lambda}},\,{\pi}_{t})$
to be positive definite on $\mathcal{Y}\times\mathcal{S}_{t}\times[0,\lambda^{\max}]$.
Fist of all, observe that matrix $\mathbf{A}(\mathbf{x},\,\boldsymbol{{\lambda}},\,{\pi}_{t})$
can be written as 
\begin{align}
\mathbf{A}(\mathbf{x},\,\boldsymbol{{\lambda}},\,{\pi}_{t}) & \triangleq\underset{\triangleq\left.\mathbf{A}(\mathbf{x},\,\boldsymbol{{\lambda}},\,{\pi}_{t})\right|_{c=0}}{\underbrace{\left[\begin{array}{cccc}
\left.\nabla_{\mathbf{x}_{1}}^{2}\mathcal{L}_{1}\right|_{c=0}, & \left.\nabla_{\mathbf{x}_{1}\mathbf{x}_{2}}^{2}\theta_{1}\right|_{c=0} & \cdots & \left.\nabla_{\mathbf{x}_{1}\mathbf{x}_{Q}}^{2}\theta_{1}\right|_{c=0}\\
\vdots & \cdots & \ddots & \vdots\\
\left.\nabla_{\mathbf{x}_{Q}\mathbf{x}_{1}}^{2}\theta_{Q}\right|_{c=0}, & \cdots & \left.\nabla_{\mathbf{x}_{Q}\mathbf{x}_{Q-1}}^{2}\theta_{Q}\right|_{c=0} & \left.\nabla_{\mathbf{x}_{Q}}^{2}\mathcal{L}_{Q}\right|_{c=0}
\end{array}\right]}}\label{eq:matrix_A_eq_}\\
 & +c\,(1-1/Q)\,\underset{\mbox{up to a permutation}}{\underbrace{\left[\begin{array}{cc}
\mathbf{D}_{f}^{-1}\left(\mathbf{I}_{Q}-\dfrac{{\mathbf{1}\mathbf{1}^{T}}}{Q}\right)\mathbf{D}_{f}^{-1} & \mathbf{0}\\
\mathbf{0} & \mathbf{0}
\end{array}\right]}},\label{eq:matrix_A_eq_2}
\end{align}
where $\mathbf{D}_{f_{s}}\triangleq\text{{diag}}\left\{ \left(\sqrt{{f_{q}}}\right)_{q=1}^{Q}\right\} $.
Since the matrix in (\ref{eq:matrix_A_eq_2}) is positive semidefinite,
we can focus only on $\left.\mathbf{A}(\mathbf{x},\,\boldsymbol{{\lambda}},\,{\pi}_{t})\right|_{c=0}$.
To obtain a sufficient condition for $\left.\mathbf{A}(\mathbf{x},\,\boldsymbol{{\lambda}},\,{\pi}_{t})\right|_{c=0}$
to be positive definite on $\wh{\mathcal{Y}}{}^{t}\times\mathcal{S}_{t}\times[0,\lambda^{\max}]$,
we follow a similar idea of that in Corollary \ref{corollary_sf_cond_uniqueness_opt_sol}.
Namely, we build a proper matrix $\overline{{\mathbf{A}}}$ such that,
for some $t>\lambda^{\max}$, 
\begin{equation}
\mathbf{y}^{T}\,\left(\left.\mathbf{A}(\mathbf{x},\,\boldsymbol{{\lambda}},\,{\pi}_{t})\right|_{c=0}\right)\,\mathbf{y}\geq\left|\mathbf{y}\right|^{T}\overline{{\mathbf{A}}}\left|\mathbf{y}\right|\quad\forall(\mathbf{x},\,{\pi}_{t},\,\boldsymbol{{\lambda}})\in\wh{\mathcal{Y}}{}^{t}\times\mathcal{S}_{t}\times[0,\lambda^{\max}],\quad\mbox{and}\quad\mathbf{y}\in\mathbb{R}^{Q\,(N+2)}.\label{eq:pd_with_A_bar}
\end{equation}
To this end, we focus on each term in (\ref{eq:matrix_A_eq_}) separately
and derive proper bounds. 

\noindent$-$Matrix $\left|\left.\nabla_{\mathbf{x}_{q}\mathbf{x}_{r}}^{2}\theta_{q}\right|_{c=0}\right|$.
Recalling the definition of $r_{q}(\mathbf{p})\triangleq\sum_{k}r_{q,k}(\mathbf{p})$,
with $r_{q,k}(\mathbf{p})$ given in (\ref{eq:rate_FSIC}), we have
\begin{equation}
\left.\nabla_{\mathbf{x}_{q}\mathbf{x}_{r}}^{2}\theta_{q}\right|_{c=0}=\left[\begin{array}{ccc}
0 & \mathbf{0}_{1\times N} & 0\\
\mathbf{0}_{N\times1} & \nabla_{\mathbf{p}_{q}\mathbf{p}_{r}}^{2}\left(-\log r_{q}(\mathbf{p})\right) & \mathbf{0}_{N\times1}\\
0 & \mathbf{0}_{1\times N} & 0
\end{array}\right],\label{eq:nabla_x_q_xr_f_q}
\end{equation}
with 
\begin{equation}
\nabla_{\mathbf{p}_{q}\mathbf{p}_{r}}^{2}\left(-\log r_{q}(\mathbf{p})\right)=\left[{\displaystyle {\frac{-\nabla_{\mathbf{p}_{q}\mathbf{p}_{r}}^{2}r_{q}(\mathbf{p})}{r_{q}(\mathbf{p}_{q})}}+{\displaystyle {\frac{\nabla_{\mathbf{p}_{q}}r_{q}(\mathbf{p})\,\nabla_{\mathbf{p}_{r}}r_{q}(\mathbf{p})^{T}}{r_{q}(\mathbf{p}_{q})^{2}}}}}\right],\label{eq:nabla_p_q_p_r_minus_log_r_q}
\end{equation}
$\nabla_{\mathbf{p}_{q}}r_{q}(\mathbf{p})$ given in (\ref{eq:nablas_of_r_q_a})
and 
\begin{align}
\nabla_{\mathbf{p}_{r}}r_{q}(\mathbf{p}) & =\mbox{vect}\left\{ \left(\,{\displaystyle {\frac{-|\wh{H}_{qr}(k)|^{2}\, p_{q,k}}{\left(\wh{\sigma}_{q,k}^{2}+{\displaystyle {\sum_{r=1}^{Q}}\,|\wh{H}_{qr}(k)|^{2}\, p_{r,k}}\right)\left(\wh{\sigma}_{q,k}^{2}+{\displaystyle {\sum_{r\neq q}}\,|\wh{H}_{qr}(k)|^{2}\, p_{r,k}}\right)}}\,}\right)_{k=1}^{N}\right\} ,\label{eq:nabla_r_r_q_}\\
\nabla_{\mathbf{p}_{q}\mathbf{p}_{r}}^{2}r_{q}(\mathbf{p}) & =\mbox{Diag}\left\{ \left(\,{\displaystyle {\frac{-|\wh{H}_{qr}(k)|^{2}}{\left(\,\wh{\sigma}_{q,k}^{2}+{\displaystyle {\sum_{r=1}^{Q}}\,|\wh{H}_{qr}(k)|^{2}\, p_{r,k}\,}\right)^{2}}}\,}\right)_{k=1}^{N}\right\} .\label{eq:eq:nabla_square_p_q_p_r_of_r_q_}
\end{align}
Using the following lower bound for the rate function $r_{q}(\mathbf{p})$:
given $t>0$ and $\mathbf{p}_{q}\in\wh{\mathcal{P}}{}_{q}^{t}$, 
\begin{align}
r_{q}(\mathbf{p}) & \geq\left(\sum_{k=1}^{N}p_{q,k}\right)\cdot\underset{\triangleq r_{q}^{\min}}{\underbrace{\min_{k=1,\ldots,N}\left\{ \log\left(1+\dfrac{{p_{q,k}^{\max}}}{\wh{\sigma}_{q,k}^{2}+{\displaystyle {\sum_{r\neq q}}\,|\wh{H}_{qr}(k)|^{2}\, p_{r,k}^{\max}}}\right)\right\} }}\label{eq:lower_bound_r_q}\\
 & \geq\left(\min_{k=1,\ldots,N}\left\{ \wh{\sigma}{}_{q,k}^{2}\right\} \right)\cdot\exp\left(\eta_{q}^{\text{{ref}}}(t)\right)\cdot r_{q}^{\min}\triangleq r_{q}^{\text{{low}}}(t),\label{eq:lower_bound_r_q_2}
\end{align}
where the second inequality follows from Lemma \ref{Lemma_lower_bound_on_the_power},
we have for $\left|\nabla_{\mathbf{p}_{q}\mathbf{p}_{r}}^{2}\left(-\log r_{q}(\mathbf{p})\right)\right|$:
given $t>0$, $\mathbf{p}_{q}\in\wh{\mathcal{P}}_{q}^{t}$ and $\mathbf{p}_{r}\in[\mathbf{0},\,\mathbf{p}_{r}^{\max}]$
with $r\neq q$,{\small 
\begin{equation}
\begin{array}{l}
\left|\nabla_{\mathbf{p}_{q}\mathbf{p}_{r}}^{2}\left(-\log r_{q}(\mathbf{p})\right)\right|\leq{\displaystyle {\frac{1}{r_{q}(\mathbf{p})}}\,\mbox{Diag}\left\{ \left(\,{\displaystyle {\frac{|\wh{H}_{qr}(k)|^{2}}{\left(\,\wh{\sigma}_{q,k}^{2}\right)^{2}}}\,}\right)_{k=1}^{N}\right\} }{\displaystyle +{\frac{1}{r_{q}(\mathbf{p})^{2}}}\,\mbox{vect}\left\{ \left(\,{\displaystyle {\frac{1}{\wh{\sigma}_{q,k}^{2}}}\,}\right)_{k=1}^{N}\right\} }\cdot\mbox{vect}\left\{ \left(\,{\displaystyle {\frac{|\wh{H}_{qr}(k)|^{2}\, p_{q,k}}{\left(\wh{\sigma}_{q,k}^{2}\right)^{2}}}\,}\right)_{k=1}^{N}\right\} ^{T}\\
\leq{\displaystyle {\frac{1}{r_{q}^{\text{{low}}}(t)}}\,\left[\mbox{Diag}\left\{ \left(\,{\displaystyle {\frac{|\wh{H}_{qr}(k)|^{2}}{\left(\wh{\sigma}_{q,k}^{2}\right)^{2}}}\,}\right)_{k=1}^{N}\right\} +\dfrac{{1}}{r_{q}^{\min}}{\displaystyle \,\mbox{vect}\left\{ \left(\,{\displaystyle {\frac{1}{\wh{\sigma}_{q,k}^{2}}}\,}\right)_{k=1}^{N}\right\} \cdot}\mbox{vect}\left\{ \left(\,{\displaystyle {\frac{|\wh{H}_{qr}(k)|^{2}}{\left(\wh{\sigma}_{q,k}^{2}\right)^{2}}}\,}\right)_{k=1}^{N}\right\} ^{T}\right]}\\
\triangleq\left[\nabla_{\mathbf{p}_{q}\mathbf{p}_{r}}^{2}\theta_{q}\right]^{\text{up}},
\end{array}\label{eq:def_nabla_p_q_p_r_f_q}
\end{equation}
}which leads also to 
\begin{equation}
\left\Vert \nabla_{\mathbf{p}_{q}\mathbf{p}_{r}}^{2}\left(-\log r_{q}(\mathbf{p})\right)\right\Vert \leq{\max_{k=1,\ldots,N}}\left\{ {\displaystyle {\frac{|\wh{H}_{qr}(k)|^{2}}{\wh{\sigma}_{q,k}^{4}}}}\right\} \cdot\underset{\triangleq\xi_{q}^{\sup}(t)}{\underbrace{\left({\displaystyle \dfrac{1}{r_{q}^{\text{{low}}}(t)}+\dfrac{1}{r_{q}^{\text{{low}}}(t)}\cdot\dfrac{1}{r_{q}^{\min}}\cdot}{\displaystyle {\max_{k=1,\ldots,N}}\left\{ \frac{1}{\wh{\sigma}_{q,k}^{4}}\right\} }\right)}}.\label{eq:upper_spectrum_nabla_f_q_p_q_p_r}
\end{equation}

Using $\overline{\nabla_{\mathbf{x}_{q}}^{2}\mathcal{L}_{q}}$\textcolor{red}{{}
}defined in (\ref{eq:def_L_2_inf}), we are now ready to introduce
the matrix $\overline{{\mathbf{A}}}$, defines as: given $t>0$, 
\begin{equation}
\overline{{\mathbf{A}}}\triangleq\left(\overline{{\mathbf{A}}}_{qr}\right)_{q,r=1}^{Q}\quad\mbox{with}\quad\overline{{\mathbf{A}}}_{qr}\triangleq\left\{ \begin{array}{ll}
\overline{\nabla_{\mathbf{x}_{q}}^{2}\mathcal{L}_{q}},\quad & \mbox{if }q=r,\\
-\text{{Diag}}\left\{ \left[0,\,\left[\nabla_{\mathbf{p}_{q}\mathbf{p}_{r}}^{2}\theta_{q}\right]^{\text{up}},\,0\right]\right\} ,\,\, & \mbox{otherwise},
\end{array}\right.\label{eq:def_A_comparison}
\end{equation}
which satisfies the desired property in (\ref{eq:pd_with_A_bar}). 

A sufficient condition for (\ref{eq:pd_with_A_bar}) can be obtained
as in (\ref{eq:diagonal_dominance_A_pd}), by imposing that (the symmetric
part of) $\overline{{\mathbf{A}}}$ is row diagonal dominant. More
specifically, introducing 
\begin{equation}
\zeta(t)\triangleq\max_{q=1,\ldots,Q}\left\{ {\max_{k=1,\ldots,N}}\left\{ \frac{1}{\wh{\sigma}_{q,k}^{2}}\right\} \cdot\left({\displaystyle \dfrac{1}{r_{q}^{\text{{low}}}(t)}+\dfrac{1}{r_{q}^{\text{{low}}}(t)}\cdot\dfrac{1}{r_{q}^{\min}}\cdot}{\sum_{k^{'}=1}^{N}\frac{1}{\wh{\sigma}_{q,k^{'}}^{2}}}\right)\right\} \label{eq:zita_def}
\end{equation}
the diagonal dominance conditions is: for each $q=1,\ldots,Q$ and
$i=1,\ldots,N$, 
\begin{equation}
\dfrac{{1}}{2}\sum_{r\neq q}\sum_{j=1}^{N}\left[\left[\nabla_{\mathbf{p}_{q}\mathbf{p}_{r}}^{2}\theta_{q}\right]^{\text{up}}+\left(\left[\nabla_{\mathbf{p}_{q}\mathbf{p}_{r}}^{2}\theta_{q}\right]^{\text{up}}\right)^{T}\right]_{ij}\leq\dfrac{\zeta(t)}{2}\,\sum_{r\neq q}\left({\max_{k=1,\ldots,N}}\left\{ {\displaystyle {\frac{|\wh{H}_{qr}(k)|^{2}}{\wh{\sigma}_{q,k}^{2}}}}\right\} +{\max_{k=1,\ldots,N}}\left\{ {\displaystyle {\frac{|\wh{H}_{rq}(k)|^{2}}{\wh{\sigma}_{r,k}^{2}}}}\right\} \right).\label{eq:row_sum_upper_bound}
\end{equation}
After substituting the explicit expression of $\left[\nabla_{\mathbf{p}_{q}\mathbf{p}_{r}}^{2}\theta_{q}\right]^{\text{up}}$
and doing some manipulations, (\ref{eq:row_sum_upper_bound}) leads
to the desired condition (\ref{eq:diagonal_dominance_A_pd}), where
we defined $\gamma_{q}^{(2)}$ as
\begin{equation}
\gamma_{q}^{(2)}\triangleq\zeta_{q}^{\max}(t)\cdot\gamma_{q}^{(1)}\label{eq:gamma_2}
\end{equation}
with $\gamma_{q}^{(1)}$ given in (\ref{eq:omega_q}) and 
\begin{equation}
\zeta_{q}^{\max}(t)\triangleq\dfrac{{\zeta(t)}}{2\, t}\cdot\frac{1}{\min\left\{ \sum_{j}\left[[\nabla_{\mathbf{x}_{q}}^{2}I_{q}]_{\text{{off}}}^{\text{{up}}}\right]_{1j},\,{\displaystyle {\displaystyle {\min_{i=1,\ldots,N}}}}\left\{ \sum_{j}\left[[\nabla_{\mathbf{x}_{q}}^{2}I_{q}]_{\text{{off}}}^{\text{{up}}}\right]_{ij}\right\} ,\,\sum_{j}\left[[\nabla_{\mathbf{x}_{q}}^{2}I_{q}]_{\text{{off}}}^{\text{{up}}}\right]_{N+2\, j}\right\} },\label{eq:zita_q_max}
\end{equation}
where $[\nabla_{\mathbf{x}_{q}}^{2}I_{q}]_{\text{{off}}}^{\text{{up}}}$
and ${\zeta(t)}$ are defined in (\ref{eq:H_pq_up_and_S_p_q}) and
(\ref{eq:zita_def}), respectively.

\section{Convergence of Asynchronous Best-Response Algorithms for $\mathcal{G}_{\pi}(\mathcal{X},\boldsymbol{{\theta}})$\label{sec:Convergence-of-Asynchronous_BR_local_constraints}}

In this section, we study the convergence of asynchronous best-response
algorithms solving the game $\mathcal{G}_{\pi}(\mathcal{X},\boldsymbol{{\theta}})$
in (\ref{eq:game_G_t_pi}); an instance of such algorithms is represented
by Algorithm \ref{async_best-response_algo}. Since the study of convergence
is based on contraction arguments of the best-response map associated
with game $\mathcal{G}_{\pi}(\mathcal{X},\boldsymbol{{\theta}})$,
we derive first sufficient conditions for this best-response to be
a contraction; see Sec. \ref{sub:Contraction-properties-of_BR_of_G_pi}.
We then provide the main theorem stating convergence of the asynchronous
best-response algorithms; see Sec. \ref{sub:Asynchronous-convergence-theorem}.

\subsection{Contraction properties of the best-response of $\mathcal{G}_{\pi}(\mathcal{X},\boldsymbol{{\theta}})$\label{sub:Contraction-properties-of_BR_of_G_pi}}

Before introducing the main result of this section, we need the following
intermediate definitions. Given $\mathcal{L}_{q}$ defined in (\ref{eq:Lagrangian-1}),
let $\mathbf{B}_{q}\left(\mathbf{x},\lambda_{q},\pi_{t}\right)$ be
the $2\times2$ matrix, defined as 
\begin{equation}
\mathbf{B}_{q}\left(\mathbf{x},\lambda_{q},\pi_{t}\right)\triangleq\left[\begin{array}{ll}
\left.\nabla_{\wh{\tau}_{q}}^{2}\mathcal{L}_{q}\left(\mathbf{x},\,{\pi}_{t},\,\lambda_{q}\right)\right|_{c=0}, & -\left\Vert \nabla_{\wh{\tau}_{q}\,(\mathbf{p}_{q},P_{q}^{\text{{fa}}})}^{2}\mathcal{L}_{q}\left(\mathbf{x},\,{\pi}_{t},\,{\lambda_{q}}\right)\right\Vert \medskip\\
-\left\Vert \nabla_{(\mathbf{p}_{q},P_{q}^{\text{{fa}}})\,\wh{\tau}_{q}}^{2}\mathcal{L}_{q}\left(\mathbf{x},\,\pi_{t},\,\lambda_{q}\right)\right\Vert , & \lambda_{\text{{least}}}\left(\nabla_{(\mathbf{p}_{q},P_{q}^{\text{{fa}}})}^{2}\mathcal{L}_{q}\left(\mathbf{x},\,\pi_{t},\,\lambda_{q}\right)\right)
\end{array}\right],\label{eq:B_q_def}
\end{equation}
where $\left\Vert \mathbf{A}\right\Vert \triangleq\rho\left(\mathbf{A}^{T}\mathbf{A}\right)^{1/2}$
and $\lambda_{\text{{least}}}(\mathbf{B})$ denote the spectral norm
of $\mathbf{A}$ and the minimum eigenvalue of the symmetric matrix
$\mathbf{B}$, respectively. Given $t>0$ and $\wh{\mathcal{Y}}{}^{t}$
as defined in (\ref{eq:set_P_q_t_a}) (cf. Appendix \ref{sec:Proof-of-corollary_sf_cond_uniqueness_NE}),
we also introduce 

\begin{equation}
\rho_{q}(t)\triangleq\left\{ \begin{array}{ll}
\overline{{\rho}}_{q}(t)\triangleq\underset{\begin{array}{c}
\tiny\,(\mathbf{x}_{q},\,\lambda_{q})\in\wh{\mathcal{Y}}_{q}^{t}\times[0,\lambda^{\max}]\vspace{-0.2cm}\\
\tiny\,(\mathbf{x}_{-q},\pi_{t})\in{\mathcal{Y}}_{-q}\times\mathcal{S}_{t}
\end{array}}{\min}\left\{ \text{\ensuremath{\lambda}}_{\text{{least}}}\left(\mathbf{B}_{q}\left(\mathbf{x},\lambda_{q},\pi_{t}\right)\right)\right\} ,\quad & \mbox{if }\overline{{\rho}}_{q}(t)\geq0,\\
0, & \mbox{otherwise};
\end{array}\right.\label{eq:minimum_eigenvalue_original_matrix}
\end{equation}
and the diagonal matrices $\mathbf{D}_{q}(t,\, c)$ and $\mathbf{E}_{qr}\left(\mathbf{x}\right)$
\begin{equation}
\mathbf{D}_{q}(t,\, c)^{2}\triangleq\left[\begin{array}{cc}
\rho_{q}(t)+c\,\left(\dfrac{{1-1/Q}}{\sqrt{f_{q}}}\right)^{2}, & 0\\
0 & \rho_{q}(t)
\end{array}\right]\,\,\mbox{and}\,\,\mathbf{E}_{qr}\left(\mathbf{x}\right)\triangleq\left[\begin{array}{ll}
\left|\nabla_{\wh{\tau}_{q}\wh{\tau}_{r}}^{2}\theta_{q}\left({\mathbf{x}}\right)\right|, & 0\\
0, & \left\Vert \nabla_{\mathbf{p}_{q}\mathbf{p}_{r}}^{2}\theta_{q}\left({\mathbf{x}}\right)\right\Vert 
\end{array}\right],\label{eq:def_E_and_D_matrices}
\end{equation}
with $\theta_{q}(\cdot)$ defined in (\ref{eq:theta_payoff}). Given
the coefficients 
\begin{equation}
\beta_{qr}(t,\, c)\triangleq\max_{(\mathbf{x}_{q},\,\mathbf{x}_{-q})\in\wh{\mathcal{Y}}_{q}^{t}\times\mathcal{Y}_{-q}}\left\Vert \mathbf{D}_{q}(t,\, c)^{-1}\,\mathbf{E}_{qr}\left(\mathbf{x}\right)\,\mathbf{D}_{r}(t,\, c)^{-1}\right\Vert ,\label{eq:off_diag_Gamma_t}
\end{equation}
for $r,q=1,\ldots,Q$ and $r\neq q,$ we can finally define the $Q\times Q$
matrix  $\boldsymbol{{\Gamma}}(t)$ that plays a key role in studying
contraction properties of the best-response map associated with the
game $\mathcal{G}_{\pi}(\mathcal{X},\boldsymbol{{\theta}})$: 
\begin{equation}
\left[\boldsymbol{{\Gamma}}(t)\right]_{q,r}\triangleq\left\{ \begin{array}{ll}
1, & \mbox{if }r=q,\\
-\beta_{qr}(t,\, c),\quad & \mbox{otherwise. }
\end{array}\right.\label{eq:Gamma_t_matrix_and_Gamma_t_inf}
\end{equation}
It is important to remark here that the off-diagonal entries of the
matrix $\boldsymbol{{\Gamma}}(t)$ depend, among other quantities,
on the cross-channels $\left\{ |\wh{H}_{qr}(k)|^{2}\right\} $ and
$\left\{ |G_{P,q}(k)|^{2}\right\} $. Roughly speaking, this dependence
is such that the $\beta_{qr}(t,\, c)$'s tend to decrease as the aforementioned
cross-channels decrease, meaning that the $\beta_{qr}(t,\, c)$ remains
``small'' as long as the overall MUI in the system remains ``small''.
We will show shortly that this is what one needs to guarantee the
convergence of the distributed best-response based algorithms introduced
in Sec. \ref{sub:Game-with-exogenous}. More formally, by postulating
that $\boldsymbol{{\Gamma}}(t)$ is a P-matrix, Theorem \ref{Theo_Contraction_best-response}
below states the contraction properties of the best-response mapping
of the game $\mathcal{G}_{\pi}(\mathcal{X},\,\boldsymbol{{\theta}})$
with respect to the suitably defined block maximum norm {[}see proof
of the theorem for details{]}. 

\begin{theorem} \label{Theo_Contraction_best-response} Given the
game $\mathcal{G}_{\pi}(\mathcal{X},\,\boldsymbol{{\theta}})$ with
exogenous (fixed) price ${\pi}\geq0$, suppose that $\boldsymbol{{\Gamma}}(t)$
in (\ref{eq:Gamma_t_matrix_and_Gamma_t_inf}) is a P-matrix. Then
the following hold: 
\begin{description}
\item [{(a)}] Each nonconvex optimization problem in (\ref{eq:game_G_t_pi})
has a unique (nontrivial) optimal solution $\overline{{\mathcal{B}}}{}_{q}(\mathbf{x}_{-q})\triangleq\left(\wh{\tau}_{q}^{\star}(\mathbf{x}_{-q}),\,\mathbf{p}_{q}^{\star}(\mathbf{x}_{-q}),\, P_{q}^{\text{{fa}}\star}(\mathbf{x}_{-q})\right)$,
for every given $\mathbf{x}_{-q}\in{\mathcal{Y}}_{-q}$ and ${\pi}\geq0$;
\item [{(b)}] The best-response map ${\mathcal{Y}}\ni\mathbf{x}\rightarrow\overline{{\mathcal{B}}}(\mathbf{x})\triangleq\left(\overline{{\mathcal{B}}}_{q}(\mathbf{x}_{-q})\right)_{q=1}^{Q}$
is a block-contraction; the unique fixed-point of $\overline{{\mathcal{B}}}$
is the unique $\mathbf{x}$-component of the NE of the game.
\end{description}
\end{theorem}\begin{proof} To prove contraction of the best-response,
we need to specify first under which norm the best-response map contracts.
We will use the following norms: the block-maximum norm on $\mathbb{\mathbb{R}}^{Q(N+2)},$
defined as \cite{Bertsekas_Book-Parallel-Comp} 
\begin{equation}
\left\Vert \mathbf{y}\right\Vert _{\text{block}}^{\mathbf{w}}\triangleq\max_{i=1,\ldots,Q}\frac{\left\Vert \mathbf{y}_{i}\right\Vert _{i}}{w_{i}},\quad\mbox{for}\quad\mathbf{y}=(\mathbf{y}_{i})_{i=1}^{Q}\in\mathbb{R}^{Q(N+2)},\label{block_max_weight_norm}
\end{equation}
where $\left\Vert \mathbf{\cdot}\right\Vert _{i}$ is a valid vector
norm on $\mathbb{R}^{N+2}$ and $\mathbf{w}\triangleq[w_{1},\ldots,w_{Q}]^{T}>\mathbf{0}$
is any given positive weight vector. In particular, we choose $\left\Vert \mathbf{\cdot}\right\Vert _{i}$
as follows: partitioning the vector $\mathbf{y}_{i}\in\mathbb{R}^{N+2}$
as $\mathbf{y}_{i}=(y_{i,1},\mathbf{y}_{i,2:N+2})$, with $\mathbf{y}_{i,2:N+2}$
(or $y_{i,1}$) being the $(N+1)$-length vector containing the last
$N+1$ components (or the first component) of $\mathbf{y}_{i}$, and
given the matrix $\mathbf{D}_{i}(t,\, c)$ as defined in (\ref{eq:def_E_and_D_matrices}),
let the vector norm $\left\Vert \mathbf{\cdot}\right\Vert _{i}$ be
$\left\Vert \mathbf{y}\right\Vert _{i}\triangleq\left\Vert \left(|y_{i,1}|,\,\left\Vert \mathbf{y}_{i,2:N+2}\right\Vert _{2}\right)\right\Vert _{\mathbf{D}_{i}(t,\, c)^{2}}$,
where $\left\Vert \mathbf{x}\right\Vert _{\mathbf{D}_{i}(t,\, c)^{2}}\triangleq\left\Vert \mathbf{D}_{i}(t,\, c)\,\mathbf{x}\right\Vert _{2}$.
As it will be clarified shortly, the choice of such a norm is instrumental
to obtain convergence conditions that can be satisfied for all ranges
of $c\geq0$. We also need to introduce the (weighted) maximum norm
on $\mathbb{\mathbb{R}}^{Q},$ defined as \cite{Horn-Johnson_book}
\begin{equation}
\left\Vert \mathbf{x}\right\Vert _{\infty,\text{vec}}^{\mathbf{w}}\triangleq\max_{i=1,\ldots,Q}\frac{\left\vert x_{i}\right\vert }{w_{i}},\quad\mbox{for}\quad\mathbf{x\in\mathbb{R}}^{Q};\label{weighted_infinity_vector_norm}
\end{equation}
 and the\emph{ }matrix norm $\left\Vert \mathbf{\cdot}\right\Vert _{\infty,\text{mat}}^{\mathbf{w}}$
on $\mathbb{R}^{Q\times Q}$ induced by $\left\Vert \cdot\right\Vert _{\infty,\text{vec}}^{\mathbf{w}},$
given by \cite{Horn-Johnson_book} 
\begin{equation}
\left\Vert \mathbf{A}\right\Vert _{\infty,\text{mat}}^{\mathbf{w}}\triangleq\max_{i}\frac{1}{w_{i}}\sum\limits _{j=1}^{Q}\left\vert [\mathbf{A}]_{ij}\right\vert w_{j},\quad\mbox{for}\quad\mathbf{A\in\mathbb{R}}^{Q\times Q}.\label{H_max_weight_norm}
\end{equation}

We are now ready to prove the theorem. \smallskip

\noindent (a): Given $t\geq0$, the P property of matrix $\boldsymbol{{\Gamma}}(t)$
implies $\rho_{q}(t)>0$ for all $q$, and thus $\nabla_{\mathbf{x}_{q}}^{2}\mathcal{L}_{q}(\mathbf{x},\,{\pi},\,\lambda_{q})\succ\mathbf{0}$
for all $(\mathbf{x}_{q},\,\lambda_{q})\in\wh{\mathcal{Y}}_{q}^{t}\times[0,\lambda^{\max}]$,
$\mathbf{x}_{-q}\in{\mathcal{Y}}_{-q}$, and $\pi\geq0$. According
to Proposition \ref{proposition_uniqueness_opt_sol}, this guarantees
the uniqueness of the optimal solution $\overline{{\mathcal{B}}}_{q}({\mathbf{x}}_{-q})=\left(\wh{\tau}_{q}^{\star}(\mathbf{x}_{-q}),\,\mathbf{p}_{q}^{\star}(\mathbf{x}_{-q}),\, P_{q}^{\text{{fa}}\star}(\mathbf{x}_{-q})\right)$
of each nonconvex problem in (\ref{eq:game_G_t_pi}), for every given
$\pi\geq0$ and $\mathbf{x}_{-q}\in{\mathcal{Y}}_{-q}$. \smallskip

\noindent (b): Given the unique solution $\overline{{\mathcal{B}}}_{q}({\mathbf{x}}_{-q})$,
by Lemma \ref{Lemma_ACQ}, it follows that there exists a multiplier
${\overline{{\lambda}}}_{q}$ associated with the nonconvex constraint
$I_{q}(\mathbf{x}_{q})\leq0$ such that the tuple $(\overline{{\mathcal{B}}}_{q}({\mathbf{x}}_{-q}),{\overline{{\lambda}}}_{q})$
satisfies the KKT optimality conditions of the optimization problem
in (\ref{eq:game_G_t_pi}), or equivalently, the VI$(\mathcal{K}_{q},\mathbf{F}_{q})$
defined in (\ref{eq:VI_ref}), which we rewrite here for the reader's
convenience: 
\begin{equation}
\left[\begin{array}{c}
\mathbf{y}_{q}-\overline{{\mathcal{B}}}_{q}({\mathbf{x}}_{-q})\\
{\lambda}_{q}-\overline{\lambda}_{q}
\end{array}\right]^{T}\left(\begin{array}{c}
\nabla_{\mathbf{x}_{q}}\mathcal{L}_{q}\left((\overline{{\mathcal{B}}}_{q}({\mathbf{x}}_{-q}),\,\overline{\lambda}_{q}),\,\mathbf{x}_{-q},{\pi}\right)\\
-I_{q}\left(\overline{{\mathcal{B}}}_{q}({\mathbf{x}}_{-q})\right)
\end{array}\right)\geq0,\quad\forall\left(\mathbf{y}_{q},\lambda_{q}\right)\in\mathcal{Y}_{q}\times\mathbb{R}_{+}^{M},\label{eq:VI_nonconvex_problem}
\end{equation}
with $\nabla_{\mathbf{x}_{q}}\mathcal{L}_{q}$ defined in (\ref{eq:Lagrangian-1}).
Recall that $\overline{\lambda}_{q}\in[0,\lambda^{\max}]$ (Lemma
\ref{Lemma_bounded_multipliers}) and $\overline{{\mathcal{B}}}_{q}({\mathbf{x}}_{-q})\in\wh{\mathcal{Y}}{}_{q}^{t}$
(Lemma \ref{Lemma_lower_bound_on_the_power}).

Consider now two feasible points $\mathbf{x}^{(1)}\triangleq(\mathbf{x}_{q}^{(1)})_{q=1}^{Q},\,\mathbf{x}^{(2)}\triangleq(\mathbf{x}_{q}^{(2)})_{q=1}^{Q}\in{\mathcal{Y}}$,
with $\mathbf{x}_{q}^{(i)}\triangleq\left(\wh{\tau}_{q}^{(i)},\,\mathbf{p}_{q}^{(i)},\, P_{q}^{\text{{fa}}(i)}\right)$
for $i=1,2$, and $q=1,\ldots,Q$, and let $\overline{{\lambda}}_{q}^{(i)}$'s
be the multipliers associated with the nonconvex constraints $\{I_{q}(\mathbf{x}_{q})\leq0\}$'s
at the optimal solutions $\overline{{\mathcal{B}}}_{q}(\mathbf{x}_{-q}^{(i)})=\left(\wh{\tau}_{q}^{\star}(\mathbf{x}_{-q}^{(i)}),\,\mathbf{p}_{q}^{\star}(\mathbf{x}_{-q}^{(i)}),P_{q}^{\text{{fa}}\star}(\mathbf{x}_{-q}^{(i)})\right)$,
for $i=1,2$. Evaluating (\ref{eq:VI_nonconvex_problem}) first in
the solution $(\overline{{\mathcal{B}}}_{q}(\mathbf{x}_{-q}^{(1)}),\overline{{\lambda}}_{q}^{(1)})$
given $(\mathbf{y}_{q},\,{\lambda}_{q})=(\overline{{\mathcal{B}}}_{q}(\mathbf{x}_{-q}^{(2)}),\overline{{\lambda}}_{q}^{(2)})$,
then in the solution $(\overline{{\mathcal{B}}}_{q}(\mathbf{x}_{-q}^{(2)}),\overline{{\lambda}}_{q}^{(2)})$
given $(\mathbf{y}_{q},\,{\lambda}_{q})=(\overline{{\mathcal{B}}}_{q}(\mathbf{x}_{-q}^{(1)}),\overline{{\lambda}}_{q}^{(1)})$,
and summing the resulting inequalities, we obtain
\begin{equation}
\begin{array}{l}
0\geq\left[\begin{array}{c}
\overline{{\mathcal{B}}}_{q}({\mathbf{x}}_{-q}^{(1)})-\overline{{\mathcal{B}}}_{q}({\mathbf{x}}_{-q}^{(2)})\\
\overline{{\lambda}}_{q}^{(1)}-\overline{{\lambda}}_{q}^{(2)}
\end{array}\right]^{T}\left(\begin{array}{c}
\nabla_{\mathbf{x}_{q}}\mathcal{L}_{q}\left((\overline{{\mathcal{B}}}_{q}({\mathbf{x}}_{-q}^{(1)}),\overline{{\lambda}}_{q}^{(1)}),\,\mathbf{x}_{-q}^{(1)},{\pi}\right)-\nabla_{\mathbf{x}_{q}}\mathcal{L}_{q}\left((\overline{{\mathcal{B}}}_{q}({\mathbf{x}}_{-q}^{(2)}),\overline{{\lambda}}_{q}^{(2)}),\,\mathbf{x}_{-q}^{(2)},{\pi}\right)\\
-I_{q}\left(\overline{{\mathcal{B}}}_{q}({\mathbf{x}}_{-q}^{(1)})\right)-\left(-I_{q}\left(\overline{{\mathcal{B}}}_{q}({\mathbf{x}}_{-q}^{(2)})\right)\right)
\end{array}\right).\end{array}\label{eq:sum_two_VI}
\end{equation}
By the main-value theorem{\small{} }we deduce that there exists a $\delta\in(0,1)$
and a pair $({\mathbf{x}}_{q}(\delta),{\mathbf{x}}_{-q}(\delta),\lambda_{q}(\delta))\triangleq\delta\cdot\left(\overline{{\mathcal{B}}}_{q}({\mathbf{x}}_{-q}^{(1)}),{\mathbf{x}}_{-q}^{(1)},\overline{{\lambda}}_{q}^{(1)}\right)+(1-\delta)\cdot\left(\overline{{\mathcal{B}}}_{q}({\mathbf{x}}_{-q}^{(2)}),{\mathbf{x}}_{-q}^{(2)},\overline{{\lambda}}_{q}^{(2)}\right)$
such that \vspace{-0.2cm}

\begin{eqnarray}
0 & \geq & \left(\overline{{\mathcal{B}}}_{q}({\mathbf{x}}_{-q}^{(1)})-\overline{{\mathcal{B}}}_{q}({\mathbf{x}}_{-q}^{(2)})\right)^{T}\left(\nabla_{\mathbf{x}_{q}}^{2}\mathcal{L}_{q}\left(({\mathbf{x}}_{q}(\delta),\,\lambda_{q}(\delta)),\,{\mathbf{x}}_{-q}(\delta),{\pi}\right)\right)\left(\overline{{\mathcal{B}}}_{q}({\mathbf{x}}_{-q}^{(1)})-\overline{{\mathcal{B}}}_{q}({\mathbf{x}}_{-q}^{(2)})\right)\nonumber \\
 &  & +\left(\overline{{\mathcal{B}}}_{q}({\mathbf{x}}_{-q}^{(1)})-\overline{{\mathcal{B}}}_{q}({\mathbf{x}}_{-q}^{(2)})\right)^{T}{\displaystyle {\sum_{r\neq q}}}\nabla_{\mathbf{x}_{q}\mathbf{x}_{r}}^{2}\theta_{q}\left({\mathbf{x}}_{q}(\delta),\,{\mathbf{x}}_{-q}(\delta)\right)\left({\mathbf{x}}_{r}^{(1)}-{\mathbf{x}}_{r}^{(2)}\right).\label{eq:bet-response_ineq_row2}
\end{eqnarray}
Using the definition of $\mathbf{B}_{q}(\mathbf{x},\,\pi_{t},\,\lambda_{q})$,
$\rho_{q}(t)$, $\xi_{q}^{\sup}$, and $\mathbf{D}_{q}(t,\, c)$ and
$\mathbf{E}_{rq}\left(\mathbf{x}\right)$ as given in (\ref{eq:B_q_def}),
(\ref{eq:minimum_eigenvalue_original_matrix}), (\ref{eq:def_E_and_D_matrices}),
and (\ref{eq:upper_spectrum_nabla_f_q_p_q_p_r}), respectively, let
us introduce for each $q=1,\ldots,Q$, the error vectors: 
\begin{equation}
\mathbf{e}_{\overline{{\mathcal{B}}}_{q}}\triangleq\left[\begin{array}{c}
\left|\wh{\tau}_{q}^{\star}(\mathbf{x}_{-q}^{(2)})-\wh{\tau}_{q}^{\star}(\mathbf{x}_{-q}^{(1)})\right|\\
\left\Vert \begin{array}{c}
\mathbf{p}_{q}^{\star}(\mathbf{x}_{-q}^{(2)})-\mathbf{p}_{q}^{\star}(\mathbf{x}_{-q}^{(1)})\\
P_{q}^{\text{{fa}}\star}(\mathbf{x}_{-q}^{(2)})-P_{q}^{\text{{fa}}\star}(\mathbf{x}_{-q}^{(1)})
\end{array}\right\Vert 
\end{array}\right],\quad\mbox{and}\quad\mathbf{e}_{q}\triangleq\left[\begin{array}{c}
\left|\wh{\tau}_{q}^{(2)}-\wh{\tau}_{q}^{(1)}\right|\\
\left\Vert \begin{array}{c}
\mathbf{p}_{q}^{(2)}-\mathbf{p}_{q}^{(1)}\\
P_{q}^{\text{{fa}}(2)}-P_{q}^{\text{{fa}}(1)}
\end{array}\right\Vert 
\end{array}\right]\label{eq:def_error_q}
\end{equation}
and the matrices{\small 
\begin{align}
\hspace{-0.4cm}\mathbf{C}_{q}\left(\mathbf{x}(\delta),\lambda_{q}(\delta),\pi\right) & \triangleq\left[\begin{array}{ll}
\nabla_{\wh{\tau}_{q}}^{2}\mathcal{L}_{q}\left(({\mathbf{x}}_{q}(\delta),\,\lambda_{q}(\delta)),\,{\mathbf{x}}_{-q}(\delta),{\pi}\right), & -\left\Vert \nabla_{\wh{\tau}_{q}\,(\mathbf{p}_{q},P_{q}^{\text{{fa}}})}^{2}\mathcal{L}_{q}\left(({\mathbf{x}}_{q}(\delta),\,\lambda_{q}(\delta)),\,{\mathbf{x}}_{-q}(\delta),{\pi}\right)\right\Vert \\
-\left\Vert \nabla_{(\mathbf{p}_{q},P_{q}^{\text{{fa}}})\,\wh{\tau}_{q}}^{2}\mathcal{L}_{q}\left(({\mathbf{x}}_{q}(\delta),\,\lambda_{q}(\delta)),\,{\mathbf{x}}_{-q}(\delta),{\pi}\right)\right\Vert , & \lambda_{\min}\left(\nabla_{(\mathbf{p}_{q},P_{q}^{\text{{fa}}})}^{2}\mathcal{L}_{q}\left(({\mathbf{x}}_{q}(\delta),\,\lambda_{q}(\delta)),\,{\mathbf{x}}_{-q}(\delta),{\pi}\right)\right)
\end{array}\right]\nonumber \\
 & =\mathbf{B}_{q}\left(\mathbf{x}(\delta),\lambda_{q}(\delta),\pi\right)+\left[\begin{array}{cc}
c\,\left(\dfrac{{1-1/Q}}{\sqrt{f_{q}}}\right)^{2}, & 0\\
0 & 0
\end{array}\right]\succeq\mathbf{D}_{q}(t,\, c)^{2}\label{eq:def_C_q_r2}
\end{align}
} and
\begin{align}
\mathbf{E}_{qr}\left({\mathbf{x}}(\delta)\right) & =\left[\begin{array}{ll}
\left|\nabla_{\wh{\tau}_{q}\wh{\tau}_{r}}^{2}\theta_{q}\left({\mathbf{x}}(\delta)\right)\right|, & 0\\
0, & \left\Vert \nabla_{\mathbf{p}_{q}\mathbf{p}_{r}}^{2}\theta_{q}\left({\mathbf{x}}(\delta)\right)\right\Vert 
\end{array}\right]=\left[\begin{array}{ll}
c\,\left(\dfrac{{1-1/Q}}{\sqrt{f_{q}}}\right)\left(\dfrac{{1/Q}}{\sqrt{f_{r}}}\right), & 0\\
0, & \left\Vert \nabla_{\mathbf{p}_{q}\mathbf{p}_{r}}^{2}\left(-\log r_{q}({\mathbf{p}}(\delta))\right)\right\Vert 
\end{array}\right]\nonumber \\
 & \leq\left[\begin{array}{ll}
c\,\left(\dfrac{{1-1/Q}}{\sqrt{f_{q}}}\right)\left(\dfrac{{1/Q}}{\sqrt{f_{r}}}\right), & 0\\
0, & {\displaystyle {\max_{k=1,\ldots,N}}}\left\{ {\displaystyle {\frac{|\wh{H}_{qr}(k)|^{2}}{\wh{\sigma}_{q,k}^{4}}}}\right\} \cdot\xi_{q}^{\sup}
\end{array}\right]\triangleq\mathbf{E}_{qr}^{\text{{sup}}},\label{eq:def_E_sup}
\end{align}
where the upper bound in (\ref{eq:def_E_sup}) follows from Lemma
\ref{Lemma_lower_bound_on_the_power} and (\ref{eq:upper_spectrum_nabla_f_q_p_q_p_r}).
Then, from inequality (\ref{eq:bet-response_ineq_row2}), we deduce
\begin{equation}
\mathbf{e}_{\overline{{\mathcal{B}}}_{q}}^{T}\,\mathbf{C}_{q}\left(\mathbf{x}(\delta),\lambda_{q}(\delta),\pi\right)\,\mathbf{e}_{\overline{{\mathcal{B}}}_{q}}\leq\mathbf{e}_{\overline{{\mathcal{B}}}_{q}}^{T}\,{\displaystyle {\sum_{r\neq q}}}\,\mathbf{E}_{qr}\left({\mathbf{x}}(\delta)\right)\,\mathbf{e}_{r},\label{eq:bet-response_ineq_row3-1}
\end{equation}
which, using the bounds in (\ref{eq:def_C_q_r2}) and (\ref{eq:def_E_sup})
and the definition of $\beta_{qr}(t,\, c)$ in (\ref{eq:off_diag_Gamma_t}),
leads 
\begin{equation}
\left\Vert \mathbf{D}_{q}(t,\, c)\,\mathbf{e}_{\overline{{\mathcal{B}}}_{q}}\right\Vert _{2}\leq\sum_{r\neq q}\left\Vert \mathbf{D}_{q}(t,\, c)^{-1}\,\mathbf{E}_{qr}\left({\mathbf{x}}(\delta)\right)\,\mathbf{D}_{r}(t,\, c)^{-1}\right\Vert \left\Vert \mathbf{D}_{r}(t,\, c)\,\mathbf{e}_{r}\right\Vert _{2}\leq\sum_{r\neq q}\beta_{qr}(t,\, c)\left\Vert \mathbf{D}_{r}(t,\, c)\,\mathbf{e}_{r}\right\Vert _{2},\label{eq:bet-response_ineq_row3}
\end{equation}
for all $q=1,\ldots,Q$ (the inequality in (\ref{eq:bet-response_ineq_row2})
is trivially satisfied if $\left\Vert \mathbf{D}_{q}(t,\, c)\,\mathbf{e}_{\overline{{\mathcal{B}}}_{q}}\right\Vert _{2}=0$).
Introducing the vectors $\mathbf{e}_{\overline{{\mathcal{B}}},\mathbf{D}}\triangleq\left(\left\Vert \mathbf{e}_{\overline{{\mathcal{B}}}_{q}}\right\Vert _{\mathbf{D}_{q}(t,\, c)}\right)_{q=1}^{Q}$
and $\mathbf{e}_{\mathbf{D}}\triangleq\left(\left\Vert \mathbf{e}_{q}\right\Vert _{\mathbf{D}_{q}(t,\, c)}\right)_{q=1}^{Q}$,
and the matrix $\mathbf{E}(t)\triangleq\mathbf{I}-\boldsymbol{{\Gamma}}(t)$,
the set of inequalities in (\ref{eq:bet-response_ineq_row3}) can
be written in vectorial form as\vspace{-0.1cm} 
\begin{equation}
\mathbf{e}_{\overline{{\mathcal{B}}},\mathbf{D}}\leq\mathbf{E}(t)\,\mathbf{e}_{\mathbf{D}},\qquad\forall\mathbf{x}^{(1)},\mathbf{x}^{(2)}\mathbf{\in\mathcal{Y}},\label{eq:bet-response_ineq_row4}
\end{equation}
and thus, for any given $\mathbf{w}>0$, we have 
\begin{equation}
\left\Vert \overline{{\mathcal{B}}}({\mathbf{x}}^{(1)})-\overline{{\mathcal{B}}}({\mathbf{x}}^{(2)})\right\Vert _{\text{{block}}}^{\mathbf{w}}=\left\Vert \mathbf{e}_{\overline{{\mathcal{B}}},\mathbf{D}}\right\Vert _{\infty,\text{{vec}}}^{\mathbf{w}}\leq\left\Vert \mathbf{E}(t)\right\Vert _{\infty,\text{{mat}}}^{\mathbf{w}}\left\Vert \mathbf{e}_{\mathbf{D}}\right\Vert _{\infty,\text{{vec}}}^{\mathbf{w}}=\left\Vert \mathbf{E}(t)\right\Vert _{\infty,\text{{mat}}}^{\mathbf{w}}\left\Vert {\mathbf{x}}^{(1)}-{\mathbf{x}}^{(2)}\right\Vert _{\text{{block}}}^{\mathbf{w}},\label{eq:bet-response_ineq_row5}
\end{equation}
for all $\mathbf{x}^{(1)},\mathbf{x}^{(2)}\mathbf{\in\mathcal{Y}}$.
To complete the proof we need to show that $\left\Vert \mathbf{E}(t)\right\Vert _{\infty,\text{{mat}}}^{\mathbf{w}}<1$
for some $\mathbf{w}>0$. Invoking Lemma \cite[Lemma 5.2.14]{Cottle-Pang-Stone_bookLCP92}
and \cite[Cor. 6.1]{Bertsekas_Book-Parallel-Comp}, we obtain the
desired result:
\begin{equation}
\boldsymbol{{\Gamma}}(t)\mbox{ is a P-matrix}\qquad\Leftrightarrow\qquad\exists\,\bar{{\mathbf{w}}}>0\quad\mbox{such that }c_{\mathcal{B}}\triangleq\left\Vert \mathbf{E}(t)\right\Vert _{\infty,\text{mat}}^{\bar{{\mathbf{w}}}}<1.\label{eq:contraction_constant}
\end{equation}
\end{proof}

\subsection{Asynchronous convergence theorem\label{sub:Asynchronous-convergence-theorem}}

Convergence of best-response algorithms solving the game $\mathcal{G}_{\pi}(\mathcal{X},\,\boldsymbol{{\theta}})$
follows readily from the block-contraction properties of the best-response,
as proved in Theorem \ref{Theo_Contraction_best-response} and is
thus guaranteed under the same conditions given in Theorem \ref{Theo_Contraction_best-response}. 

\begin{theorem} \label{Theo-async_best-response_NEP} Given the game
$\mathcal{G}_{\pi}(\mathcal{X},\,\boldsymbol{{\theta}})$ with exogenous
(fixed) ${\pi}\geq0$, suppose that $\boldsymbol{{\Gamma}}(t)$ in
(\ref{eq:Gamma_t_matrix_and_Gamma_t_inf}) is a P-matrix. Then, any
sequence generated by the asynchronous algorithm based on the best-response
$\overline{{\mathcal{B}}}$ and starting from any point in $\mathcal{Y}$
converges to a NE of the game, for any given updating feasible schedule
of the players.\end{theorem}

\subsection{On the contraction/convergence conditions}

We derive here easier conditions to be checked implying those in Theorem
\ref{Theo-async_best-response_NEP} (and Theorem \ref{Theo_Contraction_best-response});
this sheds light also on their physical interpretation. The approach
is similar to that followed to prove Corollary \ref{corollary_sf_cond_uniqueness_NE};
we thus provide only a sketch of the proof. 

The main idea is to build a matrix, say $\boldsymbol{{\Gamma}}^{\text{{low}}}(t)$,
such that $\boldsymbol{{\Gamma}}(t)\geq\boldsymbol{{\Gamma}}^{\text{{low}}}(t)$
{[}the inequality has to be intended component-wise{]}, implying that
if $\boldsymbol{{\Gamma}}^{\text{{low}}}(t)$ is a P matrix, then
$\boldsymbol{{\Gamma}}(t)$ is so \cite{Cottle-Pang-Stone_bookLCP92},
which is the condition required by Theorem \ref{Theo-async_best-response_NEP}.
Then, we provide sufficient conditions for $\boldsymbol{{\Gamma}}^{\text{{low}}}(t)$
to be a P matrix. 

To obtain such a $\boldsymbol{{\Gamma}}^{\text{{low}}}(t)$, it is
sufficient to properly upper bound (the modulus of) the off-diagonal
entries $\beta_{qr}(t,c)$ of $\boldsymbol{{\Gamma}}(t)$. Given the
expression of $\beta_{qr}(t,c)$ {[}cf. (\ref{eq:off_diag_Gamma_t}){]},
a way to do that is to find a matrix $\mathbf{B}_{q}^{\text{{low}}}$
such that $\mathbf{B}_{q}\left({\mathbf{x}},\,{\mathbf{\lambda}}_{q},\pi_{t}\right)\geq\mathbf{B}_{q}^{\text{{low}}}$,
and a diagonal matrix $\mathbf{D}_{q}^{\text{{low}}}(t,c)$ such that
$\mathbf{D}_{q}(t,c)\geq\mathbf{D}_{q}^{\text{{low}}}(t,c)$, where
$\mathbf{B}_{q}\left({\mathbf{x}},\,{\mathbf{\lambda}}_{q},\pi_{t}\right)$
and $\mathbf{D}_{q}(t,c)$ are defined in (\ref{eq:B_q_def}) and
(\ref{eq:def_E_and_D_matrices}), respectively. Skipping tedious intermediate
derivations, we obtain the following 
\begin{equation}
\mathbf{B}_{q}^{\text{{low}}}\triangleq\left[\begin{array}{ll}
d_{\wh{\tau}_{q}}^{\min} & -\varsigma_{q}^{\text{{up}}}(t)\\
-\varsigma_{q}^{\text{{up}}}(t) & \lambda_{\text{{least}}}\left(\left[\overline{\nabla_{\mathbf{x}_{q}}^{2}\mathcal{L}_{q}}\right]_{2:N+2}\right)
\end{array}\right],\label{eq:B_q_inf(t)}
\end{equation}
where $d_{\wh{\tau}_{q}}^{\min}$ is defined in (\ref{eq:lamnda_lower_tau_q_and_pfa}),
$\left[\overline{\nabla_{\mathbf{x}_{q}}^{2}\mathcal{L}_{q}}\right]_{2:N+2}$
denotes the $(N+1)$-dimensional lower right block of the matrix $\overline{\nabla_{\mathbf{x}_{q}}^{2}\mathcal{L}_{q}}$
defined in (\ref{eq:def_L_2_inf}), and $\varsigma_{q}^{\text{{up}}}(t)$
is given by 
\begin{equation}
\varsigma_{q}^{\text{{up}}}(t)\triangleq2\,\max\left\{ {t,\,\lambda^{\max}}\right\} \cdot{\displaystyle {\displaystyle {\max_{k=1,\ldots,N}}\left\{ \dfrac{{|{G}_{P,q}(k)|^{2}}}{I^{\,\max}}\right\} }\cdot\left\Vert \begin{array}{l}
\mbox{vect}\left(\boldsymbol{{\omega}}_{\wh{\tau}_{q}}^{\max}\right)\\
\mathbf{1}^{T}\mbox{vect}\left(\boldsymbol{{\omega}}_{\wh{\tau}_{q}P_{q}^{\text{{fa}}}}^{\max}\odot\mathbf{p}_{q}^{\max}\right)
\end{array}\right\Vert ,}\label{eq:zita_up}
\end{equation}
with $\boldsymbol{{\omega}}_{\wh{\tau}_{q}}^{\max}$ and $\boldsymbol{{\omega}}_{\wh{\tau}_{q}P_{q}^{\text{{fa}}}}^{\max}$
defined in (\ref{eq:der_Pmiss_wrt_tau}) and (\ref{eq:sec_der_Pmiss_wrt_tau_pfa}),
respectively. Note that, since the following bounds hold between the
entries of $\mathbf{B}_{q}\left({\mathbf{x}},\,{\mathbf{\lambda}}_{q},\pi_{t}\right)$
and $\mathbf{B}_{q}^{\text{{low}}}$: 
\begin{align*}
\left.\nabla_{\wh{\tau}_{q}^{2}}^{2}\mathcal{L}_{q}\left(\mathbf{x},\,\wh{\boldsymbol{{\pi}}}_{t},\,\boldsymbol{{\lambda}}_{q}\right)\right|_{c=0} & \geq d_{\wh{\tau}_{q}}^{\min},\\
\left\Vert \nabla_{\wh{\tau}_{q}\,(\mathbf{x}_{q},P_{q}^{\text{{fa}}})}^{2}\mathcal{L}_{q}\left(\mathbf{x},\boldsymbol{{\lambda}}_{q},\pi_{t}\right)\right\Vert  & \leq\varsigma_{q}^{\text{{up}}}(t),\\
\lambda_{\text{{least}}}\left(\nabla_{(\mathbf{x}_{q},P_{q}^{\text{{fa}}})}^{2}\mathcal{L}_{q}\left(\mathbf{x},\,\wh{\boldsymbol{{\pi}}}_{t},\,\boldsymbol{{\lambda}}_{q}\right)\right) & \geq\lambda_{\text{{least}}}\left(\left[\overline{\nabla_{\mathbf{x}_{q}}^{2}\mathcal{L}_{q}}\right]_{2:N+2}\right),
\end{align*}
matrix $\mathbf{B}_{q}^{\text{{low}}}$ satisfies the desired property
$\mathbf{B}_{q}\left({\mathbf{x}},\,{\mathbf{\lambda}}_{q},\pi_{t}\right)\geq\mathbf{B}_{q}^{\text{{low}}}$. 

Finally, using $\mathbf{B}_{q}^{\text{{low}}}$, we can introduce
a lower bound of the quantities $\rho_{q}(t)$ in (\ref{eq:minimum_eigenvalue_original_matrix})
\begin{equation}
\rho_{q}^{\text{{low}}}(t)\triangleq\left\{ \begin{array}{ll}
\overline{{\rho}}_{q}^{\text{{low}}}(t)\triangleq\lambda_{\text{{least}}}\left(\mathbf{B}_{q}^{\text{{low}}}\right),\quad & \mbox{if }\overline{{\rho}}_{q}^{\text{{low}}}(t)\geq0,\\
0, & \mbox{otherwise},
\end{array}\right.\label{eq:delta_inf_t}
\end{equation}
and define the matrix $\mathbf{D}_{q}^{\text{{low}}}(t,c)$ as 
\begin{equation}
\mathbf{D}_{q}^{\text{{low}}}(t,c)^{2}\triangleq\left[\begin{array}{cc}
\rho_{q}^{\text{{low}}}(t)+c\,\left(\dfrac{{1-1/Q}}{\sqrt{f_{q}}}\right)^{2}, & 0\\
0 & \rho_{q}^{\text{{low}}}(t)
\end{array}\right],\label{eq:D_low}
\end{equation}
which satisfies $\mathbf{D}_{q}(t,c)\geq\mathbf{D}_{q}^{\text{{low}}}(t,c)$.
Using the above matrices, the desired upper bound $\beta_{qr}^{\text{{up}}}(t,\, c)$
of the coefficients $\beta_{qr}(t,\, c)$ is 
\begin{align}
\beta_{qr}^{\text{{up}}}(t,\, c) & \triangleq\left\Vert \mathbf{D}_{q}^{\text{{low}}}(t,\, c)^{-1}\,\mathbf{E}_{qr}^{\text{{sup}}}\,\mathbf{D}_{r}^{\text{{low}}}(t,\, c)^{-1}\right\Vert \nonumber \\
 & =\max\left\{ \dfrac{c\cdot{1/(Q-1)}}{\sqrt{{\dfrac{\rho_{q}^{\text{{low}}}(t)\, f_{q}}{(1-1/Q)^{2}}}+c}\,\sqrt{{\dfrac{\rho_{r}^{\text{{low}}}(t)\, f_{r}}{(1-1/Q)^{2}}}+c}},\,{\displaystyle {\max_{k=1,\ldots,N}}}\left\{ {\displaystyle {\frac{|\wh{H}_{qr}(k)|^{2}}{\wh{\sigma}_{q,k}^{4}}}}\right\} \,\dfrac{\xi_{q}^{\sup}}{\sqrt{\rho_{q}^{\text{{low}}}(t)}\sqrt{\rho_{r}^{\text{{low}}}(t)}}\right\} \geq\beta_{qr}(t,\, c),\label{eq:beta_qr_up}
\end{align}
with $\xi_{q}^{\sup}$ and $\mathbf{E}_{qr}^{\sup}$ defined in (\ref{eq:upper_spectrum_nabla_f_q_p_q_p_r})
and (\ref{eq:def_E_sup}), respectively. Using these quantities it
is not difficult to see that the matrix $\boldsymbol{{\Gamma}}^{\text{{low}}}(t)$
defined as\textcolor{black}{{} 
\begin{equation}
\left[\boldsymbol{{\Gamma}}^{\text{{low}}}(t)\right]_{q,r}\triangleq\left\{ \begin{array}{ll}
1, & \mbox{if }r=q,\\
-\beta_{qr}^{\text{{up}}}(t,\, c),\quad & \mbox{otherwise, }
\end{array}\right.\label{eq:Gamma_t_inf}
\end{equation}
satisfies the desired property }$\boldsymbol{{\Gamma}}(t)\geq\boldsymbol{{\Gamma}}^{\text{{low}}}(t)$
for any $t\geq0$.

Since $\boldsymbol{{\Gamma}}^{\text{{low}}}(t)$ is a P matrix if
and only if $\rho\left(\mathbf{I}-\boldsymbol{{\Gamma}}^{\text{{low}}}(t)\right)<1$
\cite[Lemma 5.2.14]{Cottle-Pang-Stone_bookLCP92}, imposing that $\mathbf{I}-\boldsymbol{{\Gamma}}^{\text{{low}}}(t)$
is row or column diagonal dominat, leads to the desired sufficient
conditions guaranteeing convergence of asynchronous algorithms based
on the best-response $\overline{{\mathcal{B}}}$. This is made formal
in the corollary below.

\begin{corollary}\textcolor{red}{\small \label{Corollary_SF_Cond_convergence_algo_zero_pricing}}\textcolor{red}{{}
}Statements in Theorem \ref{Theo-async_best-response_NEP} (or Theorem
\ref{Theo_Contraction_best-response}) hold true \textcolor{black}{if
one of the two following conditions is satisfied:}

\noindent\textcolor{black}{-}\textcolor{black}{\emph{ Low received
MUI}}\textcolor{black}{: for all $q=1,\ldots,Q$,
\begin{equation}
\begin{array}{l}
{\displaystyle {\sum_{r\neq q}}\,}\beta_{qr}^{\text{{up}}}(t,\, c)<1,\end{array}\label{eq:sf_cond_uniq_NE_no_outer_prices_row_dominance}
\end{equation}
}

\noindent\textcolor{black}{-}\textcolor{black}{\emph{ Low transmitted
MUI}}\textcolor{black}{: for all $r=1,\ldots,Q$,}
\begin{equation}
{\displaystyle {\sum_{q\neq r}}\,}\beta_{qr}^{\text{{up}}}(t,\, c)<1.\label{eq:sf_cond_uniq_NE_no_outer_prices_col_dominance}
\end{equation}
\end{corollary}

The physical interpretation of the above conditions is similar to
that  given for the existence/uniqueness of the NE (cf. Section \ref{remark_cond}).
Roughly speaking, conditions (\ref{eq:sf_cond_uniq_NE_no_outer_prices_row_dominance})
or (\ref{eq:sf_cond_uniq_NE_no_outer_prices_col_dominance}) require
``low'' interference in the network, meaning ``small'' values
of the (normalized) cross-channels $|\wh{H}_{qr}(k)|^{2}/\wh{\sigma}_{q,k}^{4}$
as well as large values of coefficients $\rho_{q}^{\text{{low}}}(t)$,
which is met if, among all, the (normalized) cross-channels ${|G_{P,q}(k)|^{2}}/I^{\,{\rm \max}}$
are ``sufficiently small''. An illustrative example is obtained
in the two opposite cases where there is no optimization of the sensing
times (and thus $c=0$) or the sensing times are optimized by imposing
a common optimal sensing time by choosing a (sufficiently) large constant
$c$ (and there are many active SUs). For those two cases, conditions
(\ref{eq:sf_cond_uniq_NE_no_outer_prices_row_dominance}) and (\ref{eq:sf_cond_uniq_NE_no_outer_prices_col_dominance})
reduce respectively to $ $
\begin{equation}
\begin{array}{l}
{\displaystyle {\sum_{r\neq q}}\,}{\displaystyle {\max_{k=1,\ldots,N}}}\left\{ {\displaystyle {\frac{|\wh{H}_{qr}(k)|^{2}}{\wh{\sigma}_{q,k}^{4}}}}\right\} \,\gamma_{qr}<1,\quad\mbox{and}\quad{\displaystyle {\sum_{q\neq r}}\,}{\displaystyle {\max_{k=1,\ldots,N}}}\left\{ {\displaystyle {\frac{|\wh{H}_{qr}(k)|^{2}}{\wh{\sigma}_{q,k}^{4}}}}\right\} \,\gamma_{qr}<1,\end{array}\label{eq:simpler_SF}
\end{equation}
with 
\[
\gamma_{qr}\triangleq\dfrac{\xi_{q}^{\sup}}{\sqrt{\rho_{q}^{\text{{low}}}(t)}\sqrt{\rho_{r}^{\text{{low}}}(t)}}.
\]
Note that $\gamma_{qr}$'s, among all, depend on the cross-channels
${|G_{P,q}(k)|^{2}}/I^{\,{\rm \max}}$, and become ``small'' when
${|G_{P,q}(k)|^{2}}/I^{\,{\rm \max}}$ are small. Conditions (\ref{eq:simpler_SF})
are thus satisfies if there is not ``too much'' interference in
the system.

\section{Convergence of Best-Response Algorithms for $\mathcal{G}(\mathcal{X},\boldsymbol{{\theta}})$\label{app:Convergence-of-Best-Response}}

\subsection{Proof of Theorem \ref{ProxDecAlg_viaGT_conv_theo}\label{Proof_ProxDecAlg_viaGT_conv_theo}}

First of all note that, given $\boldsymbol{{\lambda}}^{0}\geq\mathbf{0}$
and $\pi_{t}^{0}\geq0$ and under the setting of Lemma \ref{Lemma_G_VI_equivalence},
the game $\mathcal{G}_{t}(\mathcal{X},\boldsymbol{{\theta}},\boldsymbol{{\lambda}}^{0},\pi_{t}^{0})$
has a unique NE, denoted by $\left(\mathbf{x}^{\star}(\boldsymbol{{\lambda}}^{0},\pi_{t}^{0}),\boldsymbol{{\lambda}}^{\star}(\boldsymbol{{\lambda}}^{0},\pi_{t}^{0}),\pi_{t}^{\star}(\boldsymbol{{\lambda}}^{0},\pi_{t}^{0})\right)$,
where we made explicit the dependence on the regularization tuple
$(\boldsymbol{{\lambda}}^{0},\pi_{t}^{0})$. This makes the sequence
$\left\{ \left(\mathbf{x}^{\star}(\boldsymbol{{\lambda}}^{n},\pi_{t}^{n}),\boldsymbol{{\lambda}}^{\star}(\boldsymbol{{\lambda}}^{n},\pi_{t}^{n}),\pi_{t}^{\star}(\boldsymbol{{\lambda}}^{n},\pi_{t}^{n})\right)\right\} _{n=0}^{\infty}$
generated by Algorithm \ref{alg:PDA_GT_interpretation} well defined.
The uniqueness of the NE of $\mathcal{G}_{t}(\mathcal{X},\boldsymbol{{\theta}},\boldsymbol{{\lambda}}^{0},\pi_{t}^{0})$
can be proved by exploring the connection between the game and a suitably
defined VI, as briefly outlined next. Under the positive definiteness
of matrix $\mathbf{A}\left(\mathbf{x},\boldsymbol{{\lambda},}\pi_{t}\right)$
(as required by Lemma \ref{Lemma_G_VI_equivalence}), $\mathcal{G}_{t}(\mathcal{X},\boldsymbol{{\theta}},\boldsymbol{{\lambda}}^{0},\pi_{t}^{0})$
is equivalent to the $\text{{VI}}(\mathcal{Z}_{t},\,\boldsymbol{{\Psi}}_{\boldsymbol{{\lambda}}^{0},\pi_{t}^{0}})$,
with $\mathcal{Z}_{t}$ given in (\ref{eq:KKT_Game_Gt}) and the VI
function $\boldsymbol{{\Psi}}_{\boldsymbol{{\lambda}}^{0},\pi_{t}^{0}}(\mathbf{x},{\boldsymbol{{\lambda}}},\pi_{t})$
defined as 
\begin{equation}
\boldsymbol{{\Psi}}_{\boldsymbol{{\lambda}}^{0},\pi_{t}^{0}}(\mathbf{x},{\boldsymbol{{\lambda}}},\pi_{t})=\boldsymbol{{\Psi}}(\mathbf{x},{\boldsymbol{{\lambda}}},\pi_{t})+\epsilon\cdot\left[\begin{array}{c}
\mathbf{0}_{Q(N+2)\times1}\\
\left(\left(\begin{array}{c}
\boldsymbol{{\lambda}}\\
\pi_{t}
\end{array}\right)-\left(\begin{array}{c}
\boldsymbol{{\lambda}}^{0}\\
\pi_{t}^{0}
\end{array}\right)\right)
\end{array}\right].\label{eq:Phi_epsilon_u}
\end{equation}
In other words, the $\text{{VI}}(\mathcal{Z}_{t},\,\boldsymbol{{\Psi}}_{\boldsymbol{{\lambda}}^{0},\pi_{t}^{0}})$
is obtained by the $\text{{VI}}(\mathcal{Z}_{t},\,\boldsymbol{{\Psi}})$
in (\ref{eq:KKT_Game_Gt}) introducing the proximal regularization
of some of the VI variables, namely the $\lambda$-variables and $\pi_{t}$-variable.
The Jacobian matrix of $\boldsymbol{{\Psi}}_{\boldsymbol{{\lambda}}^{0},\pi_{t}^{0}}(\mathbf{x},{\boldsymbol{{\lambda}}},\pi_{t})$
denoted by $\mbox{J}\boldsymbol{{\Psi}}_{\boldsymbol{{\lambda}}^{0},\pi_{t}^{0}}(\mathbf{x},{\boldsymbol{{\lambda}}},\pi_{t})$
is 
\begin{equation}
\mbox{J}\boldsymbol{{\Psi}}_{\boldsymbol{{\lambda}}^{0},\pi_{t}^{0}}(\mathbf{x},{\boldsymbol{{\lambda}}},\pi_{t})\triangleq\left[\begin{array}{c}
\begin{array}{ccc}
\mathbf{A}\left(\mathbf{x},\boldsymbol{{\lambda},}\pi_{t}\right)\, & \nabla_{\mathbf{x}}\mathbf{I}(\mathbf{x})\, & \nabla_{\mathbf{x}}I(\mathbf{x})\\
-\nabla_{\mathbf{x}}\mathbf{I}(\mathbf{x})^{T} & \epsilon\cdot\mathbf{I} & \mathbf{0}\\
-\nabla_{\mathbf{x}}I(\mathbf{x})^{T} & \mathbf{0} & \epsilon
\end{array}\end{array}\right],\label{eq:Jacobian_Psi_u_epsilon}
\end{equation}
where $\nabla_{\mathbf{x}}\mathbf{I}(\mathbf{x})\triangleq\nabla_{\mathbf{x}}[I_{1}(\mathbf{x}_{1}),\cdots,I_{Q}(\mathbf{x}_{Q})]$.
If $\mathbf{A}\left(\mathbf{x},\boldsymbol{{\lambda},}\pi_{t}\right)$
is uniformly positive definite, then so is $\mbox{J}\boldsymbol{{\Psi}}_{\boldsymbol{{\lambda}}^{0},\pi_{t}^{0}}(\mathbf{x},$
${\boldsymbol{{\lambda}}},\pi_{t})$. It turns out that, under the
setting of Lemma \ref{Lemma_G_VI_equivalence}, the regularized VI$(\mathcal{Z}_{t},\,\boldsymbol{{\Psi}}_{\boldsymbol{{\lambda}}^{0},\pi_{t}^{0}})$
is strongly monotone \cite[Prop. 2.3.2(c)]{Facchinei-Pang_FVI03}
and thus has a unique solution \cite[Th. 2.3.3]{Facchinei-Pang_FVI03},
implying the uniqueness of the NE $\left(\mathbf{x}^{\star}(\boldsymbol{{\lambda}}^{0},\pi_{t}^{0}),\boldsymbol{{\lambda}}^{\star}(\boldsymbol{{\lambda}}^{0},\pi_{t}^{0}),\pi_{t}^{\star}(\boldsymbol{{\lambda}}^{0},\pi_{t}^{0})\right)$
of $\mathcal{G}_{t}(\mathcal{X},\boldsymbol{{\theta}},\boldsymbol{{\lambda}}^{0},\pi_{t}^{0})$.

Once we have proved that $\left(\mathbf{x}^{\star}(\boldsymbol{{\lambda}},\pi_{t}),\boldsymbol{{\lambda}}^{\star}(\boldsymbol{{\lambda}},\pi_{t}),\pi_{t}^{\star}(\boldsymbol{{\lambda}},\pi_{t})\right)$
is well defined for any given $\boldsymbol{{\lambda}}\geq\mathbf{0}$
and $\pi_{t}\geq0$, we can derive the main properties of such a tuple
{[}interpreting its components as functions of $(\boldsymbol{{\lambda}},\pi_{t})${]},
along with its connection with the NE of the game $\mathcal{G}_{t}(\mathcal{X},\,\boldsymbol{{\theta}})$
{[}and thus $\mathcal{G}(\mathcal{X},\,\boldsymbol{{\theta}})${]};
these properties will be instrumental to prove Theorem \ref{ProxDecAlg_viaGT_conv_theo}.

\begin{proposition} \label{Prop_partial_regularization}Given $t>0$,
suppose that $\mathbf{A}\left(\mathbf{x},\boldsymbol{{\lambda}},\pi_{t}\right)$
in (\ref{eq:def_A_matrix}) is uniformly positive definite for all
$(\mathbf{x},\boldsymbol{{\lambda}},\pi_{t})\in{\mathcal{Y}}\times[0,\lambda^{\max}]^{Q}\times\mathcal{S}_{t}$,
and let $\epsilon>0$ be given. Then the following hold: \vspace{-0.2cm}
\begin{description}
\item [{(a)}] The mapping associated with the $\lambda$-components and
$\pi$-component of $\left(\mathbf{x}^{\star}(\boldsymbol{{\lambda}},\pi_{t}),\boldsymbol{{\lambda}}^{\star}(\boldsymbol{{\lambda}},\pi_{t}),\pi_{t}^{\star}(\boldsymbol{{\lambda}},\pi_{t})\right)$,
i.e., 
\begin{equation}
\left(\begin{array}{c}
\boldsymbol{{\lambda}}^{\star}(\cdot)\\
\pi_{t}^{\star}(\cdot)
\end{array}\right):[0,\lambda^{\max}]^{Q}\times\mathcal{S}_{t}\ni\left(\boldsymbol{{\lambda}},\pi_{t}\right)\mapsto\left(\begin{array}{c}
\boldsymbol{{\lambda}}^{\star}(\boldsymbol{{\lambda},}\pi_{t})\\
\pi_{t}^{\star}(\boldsymbol{{\lambda},}\pi_{t})
\end{array}\right)\label{eq:lambda_pi_mapping}
\end{equation}
 has a fixed point, and it is nonexpansive on $[0,\lambda^{\max}]^{Q}\times\mathcal{S}_{t}$;
\vspace{-0.2cm}
\item [{(b)}] The mapping associated with the $x$-components of $\left(\mathbf{x}^{\star}(\boldsymbol{{\lambda}},\pi_{t}),\boldsymbol{{\lambda}}^{\star}(\boldsymbol{{\lambda}},\pi_{t}),\pi_{t}^{\star}(\boldsymbol{{\lambda}},\pi_{t})\right)$,
i.e., 
\begin{equation}
\mathbf{x}^{\star}(\cdot):[0,\lambda^{\max}]^{Q}\times\mathcal{S}_{t}\ni\left(\boldsymbol{{\lambda}},\pi_{t}\right)\mapsto\mathbf{x}^{\star}(\boldsymbol{{\lambda},}\pi_{t})\label{eq:x_star_mapping}
\end{equation}
 is Lipschitz continuous on $[0,\lambda^{\max}]^{Q}\times\mathcal{S}_{t}$,
i.e., there exists a constant $0<\nu<+\infty$ such that 
\begin{equation}
\|\mathbf{x}^{\star}(\boldsymbol{{\lambda}}^{(1)},\pi_{t}^{(1)})-\mathbf{x}^{\star}(\boldsymbol{{\lambda}}^{(2)},\pi_{t}^{(2)})\|_{2}\,\leq\,\nu\,\|(\boldsymbol{{\lambda}}^{(1)},\pi_{t}^{(1)})-(\boldsymbol{{\lambda}}^{(2)},\pi_{t}^{(2)})\|_{2},\label{eq:Lip_cont}
\end{equation}
for all $(\boldsymbol{{\lambda}}^{(1)},\pi_{t}^{(1)}),\,(\boldsymbol{{\lambda}}^{(2)},\pi_{t}^{(2)})\in[0,\lambda^{\max}]^{Q}\times\mathcal{S}_{t}$;\vspace{-0.2cm}
\item [{(c)}] For any fixed-point $(\overline{\boldsymbol{{\lambda}}},\overline{\pi}_{t})\in[0,\lambda^{\max}]^{Q}\times\mathcal{S}_{t}$
of $\left(\boldsymbol{{\lambda}}^{\star}(\cdot),\pi_{t}^{\star}(\cdot)\right)$,
the tuple $(\mathbf{x}^{\star}(\overline{\boldsymbol{{\lambda}}},\overline{\pi}_{t}),\,\overline{\boldsymbol{{\lambda}}},\overline{\pi}_{t})$
is a solution of the $\text{{VI}}(\mathcal{Z}_{t},\,\boldsymbol{{\Psi}})$;
therefore, it is a NE of $\mathcal{G}_{t}(\mathcal{X},\,\boldsymbol{{\theta}})$.
\end{description}
\end{proposition}

\begin{proof} We prove next only (a) and (b); (c) follows similarly. 

\noindent(a) Let $\left(\overline{\mathbf{x}},\overline{\boldsymbol{{\lambda}}},\overline{\pi}_{t}\right)\in\mathcal{Z}_{t}$
be a solution of the $\text{{VI}}(\mathcal{Z}_{t},\,\boldsymbol{{\Psi}})$
in (\ref{eq:KKT_Game_Gt}), whose existence is guaranteed by Lemma
\ref{Lemma_G_VI_equivalence}; recall that, by Lemma \ref{Lemma_bounded_multipliers},
it must be $\left(\overline{\mathbf{x}},\overline{\boldsymbol{{\lambda}}},\overline{\pi}_{t}\right)\in\mathcal{Y}\times[0,\lambda^{\max}]^{Q}\times\mathcal{S}_{t}$.
It follows that: i) $(\mathbf{x}^{\star}(\overline{\boldsymbol{{\lambda}}},\overline{\pi}_{t}),\,\boldsymbol{{\lambda}}^{\star}(\overline{\boldsymbol{{\lambda}}},\overline{\pi}_{t}),$
$\pi_{t}^{\star}(\overline{\boldsymbol{{\lambda}}},\overline{\pi}_{t}))$
is the \emph{unique} solution of the $\text{{VI}}(\mathcal{Z}_{t},\,\boldsymbol{{\Psi}}_{\overline{\boldsymbol{{\lambda}}},\overline{\pi}_{t}})$;
and ii) $\left(\overline{\mathbf{x}},\overline{\boldsymbol{{\lambda}}},\overline{\pi}_{t}\right)$
is also a solution of $\text{{VI}}(\mathcal{Z}_{t},\,\boldsymbol{{\Psi}}_{\overline{\boldsymbol{{\lambda}}},\overline{\pi}_{t}})$.
Hence, it must be $\mathbf{x}^{\star}(\overline{\boldsymbol{{\lambda}}},\overline{\pi}_{t})=\overline{\mathbf{x}}$,
$\boldsymbol{{\lambda}}^{\star}(\overline{\boldsymbol{{\lambda}}},\overline{\pi}_{t})=\overline{\boldsymbol{{\lambda}}}$,
and $\pi_{t}^{\star}(\overline{\boldsymbol{{\lambda}}},\overline{\pi}_{t})=$$\overline{\pi}_{t}$,
which implies the existence of a fixed-point of the mapping $\left(\boldsymbol{{\lambda}}^{\star}(\cdot),\pi_{t}^{\star}(\cdot)\right)$
in (\ref{eq:lambda_pi_mapping}); moreover, since $\left(\overline{\mathbf{x}},\overline{\boldsymbol{{\lambda}}},\overline{\pi}_{t}\right)\in\mathcal{Y}\times[0,\lambda^{\max}]^{Q}\times\mathcal{S}_{t}$,
such a fixed point is in $[0,\lambda^{\max}]^{Q}\times\mathcal{S}_{t}$. 

We prove now that $\left(\boldsymbol{{\lambda}}^{\star}(\cdot),\pi_{t}^{\star}(\cdot)\right)$
is nonexpansive on $[0,\lambda^{\max}]^{Q}\times\mathcal{S}_{t}$.
Given two distinct tuples $(\boldsymbol{{\lambda}}^{(1)},\pi_{t}^{(1)}),\,(\boldsymbol{{\lambda}}^{(2)},\pi_{t}^{(2)})\in[0,\lambda^{\max}]^{Q}\times\mathcal{S}_{t}$,
by definition, the tuples $(\mathbf{x}^{\star}(\boldsymbol{{\lambda}}^{(i)},\pi_{t}^{(i)}),\boldsymbol{{\lambda}}^{\star}(\boldsymbol{{\lambda}}^{(i)},\pi_{t}^{(i)}),\pi_{t}^{\star}(\boldsymbol{{\lambda}}^{(i)},\pi_{t}^{(i)}))$,
with $i=1,2$, satisfy the following: 
\begin{equation}
\left(\begin{array}{c}
\mathbf{x}-\mathbf{x}^{\star}(\boldsymbol{{\lambda}}^{(i)},\pi_{t}^{(i)})\\
\boldsymbol{{\lambda}}-\boldsymbol{{\lambda}}^{\star}(\boldsymbol{{\lambda}}^{(i)},\pi_{t}^{(i)})\\
\pi_{t}-\pi_{t}^{\star}(\boldsymbol{{\lambda}}^{(i)},\pi_{t}^{(i)})
\end{array}\right)^{T}\left[\begin{array}{c}
\left(\nabla_{\mathbf{x}_{q}}\mathcal{L}_{q}\left(\mathbf{x}^{\star}(\boldsymbol{{\lambda}}^{(i)},\pi_{t}^{(i)}),\,{\boldsymbol{{\lambda}}}_{q}^{\star}(\boldsymbol{{\lambda}}^{(i)},\pi_{t}^{(i)}),{\pi}_{t}^{\star}(\boldsymbol{{\lambda}}^{(i)},\pi_{t}^{(i)}),\right)\right)_{q=1}^{Q}\\
-\left(\begin{array}{c}
\left(I_{q}\left(\mathbf{x}_{q}^{\star}(\boldsymbol{{\lambda}}^{(i)},\pi_{t}^{(i)})\right)\right)_{q=1}^{Q}\\
I\left(\mathbf{x}^{\star}(\boldsymbol{{\lambda}}^{(i)},\pi_{t}^{(i)})\right)
\end{array}\right)+\epsilon\cdot\left(\begin{array}{c}
\boldsymbol{{\lambda}}^{\star}(\boldsymbol{{\lambda}}^{(i)},\pi_{t}^{(i)})-\boldsymbol{{\lambda}}^{(i)}\\
\pi_{t}^{\star}(\boldsymbol{{\lambda}}^{(i)},\pi_{t}^{(i)})-\pi_{t}^{(i)}
\end{array}\right)
\end{array}\right]\geq0,\label{eq:VI_two_points_prices}
\end{equation}
for all $(\mathbf{x},\,\boldsymbol{{\lambda}},\pi_{t})\in{\mathcal{Y}}^{t}\times[0,\lambda^{\max}]^{Q}\times\mathcal{S}_{t}$
and $i=1,2$. Thus, similar to the proof of Theorem \ref{Theo_Contraction_best-response},
we deduce

\begin{equation}
\begin{array}{l}
\!\!\!\!\!\!\!\!\!\!\left(\begin{array}{c}
\mathbf{x}^{\star}(\boldsymbol{{\lambda}}^{(2)},\pi_{t}^{(2)})-\mathbf{x}^{\star}(\boldsymbol{{\lambda}}^{(1)},\pi_{t}^{(1)})\\
\boldsymbol{{\lambda}}^{\star}(\boldsymbol{{\lambda}}^{(2)},\pi_{t}^{(2)})-\boldsymbol{{\lambda}}^{\star}(\boldsymbol{{\lambda}}^{(1)},\pi_{t}^{(1)})\\
\pi_{t}^{\star}(\boldsymbol{{\lambda}}^{(2)},\pi_{t}^{(2)})-\pi_{t}^{\star}(\boldsymbol{{\lambda}}^{(1)},\pi_{t}^{(1)})
\end{array}\right)^{T}\times\medskip\\
\qquad\qquad\qquad\qquad\qquad\left[\left(\begin{array}{c}
\left(\nabla_{\mathbf{x}_{q}}\mathcal{L}_{q}\left(\mathbf{x}^{\star}(\boldsymbol{{\lambda}}^{(1)},\pi_{t}^{(1)}),\,\boldsymbol{{\lambda}}_{q}^{\star}(\boldsymbol{{\lambda}}^{(1)},\pi_{t}^{(1)}),\,\pi_{t}^{\star}(\boldsymbol{{\lambda}}^{(1)},\pi_{t}^{(1)})\right)\right)_{q=1}^{Q}\\
\left(\begin{array}{c}
\left(I_{q}\left(\mathbf{x}_{q}^{\star}(\boldsymbol{{\lambda}}^{(2)},\pi_{t}^{(2)})\right)\right)_{q=1}^{Q}\\
I\left(\mathbf{x}^{\star}(\boldsymbol{{\lambda}}^{(2)},\pi_{t}^{(2)})\right)
\end{array}\right)+\epsilon\left(\begin{array}{c}
\boldsymbol{{\lambda}}^{\star}(\boldsymbol{{\lambda}}^{(1)},\pi_{t}^{(1)})-\boldsymbol{{\lambda}}^{(1)}\\
\pi_{t}^{\star}(\boldsymbol{{\lambda}}^{(1)},\pi_{t}^{(1)})-\pi_{t}^{(1)}
\end{array}\right)
\end{array}\right)+\right.\\
\left.\qquad\qquad\qquad\qquad\qquad-\left(\begin{array}{c}
\left(\nabla_{\mathbf{x}_{q}}\mathcal{L}_{q}\left(\mathbf{x}^{\star}(\boldsymbol{{\lambda}}^{(2)},\pi_{t}^{(2)}),\,\boldsymbol{{\lambda}}_{q}^{\star}(\boldsymbol{{\lambda}}^{(2)},\pi_{t}^{(2)}),\,\pi_{t}^{\star}(\boldsymbol{{\lambda}}^{(2)},\pi_{t}^{(2)})\right)\right)_{q=1}^{Q}\\
\left(\begin{array}{c}
\left(I_{q}\left(\mathbf{x}_{q}^{\star}(\boldsymbol{{\lambda}}^{(1)},\pi_{t}^{(1)})\right)\right)_{q=1}^{Q}\\
I\left(\mathbf{x}^{\star}(\boldsymbol{{\lambda}}^{(1)},\pi_{t}^{(1)})\right)
\end{array}\right)-\epsilon\left(\begin{array}{c}
\boldsymbol{{\lambda}}^{\star}(\boldsymbol{{\lambda}}^{(2)},\pi_{t}^{(2)})-\boldsymbol{{\lambda}}^{(2)}\\
\pi_{t}^{\star}(\boldsymbol{{\lambda}}^{(2)},\pi_{t}^{(2)})-\pi_{t}^{(2)}
\end{array}\right)
\end{array}\right)\right]\geq0.
\end{array}\label{eq:VI_difference_2}
\end{equation}
By the mean-value theorem, it follows that there exists a tuple $\left(\mathbf{x}_{\delta},\,{\boldsymbol{{\lambda}}}_{\delta},\,{\pi}_{\delta}\right)$
lying on the line segment joining $(\mathbf{x}^{\star}(\boldsymbol{{\lambda}}^{(1)},\pi_{t}^{(1)}),\boldsymbol{{\lambda}}^{\star}(\boldsymbol{{\lambda}}^{(1)},\pi_{t}^{(1)}),\pi_{t}^{\star}(\boldsymbol{{\lambda}}^{(1)},\pi_{t}^{(1)}))$
and $(\mathbf{x}^{\star}(\boldsymbol{{\lambda}}^{(2)},\pi_{t}^{(2)}),\boldsymbol{{\lambda}}^{\star}(\boldsymbol{{\lambda}}^{(2)},\pi_{t}^{(2)}),\pi_{t}^{\star}(\boldsymbol{{\lambda}}^{(2)},\pi_{t}^{(2)}))$
such that {[}see also (\ref{eq:Jacobian_Psi_u_epsilon}){]}
\begin{align}
\left(\mathbf{x}^{\star}(\boldsymbol{{\lambda}}^{(2)},\pi_{t}^{(2)})-\mathbf{x}^{\star}(\boldsymbol{{\lambda}}^{(1)},\pi_{t}^{(1)})\right)^{T}\mathbf{A}\left(\mathbf{x}_{\delta},\,{\boldsymbol{{\lambda}}}_{\delta},\,{\pi}_{\delta}\right)\left(\mathbf{x}^{\star}(\boldsymbol{{\lambda}}^{(2)},\pi_{t}^{(2)})-\mathbf{x}^{\star}(\boldsymbol{{\lambda}}^{(1)},\pi_{t}^{(1)})\right)\hspace{3cm}\quad\quad\quad\quad\medskip\nonumber \\
\leq\,\epsilon\cdot\left(\begin{array}{c}
\boldsymbol{{\lambda}}^{\star}(\boldsymbol{{\lambda}}^{(2)},\pi_{t}^{(2)})-\boldsymbol{{\lambda}}^{\star}(\boldsymbol{{\lambda}}^{(1)},\pi_{t}^{(1)})\\
\pi_{t}^{\star}(\boldsymbol{{\lambda}}^{(2)},\pi_{t}^{(2)})-\pi_{t}^{\star}(\boldsymbol{{\lambda}}^{(1)},\pi_{t}^{(1)})
\end{array}\right)^{T}\left(\begin{array}{c}
\boldsymbol{{\lambda}}^{(2)}-\boldsymbol{{\lambda}}^{(1)}\\
\pi_{t}^{(2)}-\pi_{t}^{(1)}
\end{array}\right)-\epsilon\cdot\left\Vert \left(\begin{array}{c}
\boldsymbol{{\lambda}}^{\star}(\boldsymbol{{\lambda}}^{(2)},\pi_{t}^{(2)})-\boldsymbol{{\lambda}}^{\star}(\boldsymbol{{\lambda}}^{(1)},\pi_{t}^{(1)})\\
\pi_{t}^{\star}(\boldsymbol{{\lambda}}^{(2)},\pi_{t}^{(2)})-\pi_{t}^{\star}(\boldsymbol{{\lambda}}^{(1)},\pi_{t}^{(1)})
\end{array}\right)\right\Vert _{2}^{2}.\label{eq:nonexpansive_price_map}
\end{align}
Applying the Cauchy\textendash{}Schwartz inequality and reorganizing
terms we obtain: 
\[
\begin{array}{l}
\left\Vert \left(\begin{array}{c}
\boldsymbol{{\lambda}}^{\star}(\boldsymbol{{\lambda}}^{(2)},\pi_{t}^{(2)})-\boldsymbol{{\lambda}}^{\star}(\boldsymbol{{\lambda}}^{(1)},\pi_{t}^{(1)})\\
\pi_{t}^{\star}(\boldsymbol{{\lambda}}^{(2)},\pi_{t}^{(2)})-\pi_{t}^{\star}(\boldsymbol{{\lambda}}^{(1)},\pi_{t}^{(1)})
\end{array}\right)\right\Vert _{2}^{2}\bigskip\\
\qquad\qquad\qquad\leq\left\Vert \left(\begin{array}{c}
\boldsymbol{{\lambda}}^{\star}(\boldsymbol{{\lambda}}^{(2)},\pi_{t}^{(2)})-\boldsymbol{{\lambda}}^{\star}(\boldsymbol{{\lambda}}^{(1)},\pi_{t}^{(1)})\\
\pi_{t}^{\star}(\boldsymbol{{\lambda}}^{(2)},\pi_{t}^{(2)})-\pi_{t}^{\star}(\boldsymbol{{\lambda}}^{(1)},\pi_{t}^{(1)})
\end{array}\right)\right\Vert _{2}\cdot\left\Vert \left(\begin{array}{c}
\boldsymbol{{\lambda}}^{(2)}-\boldsymbol{{\lambda}}^{(1)}\\
\pi_{t}^{(2)}-\pi_{t}^{(1)}
\end{array}\right)\right\Vert _{2}\medskip\\
\qquad\quad\qquad\qquad-\dfrac{{1}}{\epsilon}\left(\mathbf{x}^{\star}(\boldsymbol{{\lambda}}^{(2)},\pi_{t}^{(2)})-\mathbf{x}^{\star}(\boldsymbol{{\lambda}}^{(1)},\pi_{t}^{(1)})\right)^{T}\mathbf{A}\left(\mathbf{x}_{\delta},\,{\boldsymbol{{\lambda}}}_{\delta},\,{\pi}_{\delta}\right)\left(\mathbf{x}^{\star}(\boldsymbol{{\lambda}}^{(2)},\pi_{t}^{(2)})-\mathbf{x}^{\star}(\boldsymbol{{\lambda}}^{(1)},\pi_{t}^{(1)})\right)\bigskip\\
\qquad\qquad\qquad\leq\left\Vert \left(\begin{array}{c}
\boldsymbol{{\lambda}}^{\star}(\boldsymbol{{\lambda}}^{(2)},\pi_{t}^{(2)})-\boldsymbol{{\lambda}}^{\star}(\boldsymbol{{\lambda}}^{(1)},\pi_{t}^{(1)})\\
\pi_{t}^{\star}(\boldsymbol{{\lambda}}^{(2)},\pi_{t}^{(2)})-\pi_{t}^{\star}(\boldsymbol{{\lambda}}^{(1)},\pi_{t}^{(1)})
\end{array}\right)\right\Vert _{2}\cdot\left\Vert \left(\begin{array}{c}
\boldsymbol{{\lambda}}^{(2)}-\boldsymbol{{\lambda}}^{(1)}\\
\pi_{t}^{(2)}-\pi_{t}^{(1)}
\end{array}\right)\right\Vert _{2}
\end{array}
\]
where the last inequality follows from the positivity of the quadratic
form, due to the positive definiteness of $\mathbf{A}\left(\mathbf{x}_{\delta},\,{\boldsymbol{{\lambda}}}_{\delta},\,{\pi}_{\delta}\right)$;
which proves the desired nonexpansive property of the mapping $\left(\boldsymbol{{\lambda}}^{\star}(\cdot),\pi_{t}^{\star}(\cdot)\right)$:
\begin{equation}
\begin{array}{l}
\left\Vert \left(\begin{array}{c}
\boldsymbol{{\lambda}}^{\star}(\boldsymbol{{\lambda}}^{(2)},\pi_{t}^{(2)})-\boldsymbol{{\lambda}}^{\star}(\boldsymbol{{\lambda}}^{(1)},\pi_{t}^{(1)})\\
\pi_{t}^{\star}(\boldsymbol{{\lambda}}^{(2)},\pi_{t}^{(2)})-\pi_{t}^{\star}(\boldsymbol{{\lambda}}^{(1)},\pi_{t}^{(1)})
\end{array}\right)\right\Vert _{2}\leq\left\Vert \left(\begin{array}{c}
\boldsymbol{{\lambda}}^{(2)}-\boldsymbol{{\lambda}}^{(1)}\\
\pi_{t}^{(2)}-\pi_{t}^{(1)}
\end{array}\right)\right\Vert _{2}.\end{array}\label{eq:nonexpansive}
\end{equation}
\medskip{}

\noindent (b) Following similar steps as in (a) and using the Cartesian
structure of the set $\mathcal{Z}_{t}$ we deduce that, for any given
$(\boldsymbol{{\lambda}}^{(1)},\pi_{t}^{(1)}),\,(\boldsymbol{{\lambda}}^{(2)},\pi_{t}^{(2)})\in[0,\lambda^{\max}]^{Q}\times\mathcal{S}_{t}$,
there exists a tuple $\left(\mathbf{x}_{\eta},\,{\boldsymbol{{\lambda}}}_{\eta},\,{\pi}_{\eta}\right)$
lying on the segment joining $(\mathbf{x}^{\star}(\boldsymbol{{\lambda}}^{(1)},\pi_{t}^{(1)}),\boldsymbol{{\lambda}}^{\star}(\boldsymbol{{\lambda}}^{(1)},\pi_{t}^{(1)}),\pi_{t}^{\star}(\boldsymbol{{\lambda}}^{(1)},\pi_{t}^{(1)}))$
and $(\mathbf{x}^{\star}(\boldsymbol{{\lambda}}^{(2)},\pi_{t}^{(2)}),\boldsymbol{{\lambda}}^{\star}(\boldsymbol{{\lambda}}^{(2)},\pi_{t}^{(2)}),\pi_{t}^{\star}(\boldsymbol{{\lambda}}^{(2)},\pi_{t}^{(2)}))$
such that
\begin{equation}
\begin{array}{l}
\left(\mathbf{x}^{\star}(\boldsymbol{{\lambda}}^{(2)},\pi_{t}^{(2)})-\mathbf{x}^{\star}(\boldsymbol{{\lambda}}^{(1)},\pi_{t}^{(1)})\right)^{T}\mathbf{A}\left(\mathbf{x}_{\eta},\,{\boldsymbol{{\lambda}}}_{\eta},\,{\pi}_{\eta}\right)\left(\mathbf{x}^{\star}(\boldsymbol{{\lambda}}^{(2)},\pi_{t}^{(2)})-\mathbf{x}^{\star}(\boldsymbol{{\lambda}}^{(1)},\pi_{t}^{(1)})\right)\hspace{3cm}\quad\quad\quad\quad\medskip\\
\hspace{1cm}\leq\left(\mathbf{x}^{\star}(\boldsymbol{{\lambda}}^{(2)},\pi_{t}^{(2)})-\mathbf{x}^{\star}(\boldsymbol{{\lambda}}^{(1)},\pi_{t}^{(1)})\right)^{T}\,\left[\nabla_{\mathbf{x}}\mathbf{I}(\mathbf{x}_{\eta}),\nabla_{\mathbf{x}}I(\mathbf{x}_{\eta})\right]\,\left(\begin{array}{c}
\boldsymbol{{\lambda}}^{\star}(\boldsymbol{{\lambda}}^{(1)},\pi_{t}^{(1)})-\boldsymbol{{\lambda}}^{\star}(\boldsymbol{{\lambda}}^{(2)},\pi_{t}^{(2)})\\
\pi_{t}^{\star}(\boldsymbol{{\lambda}}^{(1)},\pi_{t}^{(1)})-\pi_{t}^{\star}(\boldsymbol{{\lambda}}^{(2)},\pi_{t}^{(2)})
\end{array}\right)\medskip\\
\hspace{1cm}\leq\left\Vert \mathbf{x}^{\star}(\boldsymbol{{\lambda}}^{(2)},\pi_{t}^{(2)})-\mathbf{x}^{\star}(\boldsymbol{{\lambda}}^{(1)},\pi_{t}^{(1)})\right\Vert _{2}\cdot\left\Vert \left[\nabla_{\mathbf{x}}\mathbf{I}(\mathbf{x}_{\eta}),\nabla_{\mathbf{x}}I(\mathbf{x}_{\eta})\right]\right\Vert \cdot\left\Vert \left(\begin{array}{c}
\boldsymbol{{\lambda}}^{(2)}-\boldsymbol{{\lambda}}^{(1)}\\
\pi_{t}^{(2)}-\pi_{t}^{(1)}
\end{array}\right)\right\Vert _{2},
\end{array}\label{eq:inequality_partial}
\end{equation}
where the last inequality follows from the Cauchy\textendash{}Schwartz
inequality and the nonexpansive property of $\left(\boldsymbol{{\lambda}}^{\star}(\cdot),\pi_{t}^{\star}(\cdot)\right)$
{[}cf. (\ref{eq:nonexpansive}){]}, and $\nabla_{\mathbf{x}}\mathbf{I}(\mathbf{x})\triangleq\nabla_{\mathbf{x}}[I_{1}(\mathbf{x}_{1}),\cdots,I_{Q}(\mathbf{x}_{Q})]$.
Invoking the uniform positive definiteness of $\mathbf{A}\left(\mathbf{x}_{\eta},\,{\boldsymbol{{\lambda}}}_{\eta},\,{\pi}_{\eta}\right)$
and the boundedness of the set $\mathcal{Y}$, we deduce from (\ref{eq:inequality_partial})
\begin{equation}
\left\Vert \mathbf{x}^{\star}(\boldsymbol{{\lambda}}^{(2)},\pi_{t}^{(2)})-\mathbf{x}^{\star}(\boldsymbol{{\lambda}}^{(1)},\pi_{t}^{(1)})\right\Vert \,\leq\dfrac{{\left\Vert \left[\nabla_{\mathbf{x}}\mathbf{I}(\mathbf{x}_{\eta}),\nabla_{\mathbf{x}}I(\mathbf{x}_{\eta})\right]\right\Vert }}{\lambda_{\text{{least}}}\left(\mathbf{A}\left(\mathbf{x}_{\eta},\,{\boldsymbol{{\lambda}}}_{\eta},\,{\pi}_{\eta}\right)\right)}\,\left\Vert \left(\begin{array}{c}
\boldsymbol{{\lambda}}^{(2)}-\boldsymbol{{\lambda}}^{(1)}\\
\pi_{t}^{(2)}-\pi_{t}^{(1)}
\end{array}\right)\right\Vert \leq\nu\cdot\left\Vert \left(\begin{array}{c}
\boldsymbol{{\lambda}}^{(2)}-\boldsymbol{{\lambda}}^{(1)}\\
\pi_{t}^{(2)}-\pi_{t}^{(1)}
\end{array}\right)\right\Vert \label{eq:Lipt_x_part}
\end{equation}
for some finite positive $\nu$, which proves the desired Lipschitz
continuity of $\mathbf{x}^{\star}(\cdot)$ on $[0,\lambda^{\max}]^{Q}\times\mathcal{S}_{t}$.\end{proof}\medskip

\noindent \textbf{Proof of Theorem \ref{ProxDecAlg_viaGT_conv_theo}.}
We are now ready to prove the theorem. The outer loop of Algorithm
\ref{ProxDecAlg_viaGT_conv_theo} {[}see (\ref{eq:price_multipliers_update_in_Algorithm_prox})
in Step 3{]} is an instance of the Jacobi Over Relaxation, JOR, method
\cite{Ortega-Rheinboldt_book} applied to the mapping $\left(\boldsymbol{{\lambda}}^{\star}(\cdot),\pi_{t}^{\star}(\cdot)\right)$;
which, using the notation introducing above, can be equivalently rewritten
as: 
\begin{equation}
\left(\begin{array}{c}
\boldsymbol{{\lambda}}^{(n+1)}\\
\pi_{t}^{(n+1)}
\end{array}\right)=(1-\epsilon)\cdot\left(\begin{array}{c}
\boldsymbol{{\lambda}}^{(n)}\\
\pi_{t}^{(n)}
\end{array}\right)+\epsilon\cdot\left(\begin{array}{c}
\boldsymbol{{\lambda}}^{\star}(\boldsymbol{{\lambda}}^{(n)},\pi_{t}^{(n)})\\
\pi_{t}^{\star}(\boldsymbol{{\lambda}}^{(n)},\pi_{t}^{(n)})
\end{array}\right).\label{eq:JOR}
\end{equation}
 Since $\left(\boldsymbol{{\lambda}}^{\star}(\cdot),\pi_{t}^{\star}(\cdot)\right)$
is nonexpansive on $[0,\lambda^{\max}]^{Q}\times\mathcal{S}_{t}$
{[}Proposition \ref{Prop_partial_regularization}(a){]}, the sequence
$\{(\boldsymbol{{\lambda}}^{(n)},\pi_{t}^{(n)})\}_{n=1}^{\infty}$
generated by the JOR scheme (\ref{eq:JOR}) converges to a fixed-point
$\left(\overline{\boldsymbol{{\lambda}}},\overline{\pi}_{t}\right)$
of $\left(\boldsymbol{{\lambda}}^{\star}(\cdot),\pi_{t}^{\star}(\cdot)\right)$
\cite[Th. 12.3.7]{Ortega-Rheinboldt_book}. By Proposition \ref{Prop_partial_regularization}(b)
{[}see (\ref{eq:Lip_cont}){]}, the convergence of $\{(\boldsymbol{{\lambda}}^{(n)},\pi_{t}^{(n)})\}_{n=1}^{\infty}$
implies also the convergence of the sequence $\{\mathbf{x}^{\star}(\boldsymbol{{\lambda}}^{(n)},\pi_{t}^{(n)})\}_{n=1}^{\infty}$
in the inner loop of Algorithm \ref{ProxDecAlg_viaGT_conv_theo} to
$\mathbf{x}^{\star}\left(\overline{\boldsymbol{{\lambda}}},\overline{\pi}_{t}\right)$;
the limit point $ $$\left(\mathbf{x}^{\star}\left(\overline{\boldsymbol{{\lambda}}},\overline{\pi}_{t}\right),\,\overline{\boldsymbol{{\lambda}}},\overline{\pi}_{t}\right)$
is the claimed NE of $\mathcal{G}_{t}(\mathcal{X},\,\boldsymbol{\theta})$
{[}Proposition \ref{Prop_partial_regularization}(c){]}, and thus
$\mathcal{G}(\mathcal{X},\,\boldsymbol{\theta})$, if $t>\lambda^{\max}$
(Theorem \ref{Theo_Existence-and-uniqueness_NE_G}). \hfill$\square$ 

\textcolor{black}{\small \bibliographystyle{IEEEtran}
\bibliography{scutari_refs}
}
\end{document}